\documentclass[twoside,11pt]{article}
\usepackage{jmlr2e}

\def \({\left(}
\def \){\right)}
\def \[{\left[}
\def \]{\right]}
\newcommand{\tbf}[1]{{\textbf{#1}}}

\newcommand{\defeq}{\vcentcolon=}

\newcommand{\bW}{{\textbf{W}}}
\newcommand{\bZ}{{\textbf{Z}}}

\newcommand{\bw}{{\textbf{w}}}

\newcommand{\bX}{{\textbf{X}}}
\newcommand{\bx}{{\textbf{x}}}

\newcommand{\bz}{{\textbf{z}}}

\newcommand{\bs}{{\textbf{s}}}
\newcommand{\bS}{{\textbf{S}}}

\newcommand{\la}{\langle}
\newcommand{\ra}{\rangle}

\newcommand{\be}{\begin{equation}}
\newcommand{\ee}{\end{equation}}
\newcommand{\bea}{\begin{eqnarray}}
\newcommand{\eea}{\end{eqnarray}}

\newcommand*{\xhat}[1]{#1\kern-0.50em\hat{\phantom{#1}}}

\usepackage{mathtools}
\usepackage{amsfonts}
\usepackage{amssymb}
\usepackage{enumitem}
\usepackage{algorithm}
\firstpageno{1}
\begin{document}
\title{Rank-one matrix estimation: analysis of algorithmic and information theoretic limits by the spatial coupling method}
\author{\name Jean Barbier$^{\dagger*}$, Mohamad Dia$^{\dagger}$ and Nicolas Macris$^{\dagger}$ \email firstname.lastname@epfl.ch \\
\addr $\dagger$ Laboratoire de Th\'eorie des Communications, Facult\'e Informatique et Communications,\\
Ecole Polytechnique F\'ed\'erale de Lausanne, CH-1015 Lausanne, Switzerland.\\
\addr $*$ International Center for Theoretical Physics, 11 Strada Costiera, I - 34151, Trieste, Italy.
\AND
\name Florent Krzakala \email florent.krzakala@ens.fr \\
\addr Laboratoire de Physique Statistique, CNRS, PSL Universit\'es et
Ecole Normale Sup\'erieure, \\Sorbonne Universit\'es et Universit\'e Pierre \& Marie Curie, 75005, Paris, France.
\AND
\name Lenka Zdeborov\'a \email lenka.zdeborova@gmail.com \\
\addr Institut de Physique Th\'éorique, CNRS, CEA, Universit\'e Paris-Saclay, \\ F-91191, Gif-sur-Yvette, France.
}
\editor{}

\maketitle
\begin{abstract}
  Factorizing low-rank matrices is a problem with many applications in
  machine learning and statistics, ranging from sparse PCA to
  community detection and sub-matrix localization. For probabilistic
  models in the Bayes optimal setting, general expressions for the
  mutual information have been proposed using powerful heuristic statistical
  physics computations via the replica and cavity methods,
  and proven in few specific cases by a variety of methods. Here, we use the spatial coupling methodology 
  developed in the framework of error correcting codes, to rigorously derive
  the mutual information for the symmetric rank-one case.  We characterize
  the detectability phase transitions in a large set of estimation
  problems, where we show that there exists a gap between what currently
  known polynomial algorithms (in particular spectral methods and
  approximate message-passing) can do and what is expected information
  theoretically. Moreover, we show that the computational gap vanishes for the proposed spatially coupled model, a promising feature with many possible applications. Our proof technique has an interest on its own and 
  exploits three essential ingredients: the interpolation
  method first introduced in statistical physics, the analysis of
  approximate message-passing algorithms first introduced in compressive sensing, and the theory of threshold
  saturation for spatially coupled systems first developed in coding theory. Our approach is very generic and can be
  applied to many other open problems in statistical estimation where
  heuristic statistical physics predictions are available.
\end{abstract}
\tableofcontents
\section{Introduction}
We consider the following probabilistic rank-one matrix estimation (or rank-one matrix factorization) problem: one has access to noisy observations
${\bf w} = (w_{ij})_{i,j=1}^n\in \mathbb{R}^{n,n}$ of the pair-wise product of the
components of a vector ${\bf s}=(s_i)_{i=1}^n \in \mathbb{R}^n$
where the components are i.i.d random variables distributed according to $S_i\sim P_0$,
$i=1,\dots, n$. The matrix elements of ${\bf w}$ are observed through a noisy
element-wise (possibly non-linear) output probabilistic channel
$P_{\rm out}(w_{ij}|s_is_j)$, with $i,j=1,\dots,n$. The goal 
is to estimate the vector ${\bf s}$ from ${\bf w}$, up to a global flip of sign in general, assuming that both distributions $P_0$ and $P_{\rm out}$ are 
known. We assume the noise to be symmetric so that $w_{ij}=w_{ji}$. There are many important problems
in statistics and machine learning that can be expressed in this way, among which:
\begin{itemize}
\item \emph{Sparse PCA}: Sparse principal component analysis (PCA) is a dimensionality reduction technique where one looks for a 
low-rank representation of a data matrix with sparsity constraints [\cite{zou2006sparse}]. The following is the simplest probabilistic 
symmetric version where one estimates a rank-one matrix. 
Consider a sparse random vector $\tbf{S}$, for instance
 drawn from a Gauss-Bernoulli distribution, and take an additive white Gaussian noise (AWGN) channel
 where the observations are 
  $W_{ij} = S_iS_j/\sqrt{n} + \Delta Z_{ij}$ whith $Z_{ij}\sim\mathcal{N}(0,1)$. Here\footnote{In this paper 
 $\mathcal{N}(x|m, \sigma^2) = (2\pi\sigma^2)^{-1/2}\exp(-(x-m)^2/2\sigma^2))$} 
  $P_{\rm out}(w_{ij}|s_is_j) = \mathcal{N}(w_{ij}|s_is_j/\sqrt{n}, \Delta)$.
\item \emph{Spiked Wigner model}: In this model the noise is still Gaussian, but the vector $\tbf{S}$ is assumed to be a Bernoulli random
  vector with i.i.d components $S_i \sim {\rm Ber}(\rho)$. This formulation is a particular
  case of the spiked covariance model in statistics introduced
  by [\cite{johnstone2004sparse,johnstone2012consistency}]. It has also
  attracted a lot of attention in the framework of random matrix theory
   (see for instance [\cite{baik2005phase}] and references therein).
\item \emph{Community detection}: In its simplest setting, one uses a Rademacher
  vector $\tbf{S}$ where each variable take values $S_i\in \{-1,1\}$
  depending on the ``community'' it belongs to. The observation model
  then introduces missing information and errors such that, for
  instance, $P_{\rm out}(w_{ij}|s_is_j)=p_1 \delta(w_{ij} - s_is_j) + p_2 \delta(w_{ij} +s_is_j) + (1-p_1-p_2) \delta(w_{ij})$, 
  where $\delta(\cdot)$ is the Delta dirac function. 
  These models have recently attracted a lot of attention both in statistics and machine learning contexts
  (see e.g. [\cite{bickel2009nonparametric,decelle2011asymptotic,karrer2011stochastic,saade2014spectral,massoulie2014community,ricci2016performance}]).
\item \emph{Sub-matrix localization}: This is the problem of finding a
  submatrix with an elevated mean in a large noisy matrix, as in
  [\cite{hajek2015submatrix,chen2014statistical}].
\item \emph{Matrix completion}: A last example is the matrix completion
  problem where a part of the information (the matrix elements) is hidden, while the rest is
  given with noise. For instance, a classical model is $P_{\rm out}(w_{ij}|s_is_j) = p\delta(w_{ij}) + (1-p) \mathcal{N}(w_{ij}|s_is_j,\Delta)$. Such problems have been extensively discussed over the last decades, in particular because of their connection to collaborative filtering (see for instance [\cite{candes2009exact,cai2010singular,keshavan2009matrix,saade2015matrix}]).
\end{itemize}

Here we shall consider the probabilistic formulation of these
problems and focus on estimation in the mean square
error (MSE) sense. We rigorously derive an 
explicit formula for the mutual information 
in the asymptotic limit, and for the
information theoretic minimal mean square error (MMSE).
Our results imply that in a large region of parameters, the
posterior expectation of the underlying signal, a quantity often
assumed intractable to compute, can be obtained using a
polynomial-time scheme via the approximate message-passing (AMP) 
framework [\cite{rangan2012iterative,matsushita2013low,6875223,deshpande2015asymptotic,lesieur2015phase}]. 
We also demonstrate the existence of a region where no
{\it known} tractable algorithm is able to find a solution
correlated with the ground truth. Nevertheless, we prove 
explicitly that it is information theoretically possible to do so (even in this region), and 
discuss the implications in terms of computational complexity.

The crux of our analysis rests on an "auxiliary" spatially coupled (SC) system. The hallmark of SC models is that one can tune them so that 
the gap between the algorithmic and information theoretic limits is eliminated, while at the same time the mutual information is maintained unchanged for the coupled and original models. Roughly speaking, this means that it is possible to algorithmically compute 
the information theoretic limit of the original model because a
suitable algorithm is optimal on the coupled system. 

Our proof technique has an interest by its own as it combines recent rigorous results in coding theory
along the study of capacity-achieving SC
codes [\cite{hassani2010coupled,kudekar2011threshold,yedla2014maxwell,
giurgiu2016,BDM_trans2017,DiaThesis}]
with other progress coming from developments in mathematical
physics of spin glass
theory [\cite{guerra2005introduction}].
Moreover, our proof exploits the ``threshold saturation'' phenomenon of the AMP algorithm and uses spatial coupling as a proof technique. From this point of view, we believe that the
theorem proven in this paper is relevant in a broader context going
beyond low-rank matrix estimation and can be applied for a wide range of inference problems where message-passing algorithm and spatial coupling can be applied. Furthermore, our work provides important results on the exact formula for the MMSE and on the optimality of the AMP algorithm. 

Hundreds of papers have been
published in statistics, machine learning or information theory using
the non-rigorous statistical physics approach. We believe that our
result helps setting a rigorous foundation of a broad line of
work. While we focus on rank-one symmetric matrix estimation, our
proof technique is readily extendable to more generic low-rank
symmetric matrix or low-rank symmetric tensor estimation. We also
believe that it can be extended to other problems of interest in
machine learning and signal processing. It has already been extended to linear estimation and compressed sensing [\cite{BDMK_alerton2016,BMDK_trans2017}].


We conclude this introduction by giving a few pointers to the recent literature  
on rigorous results.
For rank-one symmetric matrix estimation problems, AMP has been
introduced by [\cite{rangan2012iterative}], who also computed the
state evolution formula to analyze its performance, generalizing
techniques developed by [\cite{BayatiMontanari10}] and
[\cite{Montanari-Javanmard}]. State evolution was further studied by
[\cite{6875223}] and [\cite{deshpande2015asymptotic}]. In
[\cite{lesieur2015mmse,lesieur2015phase}], the generalization to
larger rank was also considered.
The mutual information was already computed in the special case when
$S_i\!=\!\pm 1\!\sim\! {\rm Ber}(1/2)$ in
[\cite{KoradaMacris}] where an equivalent spin glass model was analyzed.
The results of [\cite{KoradaMacris}] were first generalized in [\cite{krzakala2016mutual}] 
who, notably, obtained a generic matching upper
bound. The same formula was also rigorously computed following
the study of AMP in [\cite{6875223}] for spike models (provided,
however, that the signal was not {\it too} sparse) and in
[\cite{deshpande2015asymptotic}] for strictly symmetric community
detection. The general formula proposed by [\cite{lesieur2015mmse}] for the
conditional entropy and the MMSE on the basis of the heuristic cavity
method from statistical physics was first demonstrated in full
generality by the current authors in [\cite{barbierNIPS2016}]. 
This paper represents an extended version of [\cite{barbierNIPS2016}] that includes all the proofs and derivations along with more detailed discussions.
All preexisting proofs could not reach the more
interesting regime where a gap between the algorithmic and information
theoretic performances appears (i.e. in the presence of ``first order'' phase transition), leaving a gap with the statistical physics conjectured formula. Following the work of [\cite{barbierNIPS2016}], the replica formula for rank-one symmetric matrix estimation has been proven again several times using totally different techniques that involve the concentration's proof of the overlaps [\cite{lelarge2017,BM_stochastic2017}]. Our proof strategy does not require any concentration and it uses AMP and spatial coupling as proof techniques. Hence, our result has more practical implications in terms of proving the range of optimality of the AMP algorithm for both the underlying (uncoupled) and spatially coupled models.

This paper is organized as follows: the problem statement and the main results are given in Section \ref{sec:settings_results} along with a sketch of the proof, two applications for symmetric rank-one matrix estimation are presented in Section \ref{sec:examples}, the threshold saturation phenomenon and the relation between the underlying and spatially coupled models are proven in Section \ref{sec:thresh-sat} and Section \ref{sec:subadditivitystyle} respectively, the proof of the main results follows in Section \ref{sec:proof_theorem} and Section \ref{sec:proof_coro}.

A word about notations: in this paper, we use capital letters for random variables, and small 
letters for fixed realizations. Matrices and vectors are bold while scalars are not. 
Components of vectors or matrices are identified by the presence of lower indices.

\section{Setting and main results}\label{sec:settings_results}

\subsection{Basic underlying model}

A standard and natural setting is to consider an additive white gausian noise (AWGN) channel with variance $\Delta$ assumed to be known. The model reads
\begin{align} \label{eq:mainProblem}
w_{ij} = \frac{s_i s_j}{\sqrt{n}} + \sqrt{\Delta} z_{ij},
\end{align}
where ${\bf z} = (z_{ij})_{i,j=1}^n$ is a symmetric matrix with
$Z_{ij} \sim \mathcal{N}(0,1)$, $1\leq i\leq j\leq n$, and ${\bf s} = (s_i)_{i=1}^n$ has i.i.d components
$S_i\sim P_0$. We set $\mathbb{E}[S^2] = v$. Precise hypothesis on $P_0$ are given later.

Perhaps surprisingly, it turns out that the study of this Gaussian setting is 
sufficient to completely characterize all the problems discussed in the introduction, even 
if we are dealing with more complicated (noisy) observation models. This 
is made possible by a theorem of channel universality. Essentially, the theorem states that for any output channel $P_{\rm out}(w|y)$
such that at $y= 0$ the function $y\mapsto \log P_{\rm out}(w|y)$ is three times differentiable
with bounded second and third derivatives, then the
mutual information satisfies
\begin{equation}
  I(\tbf{S};\tbf{W}) = I (\tbf{S}\tbf{S}^{\intercal}; \tbf{S}\tbf{S}^{\intercal} + \sqrt{\Delta} \; \tbf{Z}) + O(\sqrt{n} ), 
\end{equation}
where $\Delta$ is the
inverse Fisher information (evaluated at $y = 0$) of the output
channel
\begin{equation}
  \Delta^{-1} \defeq \int dw P_{\rm out}(w|0) \left( \frac{\partial \log P_{\rm out}(w|y)}{\partial y}\Big|_{y=0}\right)^2\, . 
  \label{FisherInformation}
\end{equation}
This means that the mutual information per variable $ I(\tbf{S};\tbf{W})/n$ is asymptotically equal the mutual information per variable of an AWGN channel. Informally, it implies 
that we only have to compute the mutual
information for an ``effective'' Gaussian channel to take care of a wide range of
problems. The statement was conjectured in [\cite{lesieur2015mmse}] and can be proven by an 
application of the Lindeberg principle [\cite{deshpande2015asymptotic}], [\cite{krzakala2016mutual}].

\subsection{AMP algorithm and state evolution}\label{subsec:AMP_SE}

AMP has been applied for the rank-one symmetric matrix estimation problems by [\cite{rangan2012iterative}], who also computed the
state evolution formula to analyze its performance, generalizing
techniques developed by [\cite{BayatiMontanari10}] and
[\cite{Montanari-Javanmard}]. State evolution was further studied by
[\cite{6875223}] and [\cite{deshpande2015asymptotic}]. AMP is an iterative algorithm that provides an estimate $\hat{\bf s}^{(t)}(\bw)$, at each 
iteration $t \in \mathbb{N}$, of the vector $\bf s$. It turns out that tracking the asymptotic vector and matrix MSE of the AMP algorithm
is equivalent to running a simple recursion called {\it state evolution} (SE).  

The AMP algorithm reads
\begin{align}
\begin{cases}
\hat{s}_j^{(t)} = \eta_t((\bw \hat{\bs}^{(t-1)})_j - b^{(t-1)}\hat{\bs}^{(t-2)}_j), 
\\
b^{(t)}= \frac{1}{n}\sum_{i=1}^n \eta_t^\prime((\bw \hat{\bs}^{(t-1)})_i - b^{(t-1)}\hat{s}^{(t-2)}_i)
\end{cases}
\end{align}
for $j=1, \cdots, n$, where $\eta_t(y)$ is called the denoiser and $\eta_t^\prime(y)$ is the derivative w.r.t $y$.
The {\it denoiser} is the MMSE estimate associated to an ``equivalent scalar denoising problem''
\begin{align}\label{eq:scalarDen}
y = s + \Sigma(E) z,\qquad \Sigma(E)^{-2} \defeq \frac{v-E}{\Delta}\,.
\end{align}
with
$Z\sim \mathcal{N}(0,1)$ and
\begin{align}
 \eta(y) = \mathbb{E}[X|Y=y] = \frac{\int dx x P_0(x)e^{-\frac{(x - Y)^2}{2\Sigma(E)^{2}}}}{\int dx P_0(x)
e^{-\frac{(x - Y)^2}{2\Sigma(E)^{2}}}}\, ,
\end{align}
where $E$ is updated at each time instance $t$ according to the recursion \eqref{eq:state-evolution}.

Natural performance measures are the ``vector'' and ``matrix'' MSE's of the AMP estimator defined below.
\begin{definition}[Vector and matrix MSE of AMP]\label{def:mse_amp}
The vector and matrix MSE of the AMP estimator $\xhat{\bS}^{(t)}(\bW)$ at iteration $t$ are defined respectively  as follows
\begin{align}
{\rm Vmse}_{n, \rm AMP}^{(t)}(\Delta^{-1}) \defeq \frac{1}{n}\mathbb{E}_{\tbf{S}, \bW}\bigl[\| \xhat{\tbf{S}}^{(t)} - \tbf{S}\|_{2}^2\bigr],
\end{align}
\begin{align}
{\rm Mmse}_{n, \rm AMP}^{(t)}(\Delta^{-1}) \defeq \frac{1}{n^2}\mathbb{E}_{\tbf{S}, \bW}\bigl[\|\xhat{\tbf{S}}^{(t)} \xhat{\tbf{S}}^{{(t)}^\intercal} - \tbf{S} \tbf{S}^\intercal\|_{\rm F}^2\bigr],
\end{align}
where $\|A\|_{\rm F}^2 = \sum_{i,j} A_{ij}^2$ stands for the Frobenius norm of a matrix $A$.
\end{definition}

A remarkable fact that follows from a general theorem of [\cite{BayatiMontanari10}] (see [\cite{deshpande2015asymptotic}] for its use in the matrix case) is that the state evolution sequence tracks these two MSE's and thus allows to assess the performance of AMP. 
Consider the scalar denoising problem \eqref{eq:scalarDen}.
Hence, the (scalar) mmse function associated to this problem reads
\begin{align}\label{eq:SE_post}
 {\rm mmse}(\Sigma(E)^{-2}) \defeq \mathbb{E}_{S, Y}[(S - \mathbb{E}[X \vert Y])^2]\, .
\end{align}
The {\it state evolution sequence} $E^{(t)}$, $t\in \mathbb{N}$ is defined as 
\begin{align}\label{eq:state-evolution}
 E^{(t+1)} = {\rm mmse}(\Sigma(E^{(t)})^{-2}), \qquad E^{(0)} = v.
\end{align}
Since the ${\rm mmse}$ function is monotone decreasing (its argument has the dimension of a signal to noise ratio) it is easy to see that that $E^{(t)}$ is a decreasing non-negative sequence. Thus $\lim_{t\to +\infty} E^{(t)} \defeq E^{(\infty)}$ exists. 
One of the basic results of [\cite{BayatiMontanari10}], [\cite{deshpande2015asymptotic}] is 
\begin{align}\label{eq:mseampmatriciel}
\lim_{n\to +\infty} {\rm Vmse}_{n, \rm AMP}^{(t)}(\Delta^{-1}) = E^{(t)}, \quad {\rm and} \quad
\lim_{n\to +\infty} {\rm Mmse}_{n, \rm AMP}^{(t)}(\Delta^{-1}) = v^2 - (v- E^{(t)})^2.
\end{align}
We note that 
the results in [\cite{BayatiMontanari10}], [\cite{deshpande2015asymptotic}] are stronger in 
the sense that the non-averaged algorithmic mean square errors are tracked by state evolution with probability one.

Note that when $\mathbb{E}[S]=0$ then $v$ is an unstable fixed point, and as such, state evolution ``does not start'', 
in other words we have $E^{(t)} =v$. While this is not really a problem when one runs AMP in practice, for analysis 
purposes one can circumvent this problem by slightly biasing $P_0$ and remove the bias at the end of the analysis. For simplicity, we always assume that $P_0$ is biased so that $\mathbb{E}[S]$ is not zero.

{\bf Assumption 1:} 
In this work we assume that $P_0$ is discrete with bounded support. Moreover, we assume that $P_0$ is biased such that 
$\mathbb{E}[S]$ is non-zero.

A fundamental quantity computed by state evolution is the algorithmic threshold.
\begin{definition}[AMP threshold]\label{amp-thresh}
For $\Delta>0$ small enough, the fixed point equation corresponding to \eqref{eq:state-evolution} has a unique solution for all noise values in $]0, \Delta[$. We define $\Delta_{\rm AMP}$ as the supremum of all such $\Delta$.
\end{definition}

\subsection{Spatially coupled model}\label{subsec:SC}

The present spatially coupled construction is similar to the
one used for the coupled Curie-Weiss model [\cite{hassani2010coupled}] and is also similar to mean field spin glass systems 
introduced in [\cite{FranzToninelli,SilvioFranz}]. We consider a chain (or a ring) of underlying systems positioned at
$\mu\!\in\! \{0,\dots, L\}$ and coupled to neighboring blocks
$\{\mu\! -\! w, \dots, \mu\! +\!w\}$. Positions $\mu$ are taken modulo $L\!+\!1$ and the integer $w\!\in\!\{0,\ldots,L/2\}$ equals the size of the {\it coupling window}. The coupled model is
\begin{align} \label{eq:def_coupledSyst}
w_{i_\mu j_\nu} = s_{i_\mu} s_{j_\nu}\sqrt{\frac{\Lambda_{\mu\nu}}{n}} + z_{i_\mu j_\nu}\sqrt{\Delta} ,
\end{align}
where the index $i_\mu \!\in\! \{1, \dots, n\}$ (resp. $j_\nu$) belongs to the block $\mu$ (resp. $\nu$) along the ring, $\mathbf{\Lambda}$ is an $(L\!+\!1)\!\times\! (L\!+\!1)$ matrix which describes the strength of the coupling between blocks, and $Z_{i_\mu j_\nu}\!\sim \!\mathcal{N}(0,1)$ are i.i.d. For the analysis to work, the matrix elements have to be chosen appropriately. We assume that: 
\begin{itemize}
\item[i)] $\mathbf{\Lambda}$ is a doubly stochastic matrix; 
\item[ii)] $\Lambda_{\mu\nu}$ depends on
$\vert \mu \!- \!\nu\vert$; 
\item[iii)] $\Lambda_{\mu\nu}$ is not vanishing for $\vert \mu \!-\!\nu| \leq w$ and vanishes for $\vert \mu \!-\!\nu|\! > \!w$; 
\item[iv)] $\mathbf{\Lambda}$ is {\it smooth} in the sense $\vert \Lambda_{\mu\nu}\! -\! \Lambda_{\mu+1 \nu} \vert\! =\! \mathcal{O}(w^{-2})$ and $\Lambda^{*}\defeq\sup_{\mu, \nu}\Lambda_{\mu\nu} = \mathcal{O}(w^{-1})$; 
\item[v)] $\mathbf{\Lambda}$ has a non-negative Fourier transform. 
\end{itemize}
All these conditions can easily be met, the simplest example being a triangle of base $2w\!+\!1$ and height $1/(w\!+\!1)$, more precisely:
\begin{align}\label{triangle}
\Lambda_{\mu\nu} = 
\begin{cases}
\frac{1}{w+1}\biggl(1 - \frac{\vert\mu - \nu\vert}{w+1}\biggr), \qquad \vert\mu-\nu\vert \leq w\\
0, \qquad \vert\mu-\nu\vert > w
\end{cases}
\end{align}
We will always denote by $\mathcal{S}_\mu \defeq \{\nu \,|\, \Lambda_{\mu\nu} \neq 0 \}$ the set of $2w+1$ blocks coupled to block $\mu$.

The construction of the coupled system is completed by introducing a {\it seed} in the ring: we assume 
perfect knowledge of the signal components $\{s_{i_\mu}\}$ for $\mu\! \in \!\mathcal{B}\!\defeq \!\{-w\!-\!1, \dots, w\!-\!1\} \mod L\!+\!1$. 
This seed is what allows to close the gap between the algorithmic and information theoretic limits and therefore plays a crucial role. 
We sometimes refer to the seed as the {\it pinning construction}. 
Note that the seed can also be viewed as an ``opening'' of the chain with fixed boundary conditions.   

AMP has been applied for the rank-one symmetric matrix estimation problems by [\cite{rangan2012iterative}], who also computed the
state evolution formula to analyze its performance, generalizing
techniques developed by [\cite{BayatiMontanari10}] and
[\cite{Montanari-Javanmard}]. State evolution was further studied by
[\cite{6875223}] and [\cite{deshpande2015asymptotic}].

The AMP algorithm and the state evolution recursion [\cite{6875223,deshpande2015asymptotic}] can be easily adapted to the spatially coupled model as done in Section \ref{sec:thresh-sat}. The proof that the state evolution for the symmetric rank-one matrix estimation problem tracks the AMP on a spatially coupled model is an extension of the analysis done in [\cite{6875223,deshpande2015asymptotic}] for the uncoupled model.
The full re-derivation of such result would be lengthy and beyond the scope of our analysis. We thus assume that state evolution tracks the AMP performance for our coupled problem. However, we believe that the proof will be similar to the one done for the spatially coupled compressed sensing problem [\cite{Montanari-Javanmard}]. This assumption is vindicated numerically.

{\bf Assumption 2:} We consider the spatially coupled model \eqref{eq:def_coupledSyst} with $P_0$ satisfying Assumption 1. We assume that state evolution tracks the AMP algorithm for this model.

\subsection{Main results: basic underlying model}\label{subsec:results}

One of our central results is a proof of the expression for the asymptotic mutual information per variable via the so-called {\em replica symmetric (RS) potential}. This is the function $E\in [0, v] \mapsto i_{\rm RS}(E; \Delta)\in \mathbb{R}$ defined as
\begin{align}\label{eq:potentialfunction}
i_{\rm RS}(E;\Delta) \defeq \frac{(v-E)^2+v^2}{4\Delta} - \mathbb{E}_{S, Z}\biggl[\ln\biggl(\int dx\,P_0(x)
e^{-\frac{x^2}{2\Sigma(E)^{2}} 
+ x\bigl(\frac{S}{\Sigma(E)^{2}} + \frac{Z}{\Sigma(E)}\bigr)}\biggr)\biggr] \, ,
\end{align}
with $Z\!\sim\! \mathcal{N}(0,1)$, $S\!\sim\! P_0$. 
Most of our results will assume that $P_0$ is a discrete distribution over 
a finite bounded real alphabet $P_0(s) \!=\! \sum_{\alpha=1}^\nu p_\alpha \delta(s\! - \! a_\alpha)$ (see Assumption 1). 
Thus the only continuous integral in \eqref{eq:potentialfunction} is the Gaussian over $Z$. The extension to mixtures of continuous and discrete signals 
can be obtained by approximation methods not discussed in this paper (see e.g. the methods in [\cite{lelarge2017}]).

It turns out that both the information theoretic and algorithmic AMP thresholds are determined by the set of stationary points of \eqref{eq:potentialfunction} (w.r.t $E$).
It is possible to show that for all $\Delta > 0$ there always exist at 
least one stationary minimum.\footnote{
Note $E\!=\!0$ is never a stationary 
point (except for the trivial case of $P_0$ a single Dirac mass which we exclude from the discussion) and $E = v$ is stationary only if $\mathbb{E}[S] = 0$.} 
In this contribution 
we suppose that at most three stationary points exist, corresponding to situations with at most one phase transition as depicted in  Fig. \ref{fig:I-AMP} (see Assumption 3 below). 
Situations 
with multiple transitions could also be covered by our techniques.

{\bf Assumption 3:} 
We assume that $P_0$ is such that there exist at most three stationary points for the potential \eqref{eq:potentialfunction}.

To summarize, our main assumptions in this paper are:

\begin{enumerate}[label=(A\arabic*),noitemsep]
	\item The prior $P_0$ is discrete with bounded support. Moreover, we assume that $P_0$ is biased such that $\mathbb{E}[S]$ is non-zero. \label{ass:h1}
	\item We consider the spatially coupled model \eqref{eq:def_coupledSyst} with $P_0$ satisfying Assumption \ref{ass:h1}. We assume that state evolution tracks the AMP algorithm for this model. \label{ass:h2}
	\item  We assume that $P_0$ is such that there exist at most three stationary points for the potential \eqref{eq:potentialfunction}. \label{ass:h3}
\end{enumerate}

\begin{remark}
An important property of the replica symmetric potential is that the stationary points satisfy the state evolution fixed point equation. In other words 
$\partial i_{\rm RS}/\partial E = 0$ implies $E= {\rm mmse}(\Sigma(E)^{-2})$ and conversely. Moreover it is not difficult to see that the 
$\Delta_{\rm AMP}$ is given by the smallest solution of 
$\partial i_{\rm RS}/\partial E\!=\! \partial^2 i_{\rm RS}/\partial E^2 \!=\! 0$.
In other words the AMP threshold is the ``first'' horizontal inflexion point appearing in $i_{\rm RS}(E;\Delta)$ when $\Delta$ increases
from $0$ to $+\infty$.
\end{remark}

One of the main results of this paper is formulated in the following theorem which provides a proof of the conjectured single-letter formula for the asymptotic mutual information per variable.
%
\begin{theorem}[RS formula for the mutual information]\label{thm1}
Fix $\Delta > 0$ and let $P_0$ satisfy Assumptions \ref{ass:h1}-\ref{ass:h3}. Then 
\begin{align}\label{rsform}
 \lim_{n\to \infty} \frac{1}{n}I(\tbf{S} ; \tbf{W}) = \min_{E\in [0, v]}i_{\rm RS}(E;\Delta).
\end{align}
\end{theorem}
\begin{proof}
See Section \ref{sec:proof_theorem}.
\end{proof}
The proof of the {\it existence of the limit} does not require the above hypothesis on $P_0$. Also, it was first shown
in [\cite{krzakala2016mutual}] that  
\begin{align}\label{eq:guerrabound}
\limsup_{n\to +\infty} \frac{1}{n}I(\tbf{S}; \tbf{W}) \leq \min_{E\in [0, v]}i_{\rm RS}(E;\Delta),
\end{align}
an inequality that {\it we will use} in the proof section. 
Note that, interestingly, and perhaps surprisingly, the analysis of [\cite{krzakala2016mutual}] leads to a sharp upper bound on the ``free energy'' for all finite $n$. We will make extensive use of this inequality
and for sake of completeness, we summarize its proof in Appendix~\ref{app:upperBound}. 


Theorem \ref{thm1} allows to compute the information theoretic phase transition threshold which we define in the following way.
\begin{definition}[Information theoretic or optimal threshold]\label{def-delta-opt} 
Define $\Delta_{\rm opt}$ as the first non-analyticity point of the mutual information as $\Delta$ 
increases. Formally
\begin{align}
 \Delta_{\rm opt} \defeq \sup\{\Delta\vert\lim_{n\to \infty}\frac{1}{n}I(\tbf{S}; \tbf{W}) \ \text{is analytic in} \ ]0, \Delta[\}.
\end{align}
\end{definition}
The information theoretic threshold is also called ``optimal threshold" because we expect $\Delta_{\rm AMP} \leq \Delta_{\rm opt}$. This is indeed proven in Lemma \ref{suboptimality}. 

When $P_0$ is s.t \eqref{eq:potentialfunction} has at most three stationary points, then $\min_{E\in [0, v]} i_{\rm RS}(E;\Delta)$ has at most one {\it non-analyticity} point 
denoted $\Delta_{\rm RS}$ (see Fig. \ref{fig:I-AMP}). In case of analyticity over all $\mathbb{R}_+$, we set $\Delta_{\rm RS} \!=\!\infty$. We call $\Delta_{\rm RS}$ the RS or potential threshold. 
Theorem~\ref{thm1} gives us a mean to concretely {\it compute} the information theoretic threshold: $\Delta_{\rm opt} \!=\! \Delta_{\rm RS}$.

\begin{figure}[!t]
\centering
\includegraphics[width=1\textwidth]{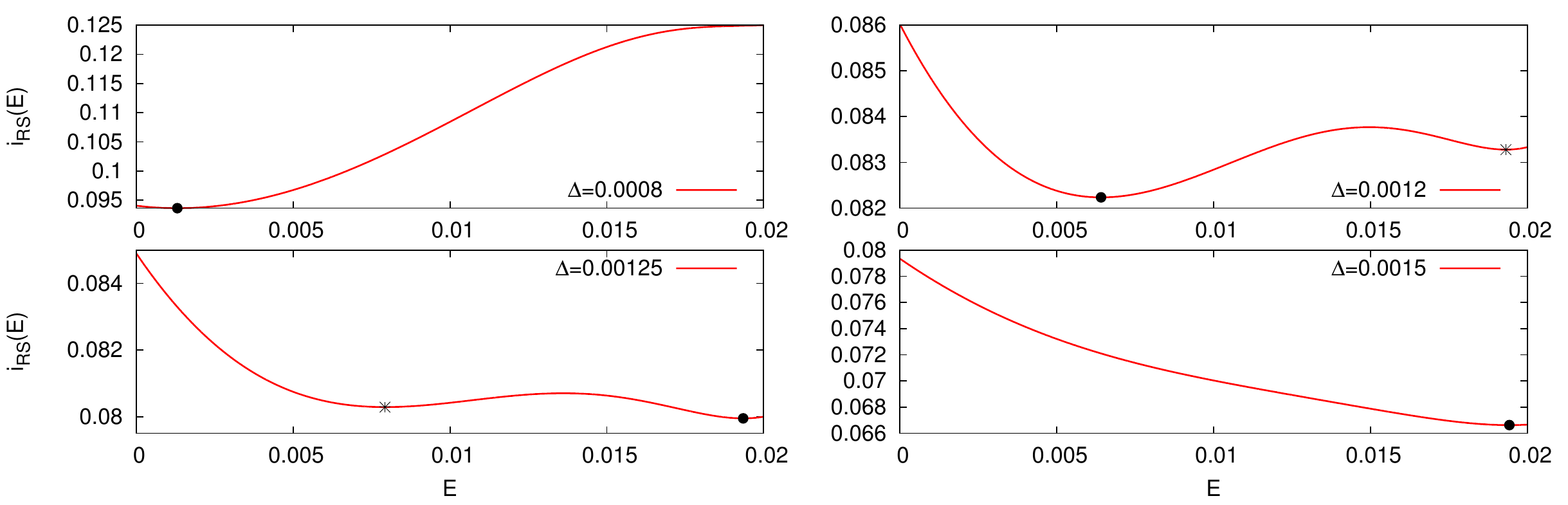}
\vspace*{-18pt}
\caption{The replica symmetric potential $i_{\rm RS}(E)$ for four values of $\Delta$ in the spiked Wigner model with $S_i \sim {\rm Ber}(\rho)$. The normalized mutual information is $\min i_{\rm RS}(E)$ (the black dot, while the black cross corresponds to the local minimum) and the asymptotic MMSE is ${\rm argmin} i_{\rm RS}(E)$, where $v\!=\!\rho$ in this case with $\rho\!=\!0.02$. From top left to bottom right: $i)$ For low noise values, here $\Delta\!=\!0.0008\!<\!\Delta_{\rm AMP}$, there exists a unique ``good'' minimum corresponding to the MMSE and AMP is Bayes optimal. $ii)$ As the noise increases, a second local ``bad'' minimum appears: this is the situation at $\Delta_{\rm AMP}\!<\!\Delta\!=\!0.0012\!<\!\Delta_{\rm RS}$. $iii)$ For $\Delta\!=\!0.00125\!>\!\Delta_{\rm RS}$, the ``bad'' minimum becomes the global one and the MMSE suddenly deteriorates. $iv)$ For larger values of $\Delta$, only the ``bad'' minimum exists. AMP can be seen as a naive minimizer of this curve starting from $E\!=\!v\!=\!0.02$. AMP can reach the global minimum in situations $i)$, $iii)$ and $iv)$. However, in situation $ii)$, when $\Delta_{\rm AMP}\!<\!\Delta\!<\!\Delta_{\rm RS}$, AMP is trapped by the local minimum with large MSE instead of reaching the global one corresponding to the MMSE.}\label{fig:I-AMP}
\end{figure}

From Theorem  \ref{thm1} we will also deduce the expressions for the \emph{vector} MMSE and the \emph{matrix} MMSE defined below.
\begin{definition}[Vector and matrix MMSE]\label{def:mmse}
The vector and matrix MMSE are defined respectively as follows
\begin{align}\label{eq:Vmmse}
 {\rm Vmmse}_n(\Delta^{-1}) \defeq \frac{1}{n}\mathbb{E}_{\tbf{S}, \tbf{W}}\Big[ \bigl\| \tbf{S} - 
 \mathbb{E}[\tbf{X}\vert\tbf{W}]\bigr\|_{2}^2 \Big].
\end{align}
\begin{align}\label{eq:Mmmse}
 {\rm Mmmse}_n(\Delta^{-1}) \defeq \frac{1}{n^2}\mathbb{E}_{\tbf{S}, \tbf{W}}\Big[ \bigl\| \tbf{S}\tbf{S}^{\intercal} - 
 \mathbb{E}[\tbf{X}\tbf{X}^{\intercal}\vert\tbf{W}]\bigr\|_{\rm F}^2 \Big].
\end{align}
\end{definition}
The conditional expectation $\mathbb{E}[\,\cdot\,\vert\tbf{W}]$ in Definition \ref{def:mmse} is w.r.t the posterior distribution
\begin{align}\label{eq:posterior_dist}
 P(\tbf{x}\vert \tbf{w}) = \frac{1}{\tilde{\mathcal{Z}}(\tbf{w})}e^{-\frac{1}{2\Delta}\sum_{i\leq j}\big(\frac{x_i x_j}{\sqrt n} - w_{ij}\big)^2}\prod_{i=1}^n P_0(x_i),
\end{align}
with the normalizing factor depending on the observation given by 
\begin{align}
\tilde{\mathcal{Z}}(\tbf{w}) = \int \big\{\prod_{i=1}^n dx_i P_0(x_i)\big\} e^{-\frac{1}{2\Delta}\sum_{i\leq j}\big(\frac{x_i x_j}{\sqrt n} - w_{ij}\big)^2}
\end{align}
The expectation $\mathbb{E}_{\tbf{S}, \tbf{W}}[\,\cdot\,]$ is the one w.r.t $P({\bf w})P({\bf s})=\tilde{\mathcal{Z}}(\tbf{w})\prod_{i=1}^{n} P_0(s_i)$. The expressions for the MMSE's in terms of \eqref{eq:potentialfunction} are given in the following corollary.
%
%
\begin{corollary}[Exact formula for the MMSE]\label{cor:MMSE}
For all $\Delta\neq \Delta_{\rm RS}$, the 
matrix MMSE  
is asymptotically 
\begin{align}\label{coro1}
\lim_{n\to\infty} {\rm Mmmse}_n(\Delta^{-1}) = v^2-(v-{\rm argmin}_{E\in[0,v]} i_{\rm RS}(E;\Delta))^2.
\end{align}
Moreover, if $\Delta \!<\!\Delta_{\rm AMP}$ or $\Delta\! >\! \Delta_{\rm RS}$, then 
the usual vector MMSE
satisfies 
\begin{align}\label{coro2}
\lim_{n\to \infty}{\rm Vmmse}_n(\Delta^{-1}) = {\rm argmin}_{E\in[0,v]} i_{\rm RS}(E;\Delta).
\end{align}
\end{corollary}
\begin{proof}
See Section \ref{sec:proof_coro}.
\end{proof}
It is natural to conjecture that the vector MMSE is given by ${\rm argmin}_{E\in[0,v]} i_{\rm RS}(E;\Delta)$ in the {\it whole} range $\Delta \neq \Delta_{\rm RS}$, but our proof does not quite yield the full statement. 

Another fundamental consequence of Theorem \ref{thm1} concerns the optimality of the performance of AMP. 

\begin{corollary}[Optimality of AMP]\label{perfamp}
For $\Delta <\Delta_{\rm AMP}$ or $\Delta > \Delta_{\rm RS}$, the AMP is asymptotically optimal in the sense that it yields upon convergence the asymptotic vector-MMSE and matrix-MMSE 
of Corollary \ref{cor:MMSE}. Namely, 
\begin{align}
\lim_{t\to +\infty}\lim_{n\to +\infty} {\rm Mmse}_{n, \rm AMP}^{(t)}(\Delta^{-1})
&= \lim_{n\to\infty} {\rm Mmmse}_n(\Delta^{-1}).
\label{ampcoro1}
\\
\lim_{t\to +\infty}\lim_{n\to +\infty} {\rm Vmse}_{n, \rm AMP}^{(t)}(\Delta^{-1}) & = \lim_{n\to \infty}{\rm Vmmse}_n(\Delta^{-1}).
\label{ampcoro2}
\end{align}
On the other hand, for $\Delta_{\rm AMP}<\Delta < \Delta_{\rm RS}$ the AMP algorithm is strictly suboptimal, namely 
\begin{align}
\lim_{t\to +\infty}\lim_{n\to +\infty} {\rm Mmse}_{n, \rm AMP}^{(t)}(\Delta^{-1})
& > \lim_{n\to\infty} {\rm Mmmse}_n(\Delta^{-1}).
\label{ampsub1}
\\
\lim_{t\to +\infty}\lim_{n\to +\infty} {\rm Vmse}_{n, \rm AMP}^{(t)}(\Delta^{-1}) & > \lim_{n\to \infty}{\rm Vmmse}_n(\Delta^{-1}).
\label{ampsub2}
\end{align}
\end{corollary}
\begin{proof}
See Section \ref{sec:proof_coro}.
\end{proof}

This leaves the region $\Delta_{\rm AMP} < \Delta < \Delta_{\rm RS}$
algorithmically open for efficient polynomial time algorithms. A
natural conjecture, backed up by many results in spin glass theory,
coding theory, planted models and the planted clique problems,
is:
\begin{conjecture}
  For $\Delta_{\rm AMP} < \Delta < \Delta_{\rm RS}$, no polynomial time efficient algorithm that outperforms AMP exists.
\end{conjecture}

\subsection{Main results: coupled model}\label{subsec:coupled-results}

In this work, the spatially coupled construction is used for the purposes of the proof. 
However, one can also imagine interesting applications of the spatially coupled estimation problem, specially in view of the fact that
AMP turns out to be optimal for the spatially coupled system.
In coding theory for example, the use of spatially coupled systems as a proof device historically followed their
initial construction which was for engineering purposes and led to the construction of capacity achieving codes.

Our first crucial result states that the mutual information of the coupled and original systems are the same
in a suitable limit. The mutual information of the coupled system of length $L$ and with coupling window $w$ is denoted 
$I_{w, L}(\tbf{S}; \tbf{W})$.

\begin{theorem}[Equality of mutual informations]\label{LemmaGuerraSubadditivityStyle}
For any fixed $w$ s.t. $P_0$ satisfies Assumption \ref{ass:h1}, the following limits exist and are equal
\begin{align}
\lim_{L\to \infty}\lim_{n\to \infty} \frac{1}{n(L+1)}I_{w, L}(\tbf{S}; \tbf{W}) = \lim_{n\to \infty} \frac{1}{n}I(\tbf{S}; \tbf{W}).
\end{align}
\end{theorem}
\begin{proof}
See Section \ref{sec:subadditivitystyle}.
\end{proof}
An immediate corollary is that the non-analyticity points (w.r.t $\Delta$) of the mutual informations
are the same in the coupled and underlying models.
In particular, defining 
$\Delta^{\rm c}_{\rm opt}\! \defeq \!
\sup\{\Delta\mid \lim_{L\to \infty}\lim_{n\to \infty} I_{w, L}(\tbf{S}; \tbf{W})/(n(L\!+\!1)) \ \text{is analytic in}\ ]0, \Delta[\}$,
we have 
$\Delta^{\rm c}_{\rm opt} \!= \!\Delta_{{\rm opt}}$.

The second crucial result states that the AMP threshold of the spatially coupled system is at least as good as $\Delta_{\rm RS}$ (threshold saturation result of Theorem \ref{thres-sat-lemma}). The analysis of
AMP applies to the coupled system as well
[\cite{BayatiMontanari10,Montanari-Javanmard}] and it can be shown that
the performance of AMP is assessed by SE. Let 
\begin{align}
E_{\mu}^{(t)} \!\defeq\!\lim_{n\to\infty}\frac{1}{n}\mathbb{E}_{{\tbf S}, {\tbf Z}}[\|{\tbf S}_{\mu} \!-\! \hat{\tbf S}_{\mu}^{(t)}\|_2^2]
\end{align}
 be the asymptotic average vector-MSE of the AMP estimate
$\hat{\tbf{S}}_\mu^{(t)}$ at time $t$ for the $\mu$-th ``block'' of ${\tbf S}$. We associate to each position $\mu\!\in\! \{0, \ldots, L\}$ an \emph{independent} scalar system with AWGN of the form $y\!=\!s\!+\!\Sigma_\mu({\tbf E}; \Delta)z$, with 
\begin{align}
\Sigma_\mu({\tbf E})^{-2} \defeq \frac{v\!-\! \sum_{\nu=0}^L\Lambda_{\mu\nu} E_\nu}{\Delta}
\end{align}
and $S\!\sim\! P_0$, $Z\!\sim \!\mathcal{N}(0,1)$. Taking into account knowledge of the signal components in the seed 
$\mathcal{B}$, SE reads
\begin{align}\label{coupled-state-evolution}
\begin{cases}
 E_{\mu}^{(t+1)} = {\rm mmse}(\Sigma_\mu({\tbf E}^{(t)}; \Delta)^{-2}), ~ E_\mu^{(0)} = v ~\text{for}~ \mu\in \{0,\ldots, L\}\setminus\mathcal{B}, 
 \\
 E_\mu^{(t)} = 0 ~\text{for}~ \mu\in \mathcal{B}, t\geq 0
 \end{cases}
\end{align}
where the mmse function is defined as in \eqref{eq:SE_post}. 

From the monotonicity of the mmse function we have 
 $E_\mu^{(t+1)} \leq E_\mu^{(t)}$ for all $\mu \in \{0, \ldots, L\}$, a partial order 
 which implies that $\lim_{t\to\infty} {\tbf E}^{(t)} = {\tbf E}^{(\infty)}$ exists. 
 This allows to define an algorithmic threshold for the coupled system on a finite chain:
$$
\Delta_{{\rm AMP}, w,L} \defeq \sup\{\Delta|E^{(\infty)}_\mu\!\le\! E_{\rm good}(\Delta) \ \forall \ \mu\}
$$
where $E_{\rm good}(\Delta)$ is the trivial fixed point solution of the SE starting with the initial condition $E^{(0)} = 0$. A more formal but equivalent definition of $\Delta_{{\rm AMP}, w,L}$ is given in Section \ref{sec:thresh-sat}. 

\begin{theorem}[Threshold saturation]\label{thres-sat-lemma} 
Let $\Delta^{\rm c}_{\rm AMP}$ be the algorithmic threshold on an infinite chain, $\Delta^{\rm c}_{\rm AMP} \!\defeq \!\liminf_{w\to \infty}\liminf_{L\to \infty} \Delta_{{\rm AMP}, w, L}$, s.t. $P_0$ satisfies Assumptions \ref{ass:h1} and \ref{ass:h2}.
We have $\Delta^{\rm c}_{\rm AMP} \!\geq \!\Delta_{\rm RS}$. 
\end{theorem}
\begin{proof}
See Section \ref{sec:thresh-sat}.
\end{proof}
Our techniques also allow to prove the equality $\Delta^{\rm c}_{\rm AMP} = \Delta_{\rm RS}$, but this is not directly needed.

\subsection{Road map of the proof of the replica symmetric formula}

Here we give a road map of the proof of Theorem \ref{thm1} that will occupy Sections \ref{sec:thresh-sat}--\ref{sec:proof_theorem}. A fruitful idea is to concentrate on the question whether $\Delta_{\rm opt} = \Delta_{\rm RS}$. The proof of this 
equality automatically generates the proof of Theorem \ref{thm1}. 

We first prove in Section \ref{sec:delta-less-delta-opt} that $\Delta_{\rm opt}\leq \Delta_{\rm RS}$. This proof is based on a joint use of the I-MMSE relation (Lemma \ref{lemma:immse}), the replica bound \eqref{eq:guerrabound} and the suboptimality of the AMP algorithm. 
In the process of proving $\Delta_{\rm opt}\leq \Delta_{\rm RS}$, we in fact get as a direct bonus the proof of Theorem \ref{thm1} for $\Delta < \Delta_{\rm opt}$.  

The proof of $\Delta_{\rm opt}\geq \Delta_{\rm RS}$ requires the use of spatial coupling. The main strategy is to show
\begin{align}\label{eq:chain_delta}
\Delta_{\rm RS} \le \Delta^{\rm c}_{\rm AMP} \leq \Delta^{\rm c}_{\rm opt} = \Delta_{\rm opt}.
\end{align}
The first inequality in \eqref{eq:chain_delta} is proven in Section \ref{sec:thresh-sat} using methods first invented in coding theory: The algorithmic AMP
threshold of the spatially coupled system $\Delta^{\rm c}_{\rm AMP}$ saturates (tends in a suitable limit) towards $\Delta_{\rm RS}$, i.e. 
$ \Delta_{\rm RS} \le \Delta^{\rm c}_{\rm AMP}$ (Theorem \ref{thres-sat-lemma}). 
To prove the (last) equality we show in Section \ref{sec:subadditivitystyle} that the free energies, and hence the mutual informations, of the underlying and spatially coupled systems are equal in a suitable asymptotic limit (Theorem \ref{LemmaGuerraSubadditivityStyle}). This implies that their non-analyticities occur at the same 
point and hence $\Delta^{\rm c}_{\rm opt} = \Delta_{\rm opt}$. This is done through an interpolation which, although similar in spirit, is different than the one used 
to prove replica bounds (e.g. \eqref{eq:guerrabound}). In the process of showing 
$\Delta^{\rm c}_{\rm opt} = \Delta_{\rm opt}$, we will also derive the existence of the limit for $I(\tbf{S}; {\tbf W})/n$.
Finally, the second inequality is due the suboptimality of the AMP algorithm. This follows by a direct extension of the SE analysis of [\cite{6875223,deshpande2015asymptotic}] to the spatially coupled case as done in [\cite{Montanari-Javanmard}]. 

Once $\Delta_{\rm opt} = \Delta_{\rm RS}$ is established it is easy to put everything together 
and conclude the proof of Theorem \ref{thm1}. In fact all that remains is to prove Theorem \ref{thm1} for $\Delta >\Delta_{\rm opt}$.
This follows by an easy argument in section \ref{sec:everything-together} which combines  $\Delta_{\rm opt} = \Delta_{\rm RS}$, the replica bound \eqref{eq:guerrabound} and the suboptimality of the AMP algorithm. Note that in the proof sections that follow, we assume that Assumptions \ref{ass:h1}-\ref{ass:h3} hold.

\subsection{Connection with the planted Sherrington-Kirkpatrick spin glass model} 

Let us briefly discuss the connection of the matrix factorization problem \eqref{eq:mainProblem} with a statistical mechanical spin glass model which is a variant of the classic Sherrington-Kirkpatrick (SK) model. This is also the occasion to express the mutual information as a ``free energy" through a simple relation that will be used in various guises later on. 

Replacing $w_{ij} = n^{-1/2}s_is_j + \sqrt{\Delta} z_{ij}$ in \eqref{eq:posterior_dist} and simplifying the fraction after expanding the squares, the posterior distribution can be expressed in terms of $\bs, \bz$ as follows
\begin{align}\label{gibbsi}
P(\bx \vert \bs, \bz) = \frac{1}{\mathcal{Z}} e^{-\mathcal{H}(\bx\vert \bs, \bz)} \prod_{i=1}^n P_0(x_i),
\end{align}
where 
\begin{align}\label{hami}
\mathcal{H}(\bx\vert \bs, \bz) = \sum_{i\leq j=1}^n \Big(\frac{x_i^2 x_j^2}{2n\Delta} - \frac{s_i s_j x_i x_j}{n\Delta} - \frac{z_{ij} x_ix_j}{\sqrt{n \Delta}}\Big)
\end{align}
and 
\begin{align}\label{parti}
\mathcal{Z} = \int \Big\{\prod_{i=1}^n dx_i P_0(x_i)\Big\} e^{-\mathcal{H}(\bx\vert \bs, \bz)}.
\end{align}
In the language of statistical mechanics, \eqref{hami} is the ``Hamiltonian", \eqref{parti} the ``partition function", and \eqref{gibbsi} is the Gibbs distribution.
This distribution is random since it depends on the realizations of $\bS$, $\bZ$. Conditional expectations with respect to \eqref{gibbsi}
are denoted by the Gibbs ``bracket" $\langle - \rangle$. More precisely
\begin{align}
\mathbb{E}_{\bX\vert \bS, \bZ}\big[A(\bX) \vert \bS= \bs, \bZ=\bz\big] = \langle A(\bX) \rangle.
\end{align}
The free energy is defined as 
\begin{align}
f_n = -\frac{1}{n} \mathbb{E}_{\bS, \bZ}[ \ln \mathcal{Z}].
\end{align}
Notice the difference between $\tilde{\mathcal{Z}}$ in \eqref{eq:posterior_dist} and $\mathcal{Z}$ in \eqref{gibbsi}. The former is the partition function with a complete square, whereas the latter is the partition function that we obtain after expanding the square and simplifying the posterior distribution.

In Appendix \ref{app:freeEnergy}, we show that mutual information and free energy are essentially the same object up to a trivial term. For the present model
\begin{align}\label{eq:energy_MI}
\frac{1}{n} I(\tbf{S}; \bW) = -\frac{1}{n} \mathbb{E}_{\bS, \bZ}[ \ln \mathcal{Z}] + \frac{v^2}{4\Delta} + \frac{1}{4 \Delta n} (2\mathbb{E}[S^4] - v^2),
\end{align}
where recall $v = \mathbb{E}[S^2]$. This relationship turns out to be very practical and will be used several times.

For binary signals we have $s_i$ and $ x_i\in \{-1, +1\}$, so the model is a binary spin glass model. The first term in the Hamiltonian is a trivial constant,  the last term corresponds exactly to the SK model with random Gaussian interactions, and the second term can be interpreted as an external random field that biases the spins. This is sometimes called a ``planted" SK model.

The rest of the paper is organized as follows.
In Section \ref{sec:examples} we provide two examples of the symmetric rank-one matrix estimation problem. 
Threshold saturation and the invariance of the mutual information due the spatial coupling are shown in Section \ref{sec:thresh-sat} and \ref{sec:subadditivitystyle} respectively. The proof of Theorem \ref{thm1} follows in Section \ref{sec:proof_theorem}. Section \ref{sec:proof_coro} is dedicated to the proof of Corollary \ref{cor:MMSE} and Corollary \ref{perfamp}.

\section{Two Examples: spiked Wigner model and community detection} \label{sec:examples}
In order to illustrate our results, we shall present them here in the
context of two examples: the spiked Wigner model, where we close a
conjecture left open by [\cite{6875223}], and the case of asymmetric
community detection. 
\subsection{Spiked Wigner model}
\begin{figure}[!t]
\centering
\includegraphics[width=0.49\textwidth]{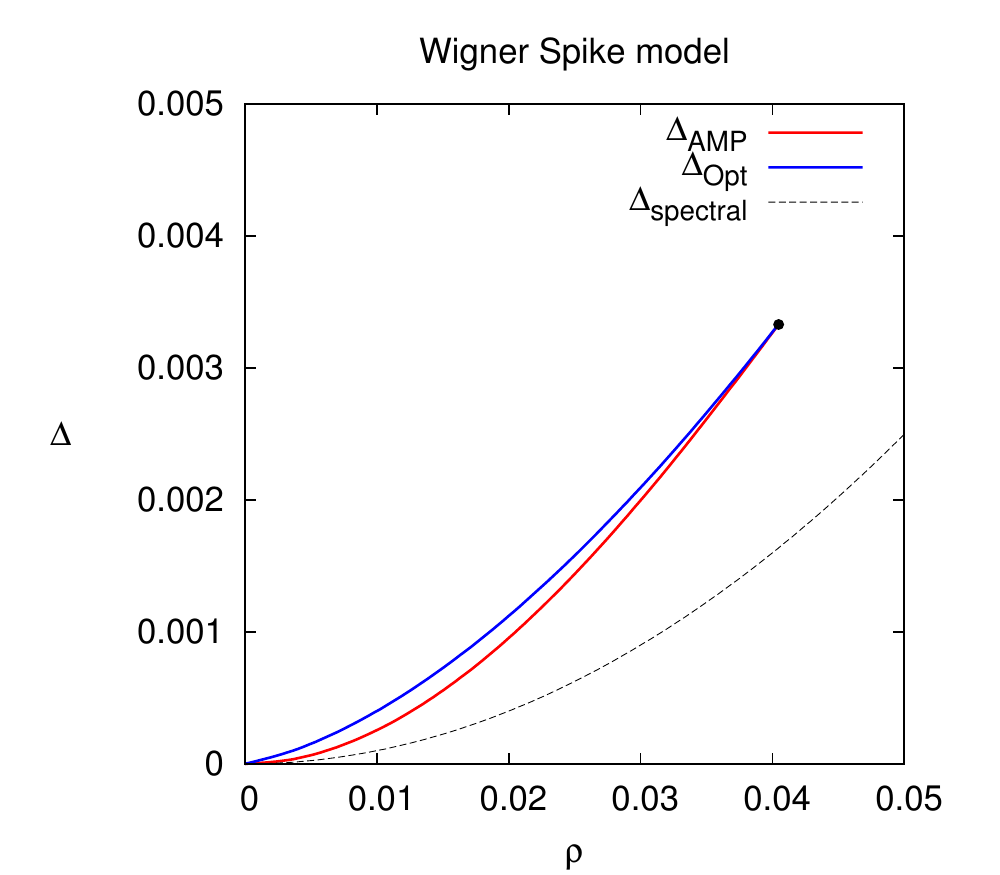}
\includegraphics[width=0.49\textwidth]{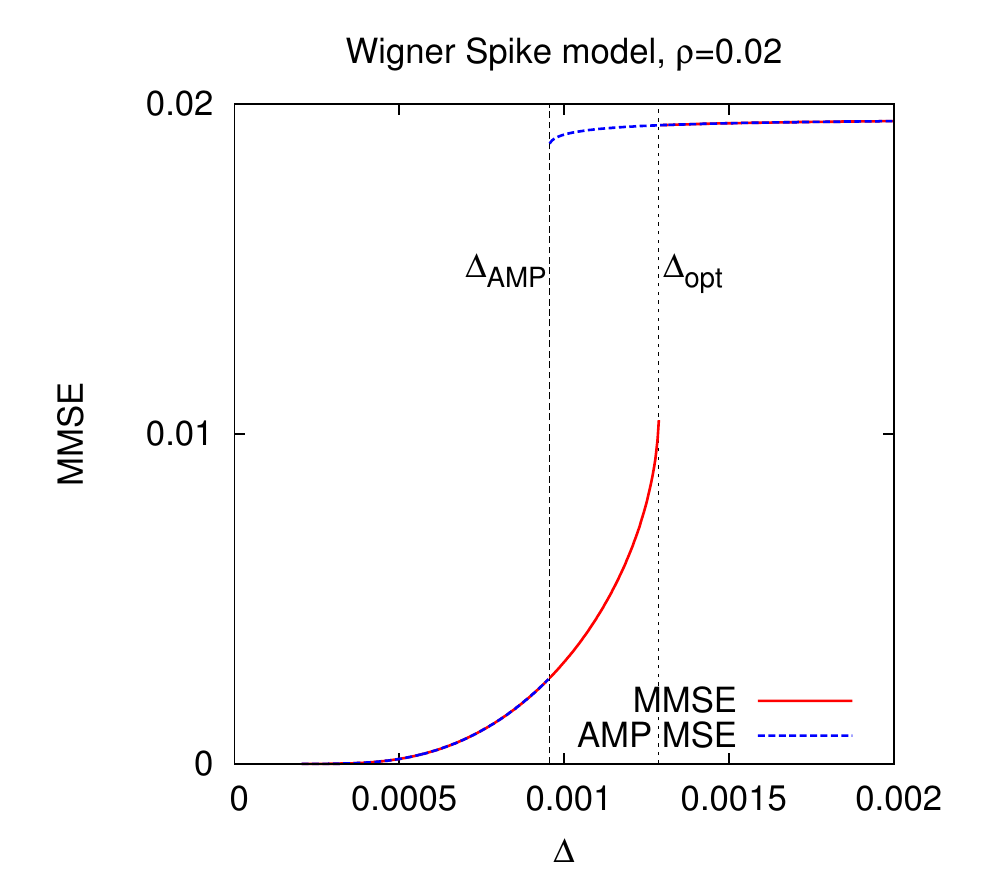}
\vspace*{-5pt}
\caption{Left: close up of the phase diagram in the density $\rho$
  (where $P_0(s)=\rho \delta(s-1) + (1-\rho)\delta(x)$) versus noise
  $\Delta$ plane for the rank-one spiked Wigner model (data from
  [\cite{lesieur2015phase}]). In [\cite{6875223}], the authors proved that AMP achieves
  the information theoretic MMSE for all noise as long as
  $\rho>0.041(1)$. We show that AMP is actually achieving the optimal
  reconstruction in the whole phase diagram except in the small region
  between the blue and red line. Notice the large gap with spectral
  methods (dotted black line).  Right: MMSE (red) at $\rho=0.02$ as a
  function of the noise variance $\Delta$. AMP (dashed blue) provably
  achieve the MMSE except in the region
  $\Delta_{\rm AMP}<\Delta<\Delta_{\rm opt}$. We conjecture that no
  polynomial-time algorithm will do better than AMP in this region.}
\label{fig:SPIKEMODEL}
\end{figure}
The first model is defined as follows: we are given data distributed
according to the spiked Wigner model where the vector $\tbf{s}$ is
assumed to be a Bernoulli $0/1$ random variable with probability
$\rho$. Data then consists of a sparse, rank-one matrix observed
through a Gaussian noise. In [\cite{6875223}], the authors proved that, for
$\rho>0.041$, AMP is a computationally efficient algorithm that
asymptotically achieves the information theoretically optimal mean-square error for {\it any} value of the noise $\Delta$.  

For very small densities (i.e. when $\rho$ is $o(1)$), there is a well
known large gap between what is information theoretically possible and
what is tractable with current algorithms in support recovery
[\cite{amini2008high}]. This gap is actually related to the planted
clique problem [\cite{d2007direct,barak2016nearly}], where it is
believed that no polynomial algorithm is able to achieve information
theoretic performances. It is thus perhaps not surprising that the
situation for $\rho<0.041$ becomes a bit more complicated. This is
summarized in Fig. \ref{fig:SPIKEMODEL} and discussed in
[\cite{lesieur2015phase}] on the basis of statistical physics
consideration which we now prove.

For such values of $\rho$, as $\Delta$ changes there is a region when
two local minima appears in $i_{\rm RS}(E; \Delta)$ (see the RS formula \eqref{eq:potentialfunction}). In particular for
$\Delta_{\rm AMP}<\Delta<\Delta_{\rm opt}$, the global minimum differs from the
AMP one and a computational gap appears (see right panel in Fig.
\ref{fig:SPIKEMODEL}). Interestingly, in this problem, the region
where AMP is Bayes optimal is still quite large.

The region where AMP is not Bayes optimal is perhaps the most
interesting one.  While this is by no means evident, statistical
physics analogies with actual phase transition in nature suggest that
this region will be hard for a very large class of algorithms. A fact
that add credibility to this prediction is the following: when looking
to small $\rho$ regime, we find that both the information
theoretic threshold {\it and} the AMP one corresponds to what has been
predicted in sparse PCA for sub-extensive values of $\rho$ [\cite{amini2008high}].

Finally, another interesting line of work for such probabilistic
models has appeared in the context of random matrix theory (see for
instance [\cite{baik2005phase}] and references therein). The focus is to
analyze the limiting distribution of the eigenvalues of the observed
matrix. The typical picture that emerges from this line of work is
that a sharp phase transition occurs at a well-defined critical value
of the noise. Above the threshold an outlier eigenvalue (and the
principal eigenvector corresponding to it) has a positive correlation
with the hidden signal. Below the threshold, however, the spectral
distribution of the observation is indistinguishable from that of the
pure random noise. In this model, this happens at
$\Delta_{\rm{spectral}} =\rho^2$. Note that for $\Delta>\Delta_{\rm{spectral}}$
spectral methods are not able to distinguish data coming from the model from
random ones, while AMP is able to sort (partly) data from noise for
any values of $\Delta$ and $\rho$.
%
%
\subsection{Asymmetric community detection}
\begin{figure}[!t]
\centering
\includegraphics[width=0.49\textwidth]{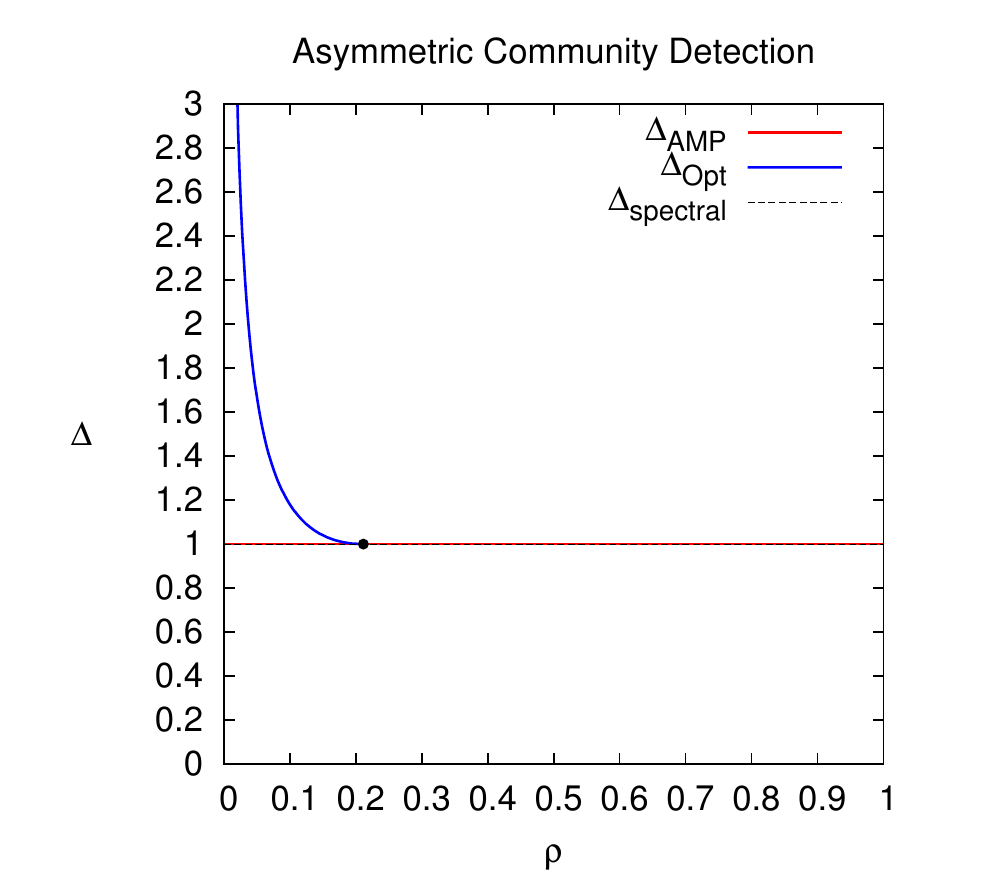}
\includegraphics[width=0.49\textwidth]{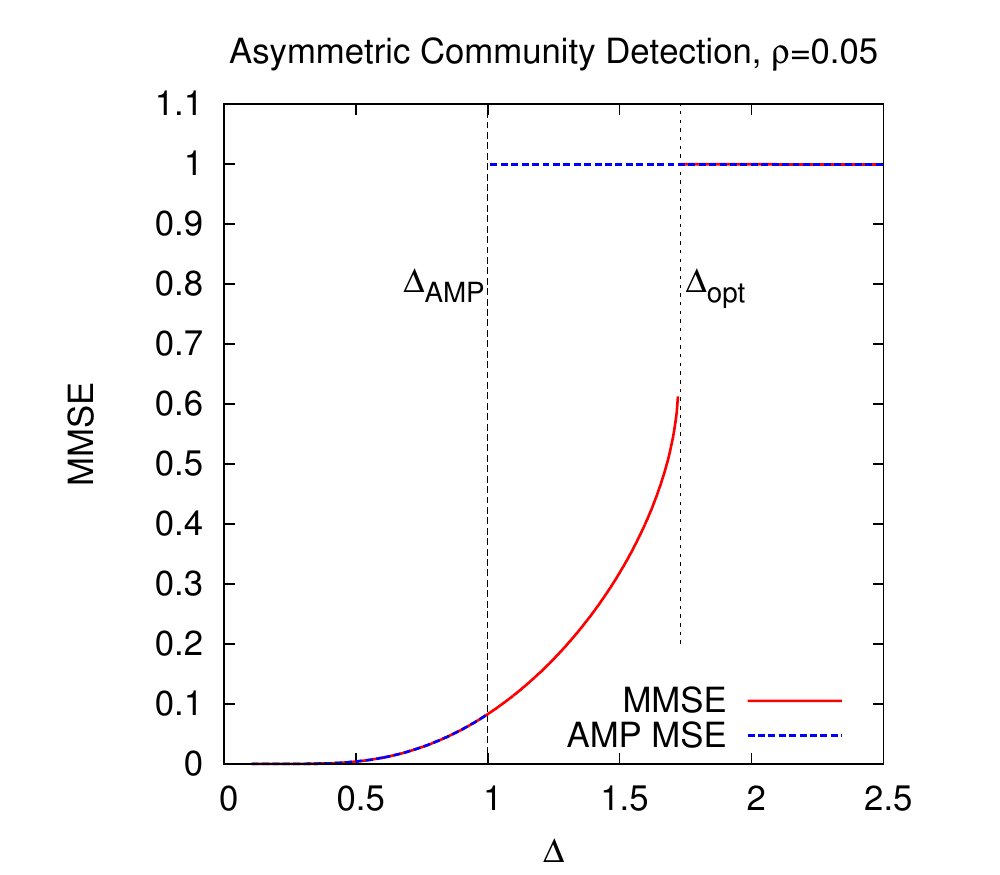}
\vspace*{-5pt}
\caption{Left: Phase diagram in the density $\rho$ (where
  $P_0(s)=\rho \delta(s-\sqrt{1-\rho/\rho}) +
  (1-\rho)\delta(s+\sqrt{\rho/1-\rho})$)
  corresponding to an asymmetric community detection problem with two
  communities. For $\rho>1-\sqrt{1/12}$ (black point), it is
  information theoretically impossible when $\Delta>1$ to find any
  overlap with the true communities and the optimal MMSE is equal $1$,
  while it is possible for $\Delta<1$.  AMP is again always achieving
  the MMSE, and spectral method can find a non-trivial overlap with
  the truth starting from $\Delta<1$. For $\rho<1-\sqrt{1/12}$,
  however, it is information theoretically possible to find
  information on the hidden assignment (below the blue line), but
  both AMP and spectral methods misses this information: Right:MMSE
  (red) at $\rho=0.05$ as a function of the noise variance
  $\Delta$. AMP (dashed blue) provably achieve the MMSE except in the
  region
  $\Delta_{\rm{spectral}}=\Delta_{\rm AMP}<\Delta<\Delta_{\rm opt}$. }\label{fig:CSMODEL}
\end{figure}
The second model is a problem of detecting two communities (groups) with
different sizes $\rho n$ and $(1\!-\!\rho) n$, that generalizes the one considered in
[\cite{deshpande2015asymptotic}]. 
One is given a graph
where the probability to have a link between nodes in the first group
is $p\!+\! \mu (1\!-\!\rho)/(\rho\sqrt{n})$, between those in the second
group is $p\!+\!\mu\rho/(\sqrt{n}(1\!-\!\rho))$, while interconnections appear
with probability $p\!-\!\mu/\sqrt{n}$. With this peculiar ``balanced''
setting, the nodes in each group have the same degree distribution with
mean $pn$, making them harder to distinguish.

According to the universality property described in Section \ref{sec:settings_results}, this is equivalent to the AWGN model \eqref{eq:mainProblem} with variance $\Delta\!=\!p(1\!-\!p)/\mu^2$ where each variable $s_i$ is chosen according to 
\begin{align}
P_0(s) = \rho \delta(s - \sqrt{(1 - \rho)/\rho}) +
(1 - \rho)\delta(s + \sqrt{\rho/(1 - \rho)}).
\end{align}

Our results for this problem\footnote{Note that here since $E = v = 1$ is an extremum of $i_{\rm RS}(E;\Delta)$,
one must introduce a small bias in $P_0$ and let it then tend to zero at the end of the proofs.} 
are
summarized in Fig. \ref{fig:CSMODEL}. For
$\rho\!>\rho_c\!=\!1/2\!-\!\sqrt{1/12}$ (black point), it is
asymptotically information theoretically {\it possible} to get an
estimation better than chance if and only if $\Delta\!<\!1$. When
$\rho\!<\!\rho_c$, however, it becomes possible for much larger values
of the noise. Interestingly, AMP and spectral methods have the same
transition and can find a positive correlation with the hidden
communities for $\Delta\!<\!1$, regardless of the value of $\rho$.
\section{Threshold saturation}\label{sec:thresh-sat}
The main goal of this section is to  prove that for a proper spatially coupled (SC) system, threshold saturation 
occurs (Theorem \ref{thres-sat-lemma} ), that is $\Delta_{\rm RS} \le \Delta^{\rm c}_{\rm AMP}$.
We begin with some preliminary formalism in Sec. \ref{stateuncoupled}, \ref{statecoupled} on state evolution for the underlying and coupled systems. 

\subsection{State evolution of the underlying system}\label{stateuncoupled}
First, define the following posterior average
\begin{align}
\langle A \rangle \defeq \frac{\int dx A(x) P_0(x)e^{-\frac{x^2}{2\Sigma(E,\Delta)^{2}} + x\big(\frac{s}{\Sigma(E,\Delta)^{2}} + \frac{z}{\Sigma(E,\Delta)} \big)} }{\int dx P_0(x)
e^{-\frac{x^2}{2\Sigma(E,\Delta)^{2}} + x\big(\frac{s}{\Sigma(E,\Delta)^{2}} + \frac{z}{\Sigma(E,\Delta)} \big)}}, \label{eq:defBracket_posteriorMean}
\end{align}
where $S\sim P_0$, $Z\sim\mathcal{N}(0,1)$. The dependence on these variables, as well as on $\Delta$ and $E$ is implicit and dropped from the notation of $\langle A \rangle$. 
Let us define the following operator.
\begin{definition}[SE operator] 
The state evolution operator associated with the underlying system is
\begin{align}
T_{\rm u}(E) \defeq {\rm mmse}(\Sigma(E)^{-2}) = \mathbb{E}_{S,Z}[(S - \langle X\rangle)^2], \label{def:SEop_underlying}
\end{align}
where $S \sim P_0$, $Z\sim \mathcal{N}(0,1)$.
%
\end{definition}
The fixed points of this operator play an important role. They can be viewed as the stationary points of the replica symmetric potential function $E\in [0,v] \mapsto i_{\rm RS}(E;\Delta)\in \mathbb{R}$
or equivalently of $f_{\rm RS}^{\rm u}: E\in [0,v]\mapsto f_{\rm RS}^{\rm u}(E)\in \mathbb{R}$ where 
\begin{align}
f_{\rm RS}^{\rm u}(E) \defeq i_{\rm RS}(E;\Delta) - \frac{v^2}{4\Delta}.
\end{align}
It turns out to be more convenient to work with $f_{\rm RS}^{\rm u}$ instead of $i_{\rm RS}$. We have 
\begin{lemma}\label{lemma:fixedpointSE_extPot_underlying}
Any fixed point of the SE corresponds to a stationary point of $f_{\rm RS}^{\rm u}$:
\begin{align}
 E = T_{{\rm u}}(E) \Leftrightarrow \frac{\partial f_{\rm RS}^{\rm u}(E;\Delta)}{\partial E}\big|_{E} =0.
\end{align}
\end{lemma} 
\begin{proof}
See Appendix~\ref{app:nishi}.
\end{proof}

The asymptotic performance of the AMP algorithm can be tracked by iterating the SE recursion as follows (this is the same as equation \eqref{eq:state-evolution} expressed here with the help of 
$T_{\rm u}$)
\begin{align} \label{eq:SEiteration}
E^{(t+1)}=T_{\rm u}(E^{(t)}), \quad t\geq 0, \quad E^{(0)} = v, 
\end{align}
where the iteration is initialized without any knowledge about the signal other than its prior distribution 
(in fact, both the asymptotic vector and matrix MSE of the AMP are tracked by the previous recursion as reviewed in Section \ref{subsec:AMP_SE}). 
Let $E_{\rm good}(\Delta) = T_{\rm u}^{(\infty)}(0)$, the fixed point reached by initializing iterations at $E=0$. With our hypothesis on $P_0$ it is not difficult to see that
definition \ref{amp-thresh} is equivalent to 
\begin{align}\label{equivDeltaAMP}
\Delta_{{\rm AMP}} \defeq {\rm sup}\,\{\Delta>0\ \! |\ \! T_{{\rm u}}^{(\infty)}(v) = E_{\rm good}(\Delta)\}\,.
\end{align}
The following definition is handy
%
%
\begin{definition}[Bassin of attraction]
The basin of attraction of the good solution $E_{\rm good}(\Delta)$ is 
$\mathcal{V}_{\rm good} \defeq \{ E \ \! |\ \! T_{\rm u}^{(\infty )}(E) = E_{\rm good}(\Delta) \}$.
\end{definition}

Finally, we introduce the notion of {\it potential gap}. This is a function $\delta f_{\rm RS}^{\rm u}: \Delta\in \mathbb{R}_+\mapsto \delta f_{\rm RS}^{\rm u}(\Delta)\in \mathbb{R}$ defined as follows:
\begin{definition}[Potential gap] \label{def:fgap}
Define 
\begin{align}
 \delta f_{\rm RS}^{\rm u}(\Delta) \defeq {\rm inf} _{E \notin \mathcal{V}_{\rm good} } (f_{\rm RS}^{\rm u}(E)-f_{\rm RS}^{\rm u} (E_{\rm good}))
\end{align}
 as 
the potential gap, with the convention that the infimum over the empty set is $\infty$ (this happens for $\Delta < \Delta_{{\rm AMP}}$ where the complement of $\mathcal{V}_{\rm good}$ is the empty set).
\end{definition}
Our hypothesis on $P_0$ imply that 
\begin{align}
 \Delta_{\rm RS} \defeq {\rm sup}\,\{\Delta>0\ \! |\ \! \delta f_{\rm RS}^{\rm u}(\Delta) > 0\}
\end{align}

\subsection{State evolution of the coupled system}\label{statecoupled}
For the SC system, the performance of the AMP decoder is tracked by an \emph{MSE profile} (or just \emph{profile}) $\tbf{E}^{(t)}$, defined componentwise by 
\begin{align}
E^{(t)}_\mu = \lim_{n\to +\infty}\frac{1}{n}\mathbb{E}_{\tbf S, \tbf W}\| {\tbf S}_\mu - \hat{\tbf S}_{\mu}^{(t)}(\tbf W) \|_2^2.
\end{align}
It has $L+1$ components and describes the scalar MSE in each block $\mu$. Let us introduce the SE associated with the AMP algorithm for the inference over this SC system. First, denote the following posterior average at fixed $s,z$ and $\Delta$.
\begin{align}
\langle A \rangle_\mu  \defeq \frac{\int dx A(x) P_0(x)e^{-\frac{x^2}{2\Sigma_\mu ({\tbf E},\Delta)^{2}} + x\big(\frac{s}{\Sigma_\mu ({\tbf E},\Delta)^{2}} + \frac{z}{\Sigma_\mu ({\tbf E},\Delta)} \big)} }{\int dx P_0(x)
e^{-\frac{x^2}{2\Sigma_\mu ({\tbf E},\Delta)^{2}} + x\big(\frac{s}{\Sigma_\mu ({\tbf E},\Delta)^{2}} + \frac{z}{\Sigma_\mu ({\tbf E},\Delta)} \big)}}. \label{def:posterior_coupled}
\end{align}
%
where the \emph{effective noise variance} of the SC system is defined as
\begin{align} \label{eq:defSigma_mu}
\Sigma_\mu(\tbf{E})^{-2} \defeq \frac{v-\sum_{\nu \in \mathcal{S}_\mu} \Lambda_{\mu\nu}E_\nu}{\Delta},
\end{align}
where we recall $\mathcal{S}_\mu \defeq \{\nu \,|\, \Lambda_{\mu\nu} \neq 0 \}$ is the set of $2w+1$ blocks coupled to block $\mu$.
\begin{definition}[SE operator of the coupled system]
The state evolution operator associated with the coupled system (\ref{eq:def_coupledSyst}) is defined component-wise as
\begin{align}
[T_{\rm c}({\tbf E})]_\mu \defeq \mathbb{E}_{S,Z}[(S - \langle X\rangle_\mu)^2]. \label{def:SEop_SC}
\end{align}
$T_{\rm c}({\tbf E})$ is vector valued and here we have written its $\mu$-th component. 
\end{definition}
We assume perfect knowledge of the 
variables $\{s_{i_\mu}\}$ inside the blocks $\mu \in\mathcal{B}\defeq \{0:w-1\}\cup\{L-w:L\}$ as mentioned 
in Section \ref{subsec:SC}, that is $x_{i_\mu} = s_{i_\mu}$ for all $i_\mu$ such that $\mu\in\mathcal{B}$. This 
implies $E_\mu=0 \ \forall \ \mu \in \mathcal{B}$. We refer to this as the {\it pinning condition}. 
The SE iteration tracking the scalar MSE profile of the SC system reads for $\mu\notin \mathcal{B}$
\begin{equation}\label{equ:statevolutioncoupled}
E_\mu^{(t+1)} = [T_{\rm c}({\tbf E}^{(t)})]_\mu \quad \forall \ t\geq 0,
\end{equation}
with the initialization $E^{(0)}_\mu = v$. For $\mu\in \mathcal{B}$, the pinning condition forces $E_\mu^{(t)} = 0 \ \forall \ t$. This equation is the same as \eqref{coupled-state-evolution} but is expressed here in terms of 
the operator $T{\rm c}$.

Let us introduce a suitable notion of degradation that will be very useful for the analysis.
\begin{definition}[Degradation] \label{def:degradation}
A profile ${\tbf{E}}$ is degraded (resp. strictly degraded) w.r.t another one ${\tbf{G}}$, denoted as ${\tbf{E}} \succeq {\tbf{G}}$ (resp. ${\tbf{E}} \succ {\tbf{G}}$), if $E_\mu \ge  G_\mu \ \forall \ \mu$ 
(resp. if ${\tbf{E}} \succeq {\tbf{G}}$ and there exists some $\mu$ such that $E_\mu > G_\mu$).
\end{definition}
Define an error profile $\tbf{E}_{\rm good}(\Delta)$ as the vector with all $L+1$ components equal to $E_{\rm good}(\Delta)$.
\begin{definition} [AMP threshold of coupled ensemble]\label{def:AMPcoupled}
The AMP threshold of the coupled system is defined as 
\begin{align}
 \Delta^{\rm c}_{\rm AMP} \defeq {\liminf}_{w, L\to \infty}\, {\rm sup}\,\{\Delta>0\ \! |\ \! T_{{\rm c}}^{(\infty)}(\boldsymbol{v}) \prec \tbf{E}_{\rm good}(\Delta)\}
\end{align}
where $\boldsymbol{v}$ is the all $v$ vector. The ${\liminf}_{w, L \to \infty}$ is taken along sequences where {\it first} $L \to \infty$ and {\it then} $w\to\infty$. We also set for a finite system $\Delta_{{\rm AMP}, w, L} \defeq {\rm sup}\,\{\Delta>0\ \! |\ \! T_{{\rm c}}^{(\infty)}(\boldsymbol{v}) \prec \tbf{E}_{\rm good}(\Delta)\}$.
\end{definition}
%
%

The proof presented in the next subsection uses extensively the following monotonicity properties of the SE operators. 
\begin{lemma}
The SE operator of the SC system maintains degradation in space, 
i.e. ${\tbf E} \succeq  {\tbf G} \Rightarrow T_{\rm{c}}({\tbf E}) \succeq T_{\rm{c}}({\tbf{G}})$. This property is verified by $T_{\rm{u}}(E)$ for a scalar error as well.
\label{lemma:spaceDegrad}
\end{lemma}
\begin{proof}
From \eqref{eq:defSigma_mu} one can immediately see that ${\tbf E} \succeq  {\tbf G} \Rightarrow \Sigma_\mu({\tbf E}) \geq \Sigma_\mu({\tbf G})\ \forall \ \mu$. 
Now, the SE operator \eqref{def:SEop_SC} can be interpreted as the mmse function associated to the Gaussian channel
$y = s +\Sigma_\mu(\tbf E, \Delta) z$.
This is an increasing function of  the noise intensity $\Sigma_\mu^2$: this is intuitively clear but we provide an explicit formula for the 
derivative below. Thus $[T_{{\rm c}}( {\tbf E})]_\mu \geq [T_{{\rm c}}( {\tbf G})]_\mu \ \forall \ \mu$, which means $T_{{\rm c}}({\tbf E}) \succeq T_{{\rm c}}({\tbf{G}})$.

The derivative of the mmse function of the Gaussian channel can be computed as 
\begin{align}
 \frac{d \,{\rm mmse}(\Sigma^{-2})}{d\,\Sigma^{-2}} = - 2\, \mathbb{E}_{X,Y}\Big[\|X - \mathbb{E}[X\vert Y]\|_2^2 \, {\rm Var}[X\vert Y]\Big].
\end{align}
This formula explicitly confirms that $T_{\rm u}(E)$ (resp. $[T_{{\rm c}}(\tbf E)]_\mu$) is an increasing function of $\Sigma^2$ (resp. $\Sigma_\mu^2$).
\end{proof}
\begin{corollary}
\label{cor:timeDegrad}
The SE operator of the coupled 
system maintains degradation in time, 
i.e., $T_{{\rm c}}({\tbf E}^{(t)}) \preceq  {\tbf E}^{(t)}\Rightarrow T_{{\rm c}}({\tbf E}^{(t+1)}) \preceq  {\tbf E}^{(t+1)}$. 
Similarly $T_{{\rm c}}({\tbf E}^{(t)}) \succeq  {\tbf E}^{(t)} \Rightarrow T_{{\rm c}}({\tbf E}^{(t+1)}) \succeq  {\tbf E}^{(t+1)}$. 
Furthermore, the limiting error profile ${\tbf E}^{(\infty)} \defeq T_{{\rm c}}^{(\infty)}({\tbf E}^{(0)})$
exists. These properties are verified by $T_{{\rm u}}(E)$ as well.
\end{corollary}
\begin{proof}
The degradation statements are a consequence of Lemma~\ref{lemma:spaceDegrad}. The existence 
of the limits follows from the monotonicity of the operator and boundedness of the scalar MSE.
\end{proof}

Finally we will also need the following generalization of the (replica symmetric) potential function to a spatially coupled system:
\begin{align}
f_{\rm RS}^{\rm{c}}(\tbf{E}) = &\sum_{\mu=0}^L\sum_{\nu \in \mathcal{S}_\mu} \frac{\Lambda_{\mu\nu}}{4\Delta}(v-E_\mu) (v-E_\nu) \nonumber \\
- &\sum_{\mu=0}^L \mathbb{E}_{S,Z}\Big[\ln\Big(\int dx P_0(x)
 e^{-\frac{1}{2\Sigma_\mu(\tbf E)^{2}}\big(x^2 - 2 xS + xZ\Sigma_\mu(\tbf E,\Delta)\big)} \Big)\Big], \label{def:RSf_SC}
\end{align}
where $Z\sim \mathcal{N}(z|0,1)$ and $S \sim P_0(s)$.
As for the underlying system, the following Lemma links the SE and RS formulations.
\begin{lemma}\label{lemma:fixedpointSE_extPot}
If ${\tbf E}$ is a fixed point of \eqref{equ:statevolutioncoupled}, 
i.e. $ E_\mu = [T_{\rm c}({\tbf E})]_\mu \Rightarrow \frac{\partial f_{\rm RS}^{\rm{c}}({\tbf E})}{\partial E_\mu}\big|_{{{\tbf E}}} = 0 \ \forall \ \mu\in \mathcal{B}^{\rm c}=\{w:L-w-1\}$.
\end{lemma}
\begin{proof}
The proof is similar to the proof of Lemma \ref{lemma:fixedpointSE_extPot_underlying} in Appendix~\ref{app:nishi}. We skip the details for brevity.
\end{proof}

Now that we have settled the required definitions and properties, we can prove threshold saturation. 
\subsection{Proof of Theorem \ref{thres-sat-lemma}}
The proof will proceed by contradiction. Let ${\tbf E}^*$ a fixed point profile of
the SE iteration \eqref{equ:statevolutioncoupled}.
We suppose that ${\tbf E}^*$ {\it does not} satisfy ${\tbf E}^* \prec {\tbf E}_{\rm good}(\Delta)$, and exhibit a contradiction for $\Delta < \Delta_{\rm RS}$ and $w$ large enough
(but independent of $L$). Thus we must have ${\tbf E}^* \prec {\tbf E}_{\rm good}(\Delta)$. This is the statement of Theorem \ref{th:mainth_thsat} in Sec. \ref{endsat} and directly implies Theorem \ref{thres-sat-lemma}.

The pinning condition together with the monotonicity properties of the coupled SE operator (Lemma \ref{lemma:spaceDegrad} and Corollary \ref{cor:timeDegrad}) ensure that any fixed point 
profile ${\tbf E}^*$ which does not satisfy ${\tbf E}^* \prec {\tbf E}_{\rm good}(\Delta)$ necessarily has a shape as described in Fig.~\ref{fig:errorProfile}. We construct an associated \emph{saturated profile} $\tbf E$ as described in Fig.~\ref{fig:errorProfile}. 
From now on we work with a \emph{saturated profile} $\tbf E$ which verifies ${\tbf E} \succeq {\tbf E}^{*}$ and 
${\tbf E} \succeq {\tbf E}_{\rm good}(\Delta)$. In the following we will need the following operator.
\begin{figure}[!t]
\centering
\includegraphics[width=0.75\textwidth]{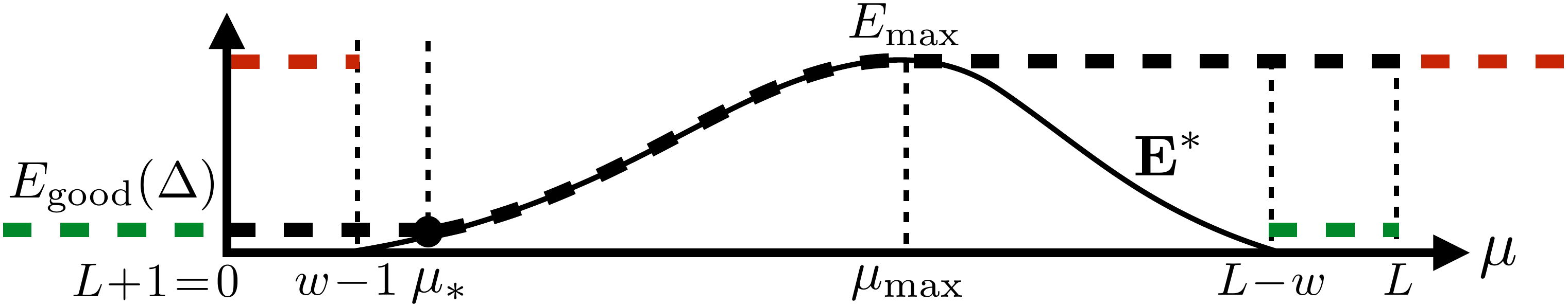}
\caption{A fixed point profile $\tbf{E}^{*}$ of the coupled SE iteration (solid line) is necessarily null
$\forall \ \mu \in \{0: w-1\}$
because of the pinning condition, and then it increases to reach $E_{{\rm max}} \in [0,v]$ at some $\mu_{\rm max}\in \{w:L-w-1\}$ (for a symmetric coupling matrix $\mu_{\rm max} = L/2$). 
Then, it starts decreasing and it is again null $\forall \ \mu \in \{L-w:L\}$. By definition, the associated \emph{saturated} profile $\tbf E$ (dashed line) starts 
at $E_{\rm good}(\Delta) \ \forall \ \mu \le \mu_*$, where $\mu_*$ is such that: $E^{*}_{\mu}\le E_{\rm good}(\Delta) \ \forall \ \mu\in\{0: \mu_*\}$ 
and $E^{*}_{\mu'} > E_{\rm good}(\Delta) \ \forall \ \mu'\in\{ \mu_*+1:L\}$. Then, ${\bf E}$ matches $\tbf E^{*} \ \forall \ \mu \in\{\mu_*:\mu_{\rm max}\}$ 
and saturates at $E_{{\rm max}} \ \forall \ \mu \ge \mu_{\rm max}$. The saturated profile is extended for $\mu <0$ and $\mu >L$ indices. The green (resp. red) branch 
shows that when the block indices of ${\bf E}$ are $\mu <0$ (resp. $\mu> L$), then $E_\mu$ equals to $E_{\rm good}(\Delta)$ (resp. $E_{\rm max}$). 
By construction, ${\tbf E}$ is non decreasing in $\mu$ and is degraded w.r.t the fixed point profile ${\tbf E} \succeq {\tbf E}^{*}$.}
\label{fig:errorProfile}
\end{figure}
\begin{definition}[Shift operator]
The shift operator ${\rm S}$ is defined componentwise as $[{\rm S}({\tbf E})]_\mu \defeq  E_{\mu-1}$.
\end{definition}
\subsubsection{Upper bound on the potential variation under a shift}
The first step in the proof of threshold saturation is based on the Taylor expansion of the RS free energy of the SC system. 
\begin{lemma}\label{leminterp}
Let ${\tbf E}$ be a saturated profile. 
Set ${\tbf E}_\lambda \defeq (1-\lambda) {\tbf E} + \lambda {\rm S}({\tbf E})$ for $\lambda\in [0,1]$ and $\delta {E}_{\mu} \defeq {E}_{\mu} -{E}_{\mu-1}$. There exists some $\lambda\in[0,1]$ such that
\begin{align}
f_{\rm RS}^{\rm{c}}({\rm S}({\tbf E})) &-f_{\rm RS}^{\rm{c}}({\tbf E}) = \frac{1}{2} \sum_{\mu,\mu'=0}^{L} \delta {E}_\mu \delta {E}_{\mu'} \frac{\partial^2 f_{\rm RS}^{\rm{c}}({\tbf E})}{\partial  E_{\mu}\partial  E_{\mu'}}\Big|_{{\tbf E}_\lambda}. \label{eq:lemma_secondDerivativeFsc}
\end{align}
\label{lemma:Fdiff_quadraticForm}
\end{lemma}
\begin{proof}
Using the remainder Theorem, the free energy difference can be expressed as
\begin{align}
f_{\rm RS}^{\rm{c}}({\rm S}({\tbf E})) -f_{\rm RS}^{\rm{c}}({\tbf E}) 
= 
- \sum_{\mu=0}^{L} \delta {E}_{\mu} \frac{\partial f_{\rm RS}^{\rm{c}}({\tbf E})}{\partial  E_\mu}\Big|_{{\tbf E}}
+
\frac{1}{2}\sum_{\mu,\mu'=0}^{L} \delta {E}_{\mu} \delta {E}_{\mu'}  
\frac{\partial^2 f_{\rm RS}^{\rm{c}}({\tbf E})}{\partial E_{\mu}\partial  E_{\mu'}}\Big|_{ {\tbf E}_\lambda} 
. \label{eq:expFbsminusFb}
\end{align}
for some $\lambda\in [0,1]$.
By definition of the saturated profile $\tbf E$, we have $\delta E_\mu=0 \ \forall \ \mu \in\mathcal{A}\defeq\{0:\mu_*\}\cup \{\mu_{\rm max}+1:L\}$ 
and $E_\mu=[T_{\rm c}(\tbf E)]_\mu$ for $r\notin\mathcal{A}$. Recalling Lemma~\ref{lemma:fixedpointSE_extPot} we see that the
derivative in the first sum cancels for $r\notin\mathcal{A}$. Hence the first sum in (\ref{eq:expFbsminusFb}) vanishes.
\end{proof}
\begin{lemma} \label{lemma:smooth}
The saturated profile ${\tbf E}$ is smooth, i.e. $\delta E^* \defeq \underset{\mu}{\rm{max}} \,|\delta E_\mu| = \mathcal{O}(1/w)$ uniformly in $L$.
\end{lemma}
\begin{proof}
By definition of the saturated profile $\tbf E$ we have
$\delta E_\mu=0 \ \forall \ \mu \in\mathcal{A}\defeq\{0:\mu_*\}\cup \{\mu_{\rm max}+1:L\}$. For $\mu \notin\mathcal{A}$, we can replace 
the fixed point profile component $E_\mu$ by $[T_{\rm c}(\tbf E)]_\mu$ so that $\delta E_\mu = [T_{\rm c}(\tbf E)]_\mu - [T_{\rm c}(\tbf E)]_{\mu-1}$.
We will Taylor expand the SE operator. To this end, we define $\delta \Sigma_\mu^{-2} \defeq \Sigma_{\mu}({\tbf E})^{-2} - \Sigma_{\mu-1}({\tbf E})^{-2}$ for 
$\mu \in\{\mu^* +1:\mu_{\rm max}\}$. Recall that $\Lambda_{\mu-1,\nu-1} = \Lambda_{\mu\nu}$, $\Lambda_{\mu\nu} \geq 0$ and $\Lambda^{*}\defeq\sup_{\mu,\nu}\Lambda_{\mu\nu} = \mathcal{O}(1/w)$. Thus from \eqref{eq:defSigma_mu} we get 
\begin{align}
|\delta \Sigma_\mu^{-2} | = & \frac{1}{\Delta} \Big|\sum_{\nu \in \mathcal{S}_\mu} \Lambda_{\mu\nu} (E_{\nu} -E_{\nu-1})\Big| 
\nonumber \\ &
\leq 
\frac{\Lambda^{*}}{\Delta} \sum_{\nu \in \mathcal{S}_\mu} (E_{\nu} -E_{\nu-1})
\nonumber \\ &
\leq 
\frac{2v \Lambda^{*}}{\Delta} = \mathcal{O}(\frac{1}{w})
\label{eq:boundSigma_minus2}
%
\end{align}
where we have used $E_{\nu} -E_{\nu-1} \geq 0$ to get  rid of the absolute value.
Note that the first and second derivatives of the SE operator \eqref{def:SEop_SC} w.r.t $\Sigma_\mu^{-2}$ 
are bounded as long as the five first moments of the posterior \eqref{def:posterior_coupled} exist and are 
bounded (which is true under our assumptions). Then by Taylor expansion at first 
order in $\delta \Sigma_\mu^{-2}$ and using the remainder theorem, we obtain
\begin{align} \label{eq:deltamstar}
|\delta E_{\mu}| &= |[T_{\rm{c}}({\tbf E})]_{\mu} - [T_{\rm{c}}({\tbf E})]_{\mu-1}| 
\le |\delta \Sigma_\mu^{-2}| \Big|\frac{\partial [T_{\rm{c}}({\tbf E})]_{\mu}}{\partial \Sigma_\mu^{-2}}\Big| + \mathcal{O}(\delta \Sigma_\mu^{-4}) \le \mathcal{O}(\frac{1}{w}),
\end{align}
where the last inequality follows from \eqref{eq:boundSigma_minus2}.
\end{proof}

\begin{proposition} \label{lemma:dF_ge_Uoverw}
Let $\tbf E$ be a saturated profile. Then for all $\Delta > 0$ there exists a constant $0 < C(\Delta) < +\infty$ independent of $L$ such that 
\begin{align}
|f^{\rm c}_{\rm RS}({\rm S}({\tbf E})) -f^{\rm c}_{\rm RS}({\tbf E})| \leq \frac{C(\Delta)}{w} .
\end{align}
\end{proposition}
\begin{proof}
From Lemma~\ref{leminterp}, in order to compute the free energy difference between the shifted and non-shifted profiles, we need to compute the Hessian associated with this free energy. We have
\begin{align}
\frac{\partial f_{\rm RS}^{\rm{c}}({\tbf E})}{\partial {E_\mu}} &= \sum_{\nu \in \mathcal{S}_\mu} \frac{\Lambda_{\mu, \nu} E_\nu}{2\Delta} -\frac{1}{2} \sum_\nu \frac{\partial \Sigma_\nu(\tbf E)^{-2}}{\partial E_\mu} [T_{\rm{c}}({\tbf E})]_\nu - \frac{v}{2\Delta} 
\nonumber \\
&= \frac{1}{2\Delta}\Big(\sum_{\nu \in \mathcal{S}_\mu} \Lambda_{\mu, \nu}E_\nu+\sum_{\nu \in\mathcal{S}_\mu}\Lambda_{\mu, \nu} [T_{\rm{c}}({\tbf E})]_\nu-v\Big)
\end{align}
and 
\begin{align}
\frac{\partial^2 f_{\rm RS}^{\rm{c}}({\tbf E})}{\partial {E_\mu}\partial {E_{\mu'}}} &= \frac{1}{2\Delta}\Big( \Lambda_{\mu,\mu'}\mathbb{I}(\mu' \in \mathcal{S}_\mu) - \frac{1}{\Delta}\sum_{\nu \in \mathcal{S}_\mu\cap \mathcal{S}_{\mu'}} \Lambda_{\mu ,\nu}\Lambda_{\mu',\nu}\frac{\partial [T_{\rm{c}}({\tbf E})]_\nu }{\partial \, \Sigma_\nu^{-2}} \Big). 
\label{eq:hessian_explicit}
\end{align}
We can now estimate the sum in the Lemma~\ref{leminterp}. The contribution of the first term on the r.h.s of \eqref{eq:hessian_explicit} can be bounded as
\begin{align}
\frac{1}{2\Delta} \Big|\sum_{\mu=0}^L\sum_{\mu'\in \mathcal{S}_\mu}\Lambda_{\mu,\mu'}\delta E_{\mu}\delta E_{\mu'}\Big| \le \frac{\delta E^*\Lambda^*(2w+1)}{2\Delta} \Big|\sum_{\mu=0}^L\delta E_{\mu}\Big| \le \mathcal{O}(\frac{1}{w}), \label{eq:firstTermbound}
\end{align}
where we used the facts: $\delta E_\mu \geq 0$, the sum over $\mu=0, \cdots, L$ is telescopic, $E_\mu\in[0,v]$, $\Lambda^* = \mathcal{O}(1/w)$ and $\delta E^{*} = \mathcal(w^{-1})$ (Lemma~\ref{lemma:smooth}). We now bound the contribution of the second term on the r.h.s of \eqref{eq:hessian_explicit}. Recall the first derivative w.r.t $\Sigma_{\nu}^{-2}$ of the SE operator is bounded uniformly in $L$. Call this bound $K=\mathcal{O}(1)$. We obtain
\begin{align}
&\frac{1}{2\Delta^2} \Big|\sum_{\mu, \mu'=1}^L \delta E_\mu \delta E_{\mu '} \sum_{\nu \in \mathcal{S}_\mu\cap \mathcal{S}_{\mu'}}\Lambda_{\mu ,\nu}\Lambda_{\mu',\nu} \frac{\partial [T_{\rm{c}}({\tbf E})]_\nu }{\partial \Sigma_{\nu}^{-2}}\Big| \nonumber\\ 
\le \ &\frac{K\Lambda^{*2}\delta E^*}{2\Delta^2} \Big|\sum_{\mu=1}^L \delta E_\mu \sum_{\mu'\in \{\mu - 2w : \mu + 2w\}} {\rm card}(\mathcal{S}_\mu\cap \mathcal{S}_{\mu'})\Big| \le \mathcal{O}(\frac{1}{w}). \label{eq:last_secondTerm_boudHessian}
\end{align}
The last inequality follows from the following facts: the sum over $\mu=1, \cdots, L$ is telescopic, $\Lambda^* = \mathcal{O}(1/w)$, Lemma~\ref{lemma:smooth}, and for any fixed $\mu$ the following holds
\be
\sum_{\mu'\in \{\mu - 2w : \mu + 2w\}} {\rm card}(\mathcal{S}_\mu\cap \mathcal{S}_{\mu'}) = (2w+1)^2.
\ee
Finally, from \eqref{eq:firstTermbound}, \eqref{eq:last_secondTerm_boudHessian} and the triangle inequality we obtain
\begin{align}
\frac{1}{2}\Big| \sum_{\mu,\mu'=1}^{L} \delta {E}_\mu \delta {E}_{\mu'} \frac{\partial^2 f^{\rm c}_{\rm RS}({\tbf E})}{\partial  E_{\mu}\partial  E_{\mu'}}\Big|_{{\tbf E}_\lambda}\Big| = \mathcal{O}(\frac{1}{w})
\end{align}
uniformly in $L$. Combining this result with Lemma~\ref{leminterp} ends the proof.
\end{proof}
\subsubsection{Lower bound on the potential variation under a shift}
The second step in the proof is based on a direct evaluation of $f_{\rm RS}^{\rm c}({\rm S}(E)) - f_{\rm RS}^{\rm c}(E)$.
We first need the Lemma:
\begin{lemma}
Let ${\tbf E}$ be a saturated profile such that ${\tbf E}\succ {\tbf E}_{\rm good}(\Delta)$. Then $ E_{\rm max} \notin \mathcal{V}_{\rm good}$.
\label{lemma:outside_basin}
\end{lemma}
\begin{proof}
The fact that the error profile is non decreasing and the assumption that ${\tbf E} \succ {\tbf E}_{\rm good}(\Delta)$ imply that $E_{\rm max} > E_0 =E_{\rm good}(\Delta)$. Moreover,
$E_{\rm max} \le [T_{\rm c}({\tbf E})]_{\mu_{\rm max}} \le T_{\rm u}(E_{\rm max})$
where the first inequality follows from $\tbf E\succ \tbf E^*$ and the monotonicity of $T_{\rm c}$, while the second comes from the fact that $\tbf E$ is non decreasing. Combining these with the 
monotonicity of $T_{\rm u}$ gives $T_{\rm u}(E_{\rm max}) \ge  E_{\rm max}$ which implies $T_{\rm u}^{(\infty)}( E_{\rm max}) \ge  E_{\rm max}> E_{\rm good}(\Delta)$ which means $E_{\rm max} \notin \mathcal {V}_{\rm good}$.
\end{proof}
\begin{proposition}\label{lemma:FcSE_minus_FcE_le_deltaF}
Fix $\Delta < \Delta_{\rm RS}$ and let ${\tbf E}$ be a saturated profile such that ${\tbf E}\succ {\tbf E}_{\rm good}(\Delta)$. Then 
\begin{align}
|f^{\rm c}_{\rm RS}(S(\tbf{E})) - f^{\rm c}_{\rm RS}(\tbf{E})| \ge \delta f^{\rm u}_{\rm RS}(\Delta)
\end{align}
where $\delta f^{\rm u}_{\rm RS}(\Delta)$ is the potential gap (Definition~\ref{def:fgap}).
\end{proposition}
\begin{proof}
Set 
\begin{align*}
\mathcal{I}(\Sigma) \defeq \mathbb{E}_{S,Z}\Big[\ln\Big(\int dxP_0(x) e^{- \frac{1}{2\Sigma^2}\big(x^2 -2 Sx - 2\Sigma Zx\big)}\Big)\Big]\, .
\end{align*}
By \eqref{def:RSf_SC} 
\begin{align}
f^{\rm c}_{\rm RS}({\rm S}(\tbf{E})) - f^{\rm c}_{\rm RS}(\tbf{E})=&\sum_{\mu = -1}^{L-1} \sum_{\nu \in \mathcal{S}_{\mu}} \frac{\Lambda_{\mu+1\nu+1} }{4\Delta}(v\!-\!E_{\mu}) (v\!-\!E_{\nu}) - \sum_{\mu = 0}^{L} \sum_{\nu \in \mathcal{S}_{\mu}} \frac{\Lambda_{\mu\nu} }{4\Delta}(v\!-\!E_{\mu}) (v\!-\!E_{\nu}) \nonumber \\
& \quad - \sum_{\mu=0}^L \mathcal{I}(\Sigma_{\mu}({\rm S}({\bf E}))) + \sum_{\mu=0}^L \mathcal{I}(\Sigma_{\mu}(\bf E))\nonumber\\
= &\sum_{\nu \in \mathcal{S}_{-1}} \frac{\Lambda_{\mu\nu} }{4\Delta}(v\!-\!E_{\mu}) (v\!-\!E_{\nu}) - \sum_{\nu \in \mathcal{S}_{L}} \frac{\Lambda_{\mu\nu} }{4\Delta}(v\!-\!E_{\mu}) (v\!-\!E_{\nu}) \nonumber \\
& \quad - \mathcal{I}(\Sigma_{-1}({\bf E})) + \mathcal{I}(\Sigma_{L}(\bf E)), \label{eq:lastFcSE_minus_FcE}
\end{align}
where we used $\Lambda_{\mu+1\nu+1}=\Lambda_{\mu\nu}$ implying also $\Sigma_{\mu}({\rm S}({\bf E})) = \Sigma_{\mu-1}({\bf E})$ as seen from \eqref{eq:defSigma_mu}. Recall $\Sigma(E)^{-2}= (v - E)/\Delta$. Now looking at \eqref{eq:defSigma_mu}, one notices that thanks to the saturation of ${\tbf E}$, $\Sigma_{-1}({\tbf E}) = \Sigma(E_0)$ where $E_0=E_{\rm good}(\Delta)$ (see the green branch in Fig.~\ref{fig:errorProfile}), while $\Sigma_{L}({\tbf E}) = \Sigma(E_L)$ where $E_L = E_{\rm{max}}$ (see the red branch Fig.~\ref{fig:errorProfile}). Finally from \eqref{eq:lastFcSE_minus_FcE}, using that the coupling matrix is (doubly) stochastic and the saturation of ${\bf E}$
\begin{align}
f^{\rm c}_{\rm RS}({\rm S}(\tbf{E})) - f^{\rm c}_{\rm RS}(\tbf{E}) &= \Big[\frac{(v\!-\!E_{0})^2 }{4\Delta} - \mathcal{I}(\Sigma(E_{\rm good}(\Delta))) \Big] - \Big[\frac{(v\!-\!E_{L})^2 }{4\Delta} - \mathcal{I}(\Sigma(E_L)) \Big] \nonumber\\
&= f^{\rm u}_{\rm RS}(E_{\rm good}) - f^{\rm u}_{\rm RS}(E_{\rm max}) \le -\delta f^{\rm u}_{\rm RS}(\Delta),
\end{align}
where we recognized the potential function of the underlying system $f_{\rm RS}^{\rm u}(E;\Delta) = i_{\rm RS}(E;\Delta) - \frac{v^2}{4\Delta}$ and the last inequality is a direct application of Lemma~\ref{lemma:outside_basin} and Definition \ref{def:fgap}. Finally, using the positivity of $\delta f^{\rm u}_{\rm RS}(\Delta)$ for $\Delta<\Delta_{\rm{RS}}$, we obtain the desired result.
\end{proof}

\subsubsection{End of proof of threshold saturation}\label{endsat}

We now have the necessary ingredients in order to prove threshold saturation.
\begin{theorem}[Asymptotic performance of AMP for the coupled system] \label{th:mainth_thsat}
Fix $\Delta<\Delta_{\rm{RS}}$. Take a spatially coupled system with $w>C(\Delta)/\delta f^{\rm u}_{\rm RS}(\Delta)$ where $C(\Delta)$ is the constant in 
Proposition \ref{lemma:dF_ge_Uoverw}. Then any fixed point profile ${\tbf{E}^*}$ of the coupled state evolution iteration \eqref{equ:statevolutioncoupled} must satisfy ${\tbf{E}}^* \prec {\tbf{E}}_{\rm good}(\Delta)$.
\end{theorem}
\begin{proof}
The proof is by contradiction. Fix $\Delta < \Delta_{\rm RS}$ and $w \ge C(\Delta)/\delta f^{\rm u}_{\rm RS}(\Delta)$.
We assume there exists a fixed point profile which does not satisfy $\tbf E^*\prec {\tbf{E}}_{\rm good}(\Delta)$. Then we construct
the associated saturated profile ${\tbf E}$. This profile satisfies both statements 
of Propositions ~\ref{lemma:dF_ge_Uoverw} and \ref{lemma:FcSE_minus_FcE_le_deltaF}. Therefore we must have
$\delta f^{\rm u}_{\rm RS}(\Delta)\le C(\Delta)/w$ which contradicts the choice  $w > C(\Delta)/\delta f^{\rm u}_{\rm RS}(\Delta)$.
We conclude that ${\tbf E}^* \prec {\tbf{E}}_{\rm good}(\Delta)$ must be true.
\end{proof}
%
%

Theorem \ref{thres-sat-lemma} is a direct corollary of Theorem~\ref{th:mainth_thsat} and Definition~\ref{def:AMPcoupled}. Take 
some $\Delta_* < \Delta_{\rm RS}$ and choose $w > C(\Delta_*)/\delta f^{\rm u}_{\rm RS}(\Delta_*)$. Then we have 
$\Delta_{{\rm AMP}, w, L} \geq \Delta_*$. Note that $\delta f^{\rm u}_{\rm RS}(\Delta_*) \to 0_+$ for $\Delta_*\to \Delta_{\rm RS}$. Thus  
Taking $L\to +\infty$ first and $w\to +\infty$ second we can make $\Delta_*$ as close to $\Delta_{\rm RS}$ as we wish. Therefore we obtain $\Delta_{\rm AMP}^{\rm c} \defeq\liminf_{L, w\to +\infty} \Delta_{{\rm AMP}, w, L} \geq \Delta_{\rm RS}$ where the limit is taken in the specified order.

\section{Invariance of the mutual information under spatial coupling}\label{sec:subadditivitystyle}
In this section we prove that the mutual information remains unchanged under spatial coupling in a suitable asymptotic limit 
(Theorem \ref{LemmaGuerraSubadditivityStyle}). 
We will compare the mutual informations of the four following variants of \eqref{eq:def_coupledSyst}. In each case, the signal $\bs$ has $n(L+1)$ i.i.d components.
\begin{itemize}
\item \emph{The fully connected:} If we choose $w = L/2$ and a homogeneous coupling matrix with elements $\Lambda_{\mu,\nu} = (L+1)^{-1}$ in \eqref{eq:def_coupledSyst}. This yields a homogeneous fully connected system equivalent to \eqref{eq:mainProblem} with $n(L+1)$ instead of $n$ variables. The associated mutual information per variable for fixed $L$ and $n$ is denoted by $i_{n, L}^{\rm con}$.
\item \emph{The SC pinned system:} This is the system studied in Section~\ref{sec:thresh-sat} to prove threshold saturation, with the pinning condition. In this case we choose $0<w < L/2$. The coupling matrix $\boldsymbol{\Lambda}$ is any matrix that fulfills the requirements in Section~\ref{subsec:SC} (the concrete example given there will do). The associated mutual information per variable is here denoted $i_{n, w, L}^{\rm cou}$. Note that $i_{n, w, L}^{\rm cou} = (n(L+1))^{-1} I_{w,L}(\tbf{S}; \bW)$
\item \emph{The periodic SC system:} This is the same SC system (with same coupling window and coupling matrix) but without the pinning condition. The associated mutual information per variable at fixed $L,w, n$ is denoted $i_{n, w, L}^{\rm per}$.
\item \emph{The decoupled system:} This corresponds simply to $L+1$ identical and independent systems of the form \eqref{eq:mainProblem} with $n$ variables each. This is equivalent to periodic SC system with $w = 0$. The associated mutual information per variable is  denoted $i_{n, L}^{\rm dec}$. Note that $i_{n, L}^{\rm dec} = n^{-1} I(\tbf{S}; \bW)$.
\end{itemize}
Let us outline the proof strategy.
In a first step, we use an interpolation method twice: first interpolating between the fully connected and periodic SC systems, and then between the decoupled and periodic SC systems. This will allow to sandwich the mutual information of the periodic SC system by those of the fully connected and decoupled systems respectively (see Lemma~\ref{lemma:freeEnergy_sandwich}). 
In the second step, using again a similar interpolation and Fekete's theorem for superadditive sequences, we prove that the decoupled and fully connected systems have asymptotically the same mutual information (see Lemma \ref{lemma:superadditivity} for the existence of the limit). From these results we deduce the proposition:
\begin{proposition}\label{lemma:freeEnergy_sandwich_limit} 
For any $0\leq w \leq L/2$ 
\begin{align}
\lim_{n\to +\infty} i_{n, w, L}^{\rm per}= \lim_{n\to +\infty}\frac{1}{n}I(\tbf{S}; \bW)
\end{align}
\end{proposition}
\begin{proof}
Lemma \ref{lemma:superadditivity} implies that $\lim_{n\to +\infty} i_{n, L}^{\rm con} = \lim_{n\to +\infty} i_{n, L}^{\rm dec}$. One also notes that $i_{n, L}^{\rm dec} = \frac{1}{n} I(\tbf{S}; \bW)$. Thus the result follows from Lemma \ref{lemma:freeEnergy_sandwich}.
\end{proof}
In a third step an easy argument shows
\begin{proposition}\label{lemma:openVSclosed} 
Assume $P_0$ has finite first four moments. For any $0\leq w \leq L/2$ 
\begin{align}
\vert i_{n, w, L}^{\rm per} - i_{n, w, L}^{\rm cou} \vert = \mathcal{O}(\frac{w}{L})
\end{align}
\end{proposition}
\begin{proof}
See Appendix \ref{appendix-pinfree}.
\end{proof}
Since $i_{n, w, L}^{\rm cou} = (n(L+1))^{-1} I_{w,L}(\tbf{S}; \bW)$, Theorem \ref{LemmaGuerraSubadditivityStyle} is an immediate consequence of Propositions \ref{lemma:freeEnergy_sandwich_limit} and \ref{lemma:openVSclosed}. 
\subsection{A generic interpolation}\label{generic}
Let us consider two systems of same total size $n(L+1)$ with coupling matrices $\bold{\Lambda}^{(1)}$ and $\bold{\Lambda}^{(0)}$ supported on coupling windows $w_1$ and $w_0$ respectively. 
Moreover, we assume that the observations associated with the first system are corrupted by an AWGN equals to $\sqrt{\Delta / t}\bz$ while the 
AWGN corrupting the second system is $\sqrt{\Delta / (1 - t)}\bz^{\prime}$, where $Z_{ij}$ and $Z^{\prime}_{ij}$ are two i.i.d. 
standard Gaussians and $t\in[0,1]$ is the \emph{interpolation parameter}. The interpolating inference problem has the form
\begin{align}
 \begin{cases}
  w_{i_\mu j_\nu} & = s_{i_\mu} s_{j_\nu}\sqrt{\frac{\Lambda_{\mu\nu}^{(1)}}{n}} + z_{i_\mu j_\nu}\sqrt{\frac{\Delta}{t}},\\
  w_{i_\mu j_\nu} & = s_{i_\mu} s_{j_\nu}\sqrt{\frac{\Lambda_{\mu\nu}^{(0)}}{n}} + z_{i_\mu j_\nu}^\prime\sqrt{\frac{\Delta}{1-t}}
 \end{cases}
\end{align}
In this setting, at $t=1$ the interpolated system corresponds to the first system as the 
noise is infinitely large in the second one and no information is available about it, while at $t=0$ the opposite happens. 
The associated interpolating posterior distribution can be expressed
as 
\begin{align}\label{post-interp-d}
 P_{t}(\bx\vert \bs, \bz, \bz^\prime) = \frac{1}{\mathcal{Z}_{\rm int}(t)} e^{-\mathcal{H}(t, \bold{\Lambda}^{(1)},\bold{\Lambda}^{(0)}))}\prod_{\mu=0}^{L}\prod_{i_\mu=1}^n P_0(x_{i_\mu})
\end{align}
where the ``Hamiltonian'' is
 $\mathcal{H}_{\rm int}(t,\bold{\Lambda}^{(1)},\bold{\Lambda}^{(0)}) \defeq 
\mathcal{H}(t,\bold{\Lambda}^{(1)}) + \mathcal{H}(1-t,\bold{\Lambda}^{(0)})$ with\footnote{Note that since the SC system is defined on a ring, we can express the Hamiltonian in terms of forward coupling only.} 
\begin{align}
\mathcal{H}(t,\bold{\Lambda}) \defeq &\frac{t}{\Delta} \sum_{\mu=0}^L \Lambda_{\mu,\mu} \sum_{i_\mu\le j_\mu}\bigg( \frac{x_{i_\mu}^2x_{j_\mu}^2}{2n} -  \frac{s_{i_\mu}s_{j_\mu}x_{i_\mu}x_{j_\mu}}{n} - \frac{x_{i_\mu}x_{j_\mu}z_{i_\mu j_\mu} \sqrt{\Delta}}{\sqrt{n t\Lambda_{\mu,\mu}}}\bigg)\nonumber\\
+&\frac{t}{\Delta} \sum_{\mu=0}^L\sum_{\nu=\mu+1}^{\mu+w}\Lambda_{\mu,\nu}\sum_{i_\mu,j_\nu=1}^{n}\bigg( \frac{x_{i_\mu}^2x_{j_\nu}^2}{2n} -  \frac{s_{i_\mu}s_{j_\nu}x_{i_\mu}x_{j_\nu}}{n} - \frac{x_{i_\mu}x_{j_\nu}z_{i_\mu j_\nu} \sqrt{\Delta}}{\sqrt{nt\Lambda_{\mu,\nu}}}\bigg).
\end{align}
and $\mathcal{Z}_{\rm int}(t)$ is the obvious normalizing factor, the ``partition function''. The posterior average with respect to \eqref{post-interp-d} is denoted by the bracket notation $\langle - \rangle_t$. It is easy to see
that the mutual information per variable (for the interpolating inference problem) can be expressed as 
\begin{align}
i_{\rm int}(t) \defeq - \frac{1}{n(L+1)}\mathbb{E}_{{\bf S}, {\bf Z}, {\bf Z}^{\prime}}[\ln \mathcal{Z}_{\rm int}(t)] + \frac{v^2}{4\Delta} + \frac{1}{4 \Delta n(L+1)} (2\mathbb{E}[S^4] - v^2)
\end{align}

The aim of the interpolation method in the present context is to compare the mutual informations of the systems at $t=1$ and $t=0$. To do so, 
one uses the fundamental theorem of calculus
\begin{align}\label{calc}
i_{\rm int}(1) - i_{\rm int}(0) = \int_{0}^1 dt \, \frac{di_{\rm int}(t)}{dt}\,.
\end{align}
and tries to determine the sign of the integral term. 

We first prove that
\begin{align} 
& \frac{di_{\rm int}(t)}{dt} = \nonumber\\
&\frac{1}{4\Delta(L+1)} \mathbb{E}_{\bf S, \bf Z,{\bf Z}^{\prime}}\Big[\Big\la 
-\frac{1}{n^2}\Big(\sum_{\mu=0}^L\sum_{\nu=\mu-w_1}^{\mu + w_1} \Lambda_{\mu\nu}^{(1)} 
\sum_{i_\mu, j_\nu=1}^n X_{i_\mu}X_{j_\nu}S_{i_\mu}S_{j_\nu} + \sum_{\mu=0}^L \Lambda_{\mu\mu}^{(1)} \sum_{i_\mu=1}^n X_{i_\mu}^2 S_{i_\mu}^2\Big) \nonumber \\
&+ \frac{1}{n^2}\Big(\sum_{\mu=0}^L\sum_{\nu = \mu-w_0}^{\mu+w_0} \Lambda_{\mu\nu}^{(0)} \sum_{i_\mu, j_\nu=1}^n X_{i_\mu}X_{j_\nu}S_{i_\mu}S_{j_\nu} 
+ \sum_{\mu=0}^L \Lambda_{\mu\mu}^{(0)} \sum_{i_\mu=1}^n X_{i_\mu}^2 S_{i_\mu}^2 \Big) \Big\ra_{t} \Big],
\label{eq:derivative_interpolation}
\end{align}  
where $\la - \ra_t$ denotes the expectation over the posterior distribution associated with the interpolated Hamiltonian $\mathcal{H}_{\rm int}(t,\bold{\Lambda}^{(1)},\bold{\Lambda}^{(0)})$.
We start with a simple differentiation of the Hamiltonian w.r.t. $t$ which yields
\begin{align*}
\frac{d}{dt}\mathcal{H}_{\rm int}(t,\bold{\Lambda}^{(1)},\bold{\Lambda}^{(0)}) = \frac{1}{\Delta} \big( \mathcal{A}(t,\bold{\Lambda^{(1)}}) - \mathcal{B}(t,\bold{\Lambda}^{(0)}) \big),
\end{align*}
where
\begin{align*}
\mathcal{A}(t,\bold{\Lambda}^{(1)}) = & \sum_{\mu=0}^L \Lambda_{\mu\mu}^{(1)} \sum_{i_\mu\le j_\mu}\bigg( \frac{x_{i_\mu}^2x_{j_\mu}^2}{2n} 
-  \frac{s_{i_\mu}s_{j_\mu}x_{i_\mu}x_{j_\mu}}{n} - \frac{x_{i_\mu}x_{j_\mu}z_{i_\mu j_\mu} \sqrt{\Delta}}{2\sqrt{n t\Lambda_{\mu\mu}^{(1)}}}\bigg)\nonumber\\
+& \sum_{\mu=0}^L\sum_{\nu=\mu+1}^{\mu+w_1}\Lambda_{\mu\nu}^{(1)}\sum_{i_\mu,j_\nu=1}^{n}\bigg( \frac{x_{i_\mu}^2x_{j_\nu}^2}{2n} 
-  \frac{s_{i_\mu}s_{j_\nu}x_{i_\mu}x_{j_\nu}}{n} - \frac{x_{i_\mu}x_{j_\nu}z_{i_\mu j_\nu} \sqrt{\Delta}}{2\sqrt{nt\Lambda_{\mu\nu}^{(1)}}}\bigg)\nonumber \\
\mathcal{B}(t,\bold{\Lambda}^{(0)}) = &\sum_{\mu=0}^L \Lambda^{(0)}_{\mu\mu} \sum_{i_\mu\le j_\mu}\bigg( \frac{x_{i_\mu}^2x_{j_\mu}^2}{2n} -  \frac{s_{i_\mu}s_{j_\mu}x_{i_\mu}x_{j_\mu}}{n} - \frac{x_{i_\mu}x_{j_\mu}z^{\prime}_{i_\mu j_\mu} \sqrt{\Delta}}{2\sqrt{n (1-t)\Lambda^{(0)}_{\mu\mu}}}\bigg)\nonumber\\
& \sum_{\mu=0}^L\sum_{\nu=\mu+1}^{\mu+w_0}\!\Lambda^{(0)}_{\mu\nu}\!\sum_{i_\mu,j_\nu=1}^{n}\!\bigg( \frac{x_{i_\mu}^2x_{j_\nu}^2}{2n} -  \frac{s_{i_\mu}s_{j_\nu}x_{i_\mu}x_{j_\nu}}{n} - \frac{x_{i_\mu}x_{j_\nu}z^{\prime}_{i_\mu j_\nu} \sqrt{\Delta}}{2\sqrt{n(1-t)\Lambda^{(0)}_{\mu\nu}}}\bigg) .
\end{align*}
Using integration by parts with respect to the Gaussian variables $Z_{ij}$, $Z_{ij}^\prime$, one gets
\begin{align}
&\mathbb{E}_{\bf S, \bf Z, \bf Z^{\prime}}[Z_{i_\mu j_\nu} \langle X_{i_\mu} X_{j_\nu}  \rangle_{t}] =  \sqrt{\frac{t\Lambda_{\mu,\nu}}{n\Delta}}  \mathbb{E}_{\bf S, \bf Z, \bf Z^{\prime}}\Big[ \langle X_{i_\mu}^2 X_{j_\nu}^2  \rangle_{t} - \langle X_{i_\mu} X_{j_\nu}  \rangle_{t}^2 \Big]
\label{intfirst}
\\
&\mathbb{E}_{\bf S, \bf Z, \bf Z^{\prime}}[Z^{\prime}_{i_\mu j_\nu} \langle X_{i_\mu} X_{j_\nu}  \rangle_{t}] =  \sqrt{\frac{(1-t)\Lambda^0_{\mu,\nu}}{n\Delta}}  \mathbb{E}_{\bf S, \bf Z, \bf Z^{\prime}}\Big[ \langle X_{i_\mu}^2 X_{j_\nu}^2  \rangle_{t} - \langle X_{i_\mu} X_{j_\nu}  \rangle_{t}^2 \Big].
\label{eq:derivative_freeEnergy_integrationByPart}
\end{align}
Moreover an application of the Nishimori identity \eqref{eq:nishCond} shows
\begin{align}\label{eq:derivative_freeEnergy_nishimori}
\mathbb{E}_{\bf S, \bf Z, \bf Z^{\prime}}[ \langle X_{i_\mu} X_{j_\nu}  \rangle_{t}^2 ] = \mathbb{E}_{\bf S, \bf Z, \bf Z^{\prime}}[ \langle X_{i_\mu} X_{j_\nu}  S_{i_\mu} S_{j_\nu} \rangle_{t} ].
\end{align}
Combining \eqref{eq:derivative_interpolation}-\eqref{eq:derivative_freeEnergy_nishimori} and using the fact that the SC system defined on a ring satisfies
\begin{align*}
&\sum_{\mu=0}^L \Lambda_{\mu\mu} \sum_{i_\mu\le j_\mu} x_{i_\mu} x_{j_\mu}s_{i_\mu} s_{j_\mu} + \sum_{\mu=0}^L\sum_{\nu=\mu+1}^{\mu+w}\Lambda_{\mu\nu}\sum_{i_\mu,j_\nu=1}^{n} x_{i_\mu} x_{j_\nu} s_{i_\mu} s_{j_\nu} = \\
&\frac{1}{2} \sum_{\mu=0}^L\sum_{\nu=\mu-w}^{\mu+w}\Lambda_{\mu\nu}\sum_{i_\mu,j_\nu=1}^{n} x_{i_\mu} x_{j_\nu} s_{i_\mu} s_{j_\nu} + \frac{1}{2} \sum_{\mu=0}^L \Lambda_{\mu\mu} x_{i_\mu}^2 s_{i_\mu}^2,
\end{align*}
we obtain \eqref{eq:derivative_interpolation}.

Now, define the {\it overlaps} associated to each block $\mu$ as 
\begin{align}
q_\mu \defeq \frac{1}{n} \sum_{i_\mu=1}^n X_{i_\mu}S_{i_\mu}, \qquad \tilde{q}_\mu \defeq \frac{1}{n} \sum_{i_\mu=1}^n X_{i_\mu}^2 S_{i_\mu}^2.
\end{align}
Hence, (\ref{eq:derivative_interpolation}) can be rewritten as
\begin{align} \label{eq:derivative_interpolation_overlap}
 \frac{di_{\rm int}(t)}{dt} = \frac{1}{4\Delta(L+1)} \mathbb{E}_{{\bf S}, {\bf Z},{\bf Z}^{\prime}} \Big[\Big\la\mathbf{q}^{\intercal} \mathbf{\Lambda}^{(0)} \, \mathbf{q} -\mathbf{q}^{\intercal} 
 \mathbf{\Lambda}^{(1)} \, \mathbf{q} + \frac{1}{n} \, \Big(\tilde{\mathbf{q}}^{\intercal} {\rm diag} ({\mathbf{\Lambda}}^{(0)}) - \tilde{\mathbf{q}}^{\intercal} {\rm diag} ({\mathbf{\Lambda}^{(1)}}) \Big) \Big\ra_t \Big],
\end{align} 
where $\mathbf{q}^\intercal = [q_0\cdots q_L]$, $\tilde{\mathbf{q}}^\intercal = [\tilde{q}_0\cdots \tilde{q}_L]$ are row vectors and 
${\rm diag} ({\mathbf{\Lambda}})$ represents the column vector with entries $\{\Lambda_{\mu\mu}\}_{\mu=0}^L$.
The coupling matrices $\mathbf{\Lambda}^{(1)}, \mathbf{\Lambda}^{(0)}$ are real, symmetric, circulant (due to the periodicity of the ring) and thus can be diagonalized in the same Fourier basis. 
%
%
We have
\begin{align} \label{eq:derivative_generic_Fourier}
\frac{di_{\rm int}(t)}{dt} = \frac{1}{4\Delta (L+1)} \mathbb{E}_{{\bf S}, {\bf Z},{\bf Z}^{\prime}} \Big[ \Big\la \hat{\mathbf{q}}^{\intercal} 
\big(\mathbf{D}^{(0)} - \mathbf{D}^{(1)}\big) \, \hat{\mathbf{q}} + \frac{1}{n} \, \Big(\tilde{\mathbf{q}}^{\intercal} {\rm diag} ({\mathbf{\Lambda}}^{(0)}) 
- \tilde{\mathbf{q}}^{\intercal} {\rm diag} ({\mathbf{\Lambda}^{(1)}}) \Big)\Big\ra_t \Big],
\end{align}
where $\widehat{\mathbf{q}}$ is the discrete Fourier transfrom of $\mathbf{q}$ and $\mathbf{D}^{(1)}, \mathbf{D}^{(0)}$ are 
the diagonal matrices with the eigenvalues of $\mathbf{\Lambda}^{(1)},\mathbf{\Lambda}^{(0)}$. Since the coupling matrices are stochastic with non-negative Fourier transform, their
largest eigenvalue equals $1$ (and is associated to the $0$-th Fourier mode) while the remaining eigenvalues are non-negative. These properties will be essential in the following paragraphs. 
\subsection{Applications}\label{appli}
Our first application is
\begin{lemma}\label{lemma:freeEnergy_sandwich} 
Let the coupling matrix $\Lambda$ verify the requirements (i)-(v) in Sec. \ref{subsec:SC}.
The mutual informations of the decoupled, periodic SC and fully connected systems verify 
\begin{align}
 i_{n, L}^{\rm{dec}} \le i_{n, w, L}^{\rm per} \le i_{n, L}^{\rm con}.
\end{align}
\end{lemma}
\begin{proof}
We start with the second inequality. We choose $\Lambda_{\mu\nu}^{(1)} = (L+1)^{-1}$ for the fully connected system at $t=1$. This matrix 
has a unique eigenvalue equal to $1$ and $L$ degenerate eigenvalues equal to $0$. Therefore it is clear 
that $\mathbf{D}^{(0)} - \mathbf{D}^{(1)}$ is positive semi-definite and 
$\hat{\mathbf{q}}^{\intercal} \big(\mathbf{D}^{(0)} - \mathbf{D}^{(1)}\big)\hat{\mathbf{q}}\geq 0$.
Moreover notice that $\Lambda_{\mu\mu}^{(0)} = \Lambda_{00}$ is independent of $L$. Therefore for $L$ large enough
\begin{align}
\tilde{\mathbf{q}}^{\intercal} {\rm diag} ({\mathbf{\Lambda}}^{(0)}) - \tilde{\mathbf{q}}^{\intercal} {\rm diag} ({\mathbf{\Lambda}^{(1)}}) = 
\Big(\Lambda_{00}-\frac{1}{L+1}\Big) \sum_{\mu=0}^L\tilde{q}_{\mu} \geq 0.
\end{align}
Therefore we conclude that \eqref{eq:derivative_generic_Fourier} is positive and from \eqref{calc} $i_{n,L}^{\rm con} - i_{n, w, L}^{\rm per} \geq 0$. 
For the first inequality we proceed similarly, but this time we choose $\Lambda_{\mu\nu}^{(1)} = \delta_{\mu\nu}$ for the decoupled system which has all eigenvalues equal to $1$. Therefore 
$\mathbf{D}^{(0)} - \mathbf{D}^{(1)}$ is negative semidefinite so $\hat{\mathbf{q}}^{\intercal} \big(\mathbf{D}^{(0)} - \mathbf{D}^{(1)}\big)\hat{\mathbf{q}}\leq 0$. Moreover this time
\begin{align}
\tilde{\mathbf{q}}^{\intercal} {\rm diag} ({\mathbf{\Lambda}}^{(0)}) - \tilde{\mathbf{q}}^{\intercal} {\rm diag} ({\mathbf{\Lambda}^{(1)}}) = 
\Big(\Lambda_{00}-1\Big) \sum_{\mu=0}^L\tilde{q}_{\mu} \leq 0
\end{align}
because we necessarily have $0\leq \Lambda_{00}^{(0)} \leq 1$. We conclude that \eqref{eq:derivative_generic_Fourier} is 
negative and from \eqref{calc} $i_{n, L}^{\rm{dec}} - i_{n, w, L}^{\rm per}\leq 0$.
\end{proof}
 
The second application is 

\begin{lemma}\label{lemma:superadditivity}
Consider the mutual information of system \eqref{eq:mainProblem} and set  $i_n = n^{-1}I(\tbf{S}; \bW)$. Consider also $i_{n_1}$ and $i_{n_2}$ the mutual informations 
of two systems of size $n_1$ and $n_2$ with $n= n_1 +n_2$.
The sequence $n i_n$ is superadditive in the sense that  
\begin{align}
 n_1 i_{n_1} + n_2 i_{n_2} \leq n i_n.
\end{align}
Fekete's lemma then implies that $\lim_{n\to +\infty} i_n$ exists. 
\end{lemma}
\begin{proof}
The proof is easily obtained by following the generic interpolation method of Sec. \ref{generic} for a coupled system with two spatial positions (i.e. $L+1=2$). We choose 
$\Lambda_{\mu\nu}^{(0)}=\delta_{\mu\nu}$, $\mu,\nu\in{0,1}$ for the ''decoupled`` system and $\Lambda_{\mu\nu}^{(1)}=1/2$ for $\mu,\nu\in{0,1}$ for the ''fully connected`` system.
This analysis is essentially identical to [\cite{guerraToninelli}] were the existence of the thermodynamic limit of the free energy for the Sherrington-Kirkpatrick mean field spin glass is proven. 
\end{proof}
\section{Proof of the replica symmetric formula (Theorem \ref{thm1})} \label{sec:proof_theorem}
In this section we provide the proof of the RS formula for the mutual information of the underlying model (Theorem \ref{thm1})
for $0 < \Delta \leq \Delta_{\rm opt}$ (Proposition \ref{cor:delta-low-regime}) and then for $\Delta \geq \Delta_{\rm opt}$ (Proposition \ref{for:above-delta-opt}). For $0 < \Delta  \leq \Delta_{\rm opt}$ the proof directly follows form the I-MMSE relation Lemma \ref{lemma:immse}, the replica bound \eqref{eq:guerrabound} and the suboptimality of the AMP algorithm. In this interval the proof doesn't require spatial coupling.
For $\Delta \geq  \Delta_{\rm opt}$ the proof uses the results of Sections \ref{sec:thresh-sat} and \ref{sec:subadditivitystyle} on the spatially coupled model. 

Let us start with two preliminary lemmas.
The first is an I-MMSE relation [\cite{GuoShamaiVerdu}] adapted to the current matrix estimation problem.
\begin{lemma}\label{lemma:immse}
Let $P_0$ has finite first four moments.
The mutual information and the matrix-MMSE are related by
\begin{align}\label{eq:GSV}
\frac{1}{n}\frac{d I(\tbf{S}; \tbf{W}) }{d\Delta^{-1}} 
= \frac{1}{4} {\rm Mmmse}_n(\Delta^{-1}) + \mathcal{O}(1/n).
\end{align}
\end{lemma}
\begin{proof}
\begin{align}
\frac{1}{n}\frac{d I(\tbf{S}; \tbf{W}) }{d\Delta^{-1}} &= 
\frac{1}{2n^2}\mathbb{E}_{\tbf{S}, \tbf {W}}\Big[\sum_{i\leq j} \big(S_iS_j - \mathbb{E}[X_iX_j\vert \tbf{W}]\big)^2\Big] 
\nonumber \\ 
&= \frac{1}{4n^2}\mathbb{E}_{\tbf{S}, \tbf {W}}\Big[ \bigl\| \tbf{S}\tbf{S}^{\intercal} - \mathbb{E}[\tbf{X}\tbf{X}^{\intercal}\vert\tbf{W}]\bigr\|_{\rm F}^2 \Big]
+\frac{1}{4n^2}\sum_{i=1}^n \mathbb{E}_{\tbf{S}, \tbf {W}}[(S_i^2 - \mathbb{E}[X_i^2 \vert \bW])^2]
\nonumber \\ 
&= \frac{1}{4} {\rm Mmmse}_n(\Delta^{-1}) + \mathcal{O}(1/n),
\end{align}
The proof details for first equality are in Appendix~\ref{app:immse}. The second equality is obatined by completing the sum and accounting for the diagonal terms. The last equality is obtained from 
\begin{align}
\mathbb{E}_{\tbf{S}, \tbf {W}}[(S_i^2 - \mathbb{E}[X_i^2 \vert \bW])^2]& 
= \mathbb{E}[S_i^4] - 2 \mathbb{E}_{S_i, \tbf {W}}[ S_i^2 \mathbb{E}[X_i^2 \vert \bW]] + 
\mathbb{E}_{\bW}[\mathbb{E}[X_i^2 \vert \bW]^2]
\nonumber \\ &
= 
 \mathbb{E}[S_i^4] - \mathbb{E}_{\bW}[\mathbb{E}[X_i^2 \vert \bW]^2]
 \nonumber \\ &
 \leq 
 \mathbb{E}[S_i^4].
\end{align}
where we have used the Nishimori identity 
$\mathbb{E}_{S_i, \tbf {W}}[ S_i^2 \mathbb{E}[X_i^2 \vert \bW]] = \mathbb{E}_{\bW}[\mathbb{E}[X_i^2 \vert \bW]^2]$ in the second equality (Appendix \ref{app:nishi}).
\end{proof}
\begin{lemma}\label{existence-thermo-limit}
The limit $\lim_{n\to +\infty}n^{-1}I(\tbf{S};\tbf{W})$ exists 
and is a concave, continuous, function of $\Delta$. 
\end{lemma}
\begin{proof}
The existence of the limit is the statement of Lemma \ref{lemma:superadditivity} in Sec. \ref{sec:subadditivitystyle}. The continuity follows from the concavity of the mutual information with respect to $\Delta^{-1}$: because the limit of a sequence of concave functions remains concave, and thus it is continuous. To see the concavity notice that the first derivative of the mutual information w.r.t $\Delta^{-1}$ equals the matrix-MMSE (Lemma \ref{lemma:immse}) and that the later cannot increase as a function of $\Delta^{-1}$.
\end{proof}

\subsection{Proof of Theorem \ref{thm1} for $0<\Delta  \leq \Delta_{\rm opt}$}\label{sec:delta-less-delta-opt}
\begin{lemma}\label{small-noise-lemma}
Assume $P_0$ is a discrete distribution. Fix $\Delta < \Delta_{\rm AMP}$. The mutual information per variable  is asymptotically given by the RS formula \eqref{rsform}.
\end{lemma}
\begin{proof}
By the suboptimality of the AMP algorithm we have
\begin{align}\label{startingpoint}
{\rm Mmse}_{n, \rm AMP}^{(t)}(\Delta^{-1}) \geq {\rm Mmmse}_n(\Delta^{-1}).
\end{align}
Taking limits in the order $\lim_{t\to +\infty}\limsup_{n\to +\infty}$ and using \eqref{eq:mseampmatriciel} we find
\begin{align}
v^2 - (v-E^{(\infty)})^2 \geq \limsup_{n\to +\infty}{\rm Mmmse}_n(\Delta^{-1}).
\end{align}
Furthermore, by applying Lemma \ref{lemma:immse} we obtain 
\begin{align}\label{thisineq}
\frac{v^2 - (v-E^{(\infty)})^2}{4} \geq \limsup_{n\to +\infty}\frac{1}{n}\frac{d I(\tbf{S}; \tbf{W}) }{d\Delta^{-1}}.
\end{align}
Now, for $\Delta < \Delta_{\rm AMP}$ we have $E^{(\infty)} = E_{\rm good}(\Delta)$ which is the unique and hence {\it global} minimum of $i_{\rm RS}(E; \Delta)$ over $E\in [0,v]$. Moreover, for $\Delta < \Delta_{\rm AMP}$ we have that $E^{(\infty)}(\Delta)$ is continuously differentiable $\Delta^{-1}$ with locally bounded derivative. Thus
\begin{align}\label{first_1}
\frac{d}{d\Delta^{-1}}\Big(\min_{E\in[0,v]}i_{\rm RS}(E; \Delta)\Big)
& =
\frac{di_{\rm RS}}{d\Delta^{-1}}(E^{(\infty)}; \Delta) 
\nonumber \\ &
= 
 \frac{\partial i_{\rm RS}}{\partial E}(E^{(\infty)}; \Delta) \frac{d E^{(\infty)}}{d\Delta^{-1}}
 + \frac{\partial i_{\rm RS}}{\partial\Delta^{-1}}(E^{(\infty)} ; \Delta)
 \nonumber \\ &
=  \frac{\partial i_{\rm RS}}{\partial\Delta^{-1}}(E^{(\infty)} ; \Delta) \nonumber \\
&= \frac{(v-E^{(\infty)})^2 +v^2}{4}
- \frac{\partial \mathbb{E}_{S, Z}[\cdots]}{\partial \Sigma^{-2}}\bigg\vert_{E^{(\infty)}}\frac{\partial\Sigma^{-2}}{\partial\Delta^{-1}}\bigg\vert_{E^{(\infty)}} 
\nonumber \\
&= \frac{v^2 - (v-E^{(\infty)})^2}{4},
\end{align}
where $\mathbb{E}_{S, Z}[\cdots]$ is the expectation that appears in the RS potential \eqref{eq:potentialfunction}. The third equality is obtained from 
\begin{align}
\frac{\partial\Sigma^{-2}}{\partial\Delta^{-1}}\bigg\vert_{E^{(\infty)}}  = v -E^{(\infty)}
\end{align}
and 
\begin{align}
\frac{\partial \mathbb{E}_{S, Z}[\cdots]}{\partial \Sigma^{-2}}\Big\vert_{E^{(\infty)}} = \frac{1}{2}(v-E^{(\infty)}).
\end{align}
This last identity immediately follows from
$\frac{\partial i_{\rm RS}}{\partial E}\Big\vert_{E^{(\infty)}} = 0$.
From \eqref{thisineq} and \eqref{first_1}
\begin{align}
\frac{d}{d\Delta^{-1}}(\min_{E\in [0,v]}i_{\rm RS}(E; \Delta)) \geq 
\limsup_{n\to +\infty}\frac{1}{n}\frac{d I(\tbf{S}; \tbf{W}) }{d\Delta^{-1}}, \nonumber\\
\end{align}
which is equivalent to 
\begin{align}
\frac{d}{d\Delta}(\min_{E\in [0,v]}i_{\rm RS}(E; \Delta)) \leq 
\liminf_{n\to +\infty}\frac{1}{n}\frac{d I(\tbf{S};\tbf{W}) }{d\Delta}. 
\label{middlepoint}
\end{align}
We now integrate inequality \eqref{middlepoint} over an interval 
$[0, \Delta] \subset [0, \Delta_{\rm AMP}[$
\begin{align}
\min_{E\in [0,v]}i_{\rm RS}(E; \Delta) - \min_{E\in [0,v]}i_{\rm RS}(E; 0) 
& \leq 
\int_0^\Delta d\tilde{\Delta}\,\liminf_{n\to +\infty}\frac{1}{n}\frac{d I(\tbf{S}; \tbf{W}) }{d\tilde{\Delta}}
\nonumber \\ &
\leq 
\liminf_{n\to +\infty}\int_0^\Delta d\tilde{\Delta}\,\frac{1}{n}\frac{d I(\tbf{S}; \tbf{W}) }{d\tilde{\Delta}}
\nonumber \\ &
= \liminf_{n\to +\infty}\frac{1}{n}I(\tbf{S}; \tbf{W}) - H(S).
\end{align}
The second inequality uses Fatou's Lemma and the last equality uses that for a discrete prior 
\begin{align}
\lim_{\Delta\to 0_+}I(\tbf{S}; \tbf{W}) = H(\tbf{S}) - \lim_{\Delta\to 0_+}H(\tbf{S}\vert \tbf{W}) =  nH(S).
\end{align}
In Appendix \ref{appendix_smallnoise} an explicit calculation shows that $\min_{E}i_{\rm RS}(E; 0) = H(S)$.  Therefore
\begin{align}\label{eq:lower_bound_small}
\min_{E\in [0,v]}i_{\rm RS}(E; \Delta) \leq \liminf_{n\to +\infty}\frac{1}{n}I(\tbf{S}; \tbf{W}).
\end{align}
The final step combines inequality \eqref{eq:lower_bound_small} with the replica bound \eqref{eq:guerrabound} to obtain
\begin{align}
\min_{E\in [0,v]}i_{\rm RS}(E; \Delta) & \leq 
\liminf_{n\to +\infty}\frac{1}{n}I(\tbf{S}; \tbf{W})
\leq 
\limsup_{n\to +\infty}\frac{1}{n}I(\tbf{S}; \tbf{W})
\leq 
\min_{E\in [0,v]}i_{\rm RS}(E; \Delta).
\end{align}
This shows that the limit of the mutual information exists and is equal to the RS formula for $\Delta <\Delta_{\rm AMP}$.
Note that in this proof we did not need the a-priori existence of the limit.
\end{proof}
%
%
\begin{remark}
One can try to apply the same proof idea to the regime $\Delta > \Delta_{\rm RS}$. 
Equations \eqref{startingpoint}-\eqref{middlepoint} work out exactly in the same way
because the AMP fixed point $E^{(\infty)}$ is a {\it global} minimum of $i_{\rm RS}(E; \Delta)$. Then when integrating on $]\Delta, +\infty[\subset [\Delta_{\rm RS}, +\infty[$, one finds 
\begin{align}
\limsup_{n\to +\infty}\frac{1}{n}I(\tbf{S}; \tbf{W}) \leq \min_{E\in [0, v]}i_{\rm RS}(E; \Delta).
\end{align}
This essentially gives an alternative proof of \eqref{eq:guerrabound} for $\Delta > \Delta_{\rm RS}$.
\end{remark}

\begin{lemma}\label{suboptimality}
We necessarily have $\Delta_{\rm AMP} \leq \Delta_{\rm opt}$.
\end{lemma}
\begin{proof}
Notice first that it not possible to have $\Delta_{\rm RS} < \Delta_{\rm AMP}$ because in the range $]0, \Delta_{\rm AMP}[$, as a function of $E$, the function $i_{\rm RS}(E;\Delta)$ has a unique stationary point. Since $\min_{E\in [0,v]} i_{\rm RS}(E; \Delta)$ is analytic
for $\Delta < \Delta_{\rm RS}$, it is analytic for $\Delta <\Delta_{\rm AMP}$. Now we proceed by contradiction: suppose we would have $\Delta_{\rm AMP} \geq \Delta_{\rm opt}$. Lemma \ref{small-noise-lemma} asserts 
that $\lim_{n\to +\infty} n^{-1} I(\tbf{S}; \bW) = \min_{E\in [0,v]} i_{\rm RS}(E; \Delta)$ for $\Delta < \Delta_{\rm AMP}$ thus we would have 
$\lim_{n\to +\infty} n^{-1} I(\tbf{S}; \bW)$ analytic at $\Delta_{\rm opt}$. This is a contradiction by definition of $\Delta_{\rm opt}$.
\end{proof}

\begin{lemma}\label{analyticity-argument-lemma}
We necessarily have $\Delta_{\rm RS} \geq \Delta_{\rm opt}$. 
\end{lemma}
\begin{proof}
If $\Delta_{\rm RS} =+\infty$ then we are done, so we suppose it is finite. The proof 
proceeds by contradiction: suppose $\Delta_{\rm RS} < \Delta_{\rm opt}$.  So we assume $\Delta_{\rm RS}\in [\Delta_{\rm AMP}, \Delta_{\rm opt}[$ (in the previous lemma we showed that this must be the case). For $\Delta \in \, ]0, \Delta_{\rm RS}[$ we have $\min_{E\in [0, v]}i_{\rm RS}(E; \Delta) = i_{\rm RS}(E_{\rm good}(\Delta); \Delta)$ which is an analytic function  in this interval. By definition of $\Delta_{\rm opt}$, the function 
$\lim_{n\to +\infty}\frac{1}{n}I(\tbf{S}; \tbf{W})$ is analytic in $]0, \Delta_{\rm opt}[$. Therefore {\it both functions are analytic on} $]0, \Delta_{\rm RS}[$ and since by Lemma~\ref{small-noise-lemma} they are equal for $]0, \Delta_{\rm AMP}[\subset ]0, \Delta_{\rm RS}[$, {\it they must 
be equal on the whole range} $]0, \Delta_{\rm RS}[$. This implies that the two functions are equal at $\Delta_{\rm RS}$ because they are continuous. Explicitly,
\begin{align}\label{eq:equalityatpoint}
\min_{E}i_{\rm RS}(E; \Delta) = \lim_{n\to +\infty}\frac{1}{n}I(\tbf{S}; \tbf{W})\vert_{\Delta} \qquad  \forall \ \Delta \in \, ]0, \Delta_{\rm RS}].
\end{align}
Now, fix some $\Delta \in \, ]\Delta_{\rm RS}, \Delta_{\rm opt}[$. Since this $\Delta$ is greater than $\Delta_{\rm RS}$ the fixed point of state evolution $E^{(\infty)}$ is also the global minimum of $i_{\rm RS}(E;\Delta)$. Hence 
exactly as in \eqref{startingpoint}-\eqref{middlepoint} we can show that for 
$\Delta \in \, ]\Delta_{\rm RS}, \Delta_{\rm opt}[$, \eqref{middlepoint} is verified.
This time, combining \eqref{eq:guerrabound}, \eqref{middlepoint} and the assumption $\Delta_{\rm RS}\in [\Delta_{\rm AMP}, \Delta_{\rm opt}[$, leads to a 
contradiction, and hence we must have $\Delta_{\rm RS} \geq \Delta_{\rm opt}$. To see explicitly how the contradiction appears, integrate \eqref{middlepoint} 
on $]\Delta_{\rm RS}, \Delta[\subset ]\Delta_{\rm RS}, \Delta_{\rm opt}[$, and use Fatou's Lemma, to obtain
\begin{align}
\min_{E}i_{\rm RS}(E; \Delta) - \min_{E}i_{\rm RS}(E; \Delta_{\rm RS})
\leq & \liminf_{n\to +\infty}\Big(\frac{1}{n}I(\tbf{S}; \tbf{W})\vert_{\Delta} - \frac{1}{n}I(\tbf{S}; \tbf{W})\vert_{\Delta_{\rm RS}}\Big)
\nonumber \\
=& 
\lim_{n\to +\infty}\frac{1}{n}I(\tbf{S}; \tbf{W})\vert_{\Delta} - 
\lim_{n\to +\infty}\frac{1}{n}I(\tbf{S}; \tbf{W})\vert_{\Delta_{\rm RS}}.
\end{align}
From \eqref{eq:equalityatpoint} and \eqref{eq:guerrabound} we obtain 
$\min_{E}i_{\rm RS}(E; \Delta) = \lim_{n\to +\infty}\frac{1}{n}I(\tbf{S}; \tbf{W})$ 
{\it when} $\Delta_{\rm AMP} \leq \Delta_{\rm RS} < \Delta < \Delta_{\rm opt}$. But from \eqref{eq:equalityatpoint}, this equality is also true for $0 < \Delta \leq \Delta_{\rm RS}$. So the equality is valid in the whole interval $]0, \Delta_{\rm opt}[$ and therefore $\min_{E}i_{\rm RS}(E; \Delta)$ is analytic at $\Delta_{\rm RS}$. But this is impossible by the definition of $\Delta_{\rm RS}$. 
\end{proof} 
\begin{proposition}\label{cor:delta-low-regime}
Assume $P_0$ is a discrete distribution. Fix $\Delta \leq \Delta_{\rm opt}$. The mutual information per variable  is asymptotically given by the RS formula \eqref{rsform}.
\end{proposition}
\begin{proof}
Lemma~\ref{small-noise-lemma} says that the two functions, $\lim_{n\to +\infty} n^{-1}I(\tbf{S}; \bW)$ and $\min_{E\in [0, v]} i_{\rm RS}(E; \Delta)$, are equal for $\Delta <\Delta_{\rm AMP}$ and Lemma~\ref{analyticity-argument-lemma} implies that {\it both} functions are analytic for $\Delta <\Delta_{\rm opt}$. Thus they must be equal on the whole range $\Delta < \Delta_{\rm opt}$. Since we also know they are continuous, then they are equal also at $\Delta =\Delta_{\rm opt}$.
\end{proof}
\subsection{Proof of Theorem \ref{thm1} for $\Delta \geq  \Delta_{\rm opt}$}\label{sec:everything-together}
We first need the following lemma where spatial coupling comes into the play.

\begin{lemma}\label{cor:equality-of-deltas}
The optimal threshold is given by the potential threshold: $\Delta_{\rm opt} = \Delta_{\rm RS}$.
\end{lemma}

\begin{proof}
It suffices to see that 
\begin{align}
\Delta_{\rm RS}\leq \Delta^{\rm c}_{\rm AMP}\leq \Delta^{\rm c}_{\rm opt} = \Delta_{\rm opt} \leq \Delta_{\rm RS}.
\end{align}
The first inequality is the threshold saturation result of Theorem \ref{thres-sat-lemma} in Section \ref{sec:thresh-sat}. The second inequality is due the suboptimality of the AMP algorithm.\footnote{More precisely one shows by the same methods Lemmas \ref{small-noise-lemma} and \ref{suboptimality} for the spatially coupled system.}
The equality is a consequence of Theorem \ref{LemmaGuerraSubadditivityStyle} in Section \ref{sec:subadditivitystyle}. Indeed, equality of asymptotic mutual informations of the coupled and underlying system implies that they must be non-analytic at the same value of $\Delta$. 
Finally, the last inequality is the statement of Lemma \ref{analyticity-argument-lemma} in Section \ref{sec:delta-less-delta-opt}.
\end{proof}

\begin{proposition}\label{for:above-delta-opt}
Assume $P_0$ is a discrete distribution. Fix $\Delta \geq \Delta_{\rm opt}$. The mutual information per variable  is asymptotically given by the RS formula \eqref{rsform}.
\end{proposition}

\begin{proof}
We already remarked in section \ref{sec:delta-less-delta-opt} that for $\Delta > \Delta_{\rm RS}$,
\begin{align}
\frac{d}{d\Delta}(\min_{E}i_{\rm RS}(E; \Delta)) \leq 
\liminf_{n\to +\infty}\frac{1}{n}\frac{d I(\tbf{S};\tbf{W}) }{d\Delta}. 
\end{align}
Now we integrate on an interval $]\Delta_{\rm RS}, \Delta]$ both sides of the inequality. 
Since from Lemma \ref{cor:equality-of-deltas} we have that $\Delta_{\rm RS} = \Delta_{\rm opt}$, it is equivalent to integrate from 
$\Delta_{\rm opt}$ upwards\footnote{This is the point we did not yet know in section \ref{sec:delta-less-delta-opt}.}
\begin{align}
\int_{\Delta_{\rm opt}}^{\Delta}d\tilde{\Delta}\, \frac{d}{d\tilde{\Delta}}(\min_{E}i_{\rm RS}(E; \tilde{\Delta})) \leq 
\int_{\Delta_{\rm opt}}^{\Delta} d\tilde{\Delta}\, \liminf_{n\to +\infty}\frac{1}{n}\frac{d I(\tbf{S};\tbf{W}) }{d\tilde{\Delta}}.
\end{align}
By Fatou's lemma the inequality is preserved if we bring the $\liminf$ outside of the integral, thus
\begin{align}
\min_{E}i_{\rm RS}(E; \Delta) - \min_{E}i_{\rm RS}(E; \Delta_{\rm RS}) & \leq \liminf_{n\to +\infty}\biggl\{
\frac{1}{n}I(\tbf{S};\tbf{W})\Big\vert_{\Delta} -  \frac{1}{n}I(\tbf{S};\tbf{W})\Big\vert_{\Delta_{\rm opt}}\biggr\}
\nonumber \\ &
= \lim_{n\to +\infty}\frac{1}{n}I(\tbf{S};\tbf{W})\Big\vert_{\Delta} -  
\lim_{n\to +\infty}\frac{1}{n}I(\tbf{S};\tbf{W})\Big\vert_{\Delta_{\rm opt}}.
\end{align}
To get the last line we have used the existence of the thermodynamic limit (see Lemma \ref{existence-thermo-limit}).
But we already know from Proposition \ref{cor:delta-low-regime} that 
$\min_{E}i_{\rm RS}(E; \Delta_{\rm opt}) = \lim_{n\to +\infty}\frac{1}{n}I(\tbf{S};\tbf{W})\Big\vert_{\Delta_{\rm opt}}$. Therefore
\begin{align}
\min_{E}i_{\rm RS}(E; \Delta) \leq \lim_{n\to +\infty}\frac{1}{n}I(\tbf{S};\tbf{W}),
\end{align}
which together with \eqref{eq:guerrabound} ends the proof.
\end{proof}

\section{Proof of corollaries \ref{cor:MMSE} and \ref{perfamp}} \label{sec:proof_coro}
In this section, we provide the proofs of Corollary \ref{cor:MMSE} and Corollary \ref{perfamp} concerning the MMSE formulae and the optimality of the AMP algorithm.
We first show the following result about the matrix and vector MMSE's in Definition \ref{def:mmse}.
\begin{lemma}\label{lemma:Mmmse}
Assume the prior $P_0$ has finite first four moments and recall the second moment is called $v$. The matrix and vector MMSE verify
\begin{align}
{\rm Mmmse}_n
\le \big( v^2-(v - {\rm Vmmse}_n)^2 \big) + \mathcal{O}(\frac{1}{n}).
\end{align}
\end{lemma}
\begin{proof}
For this proof we denote $\langle \cdot \rangle $ the expectation w.r.t the posterior distribution \eqref{eq:posterior_dist}. The matrix and vector MMSE then read
\begin{align}
{\rm Mmmse}_{n} &\defeq \frac{1}{n^2}\mathbb{E}_{\bS, \bW}\Big[ \bigl\| \bS\bS^{\intercal} - \langle\bX\bX^{\intercal}\rangle \bigr\|_{\rm F}^2 \Big], \label{eq:matrix_mmse}\\ 
{\rm Vmmse}_{n} &\defeq \frac{1}{n}\mathbb{E}_{\bS, \bW}\Big[ \bigl\| \bS - \langle\bX\rangle \bigr\|_{2}^2 \Big]. \label{eq:vector_mmse}
\end{align}
Expanding the Frobenius norm in (\ref{eq:matrix_mmse}) yields
\begin{align}
{\rm Mmmse}_{n} &= \frac{1}{n^2}\mathbb{E}_{\bS, \bW}\Big[ \sum_{i,j=1}^n ( S_i S_j - \langle X_i X_j\rangle)^2 \Big] =\frac{1}{n^2}\mathbb{E}_{\bS, \bW}\Big[ \sum_{i,j=1}^n  S_i^2 S_j^2 - \langle X_i X_j\rangle^2 \Big] \nonumber \\
&= \mathbb{E}_{\bS}\Big[\Big(\frac{1}{n}\sum_{i=1}^n S_i^2\Big)^2 \Big] - \frac{1}{n^2} \sum_{i,j=1}^n \mathbb{E}_{\bS, \bW}[\langle X_i X_j \rangle^2], 
\end{align}
where the second equality follows from $\mathbb{E}_{\bS, \bW}[\langle X_i X_j\rangle^2]=\mathbb{E}_{\bS, \bW}[ S_i S_j \langle X_i X_j\rangle]$, implied by the Nishimori 
identity \eqref{eq:nishCond}. Similarly, using $\mathbb{E}_{\bS, \bW} [\langle X_i \rangle^2] = \mathbb{E}_{\bS, \bW} [S_i \langle X_i\rangle]$ implied by the Nishimori identity, 
(\ref{eq:vector_mmse}) simplifies to
\begin{align}
{\rm Vmmse}_{n} = v - \frac{1}{n}\sum_{i=1}^n \mathbb{E}_{\bS, \bW}[\langle X_i\rangle^2 ].
\end{align}
Hence,
\begin{align}
{\rm Mmmse}_n
-  \big( v^2-(v - {\rm Vmmse}_n)^2 \big) = \mathcal{A}_{n} - \mathcal{B}_{n},
\end{align}
with
\begin{align}
\mathcal{A}_{n} &\defeq  \mathbb{E}_{\bS}\Big[\Big(\frac{1}{n}\sum_{i=1}^n S_i^2\Big)^2 \Big] - v^2, \\
\mathcal{B}_{n} &\defeq \frac{1}{n^2} \sum_{i,j=1}^n \Big(\mathbb{E}_{\bS, \bW}[ \langle X_i X_j \rangle^2 ] - \mathbb{E}_{\bS, \bW}[\langle X_i \rangle^2 ]\mathbb{E}_{\bS, \bW}[\langle X_j \rangle^2 ]\Big).\label{eq:Mmmse_Vmmse_diff_1}
\end{align}
Since the signal components $\{S_i\}$ are i.i.d and $P_0$ has finite first four moments, $\mathcal{A}_n = \mathcal{O}(1/n)$. 
It remains to show that $\mathcal{B}_n \geq 0$. This is most easily seen as follows. By defining the {\it overlap}
\begin{align}
 q(\bX,\bS) \defeq \frac{1}{n}\sum_{i=1}^n S_iX_i
\end{align}
and using the Nishimori identities
$\mathbb{E}_{\bS, \bW}[\langle X_i \rangle^2 ] = \mathbb{E}_{\bS, \bW}[S_i\langle X_i \rangle ]$ and $\mathbb{E}_{\bS, \bW}[\langle X_i X_j\rangle^2]=\mathbb{E}_{\bS, \bW}[ S_i S_j \langle X_i X_j\rangle]$, 
we observe that 
\begin{align}
\mathcal{B}_{n} & = \mathbb{E}_{\bS, \bW}[ \langle q^2\rangle ] -  \mathbb{E}_{\bS, \bW}[\langle q \rangle]^2 
\nonumber \\ &
=
\mathbb{E}_{\bS, \bW}[ (q - \mathbb{E}_{\bS, \bW}[\langle q \rangle])^2]
\end{align}
which is non-negative.
%
\end{proof}
\begin{remark}
Using ideas similar to [\cite{KoradaMacris}] to prove concentration of overlaps in inference problems suggest that Lemma \ref{lemma:Mmmse} holds with an equality
when suitable ``side observations'' are added. 
\end{remark}
\subsection{Proof of Corollary \ref{cor:MMSE}} \label{sec:proof_coro_1}
We first show how to prove the expression \eqref{coro1} for the asymptotic ${\rm Mmmse}_n$ by taking the limit $n\to +\infty$ on both sides of \eqref{eq:GSV}. 
First notice that since $n^{-1} I({\tbf S}; \bW)$ is a sequence of concave functions with respect to $\Delta^{-1}$, the limit when $n\to +\infty$ is also concave 
and differentiable for almost all $\Delta^{-1}$ and at all differentiability points we have (by a standard theorem of real analysis on convex functions)
\begin{align}\label{limderexchg}
\lim_{n\to + \infty} \frac{1}{n}\frac{d}{d\Delta^{-1}} I(\tbf{S}; \tbf{W})
= 
 \frac{1}{n}\frac{d}{d\Delta^{-1}}\lim_{n\to + \infty} I(\tbf{S}; \tbf{W})\,.
\end{align}
Thus from Lemma \ref{lemma:immse} and Theorem \ref{thm1} we have for all $\Delta \neq \Delta_{\rm RS}$
\begin{align}\label{mmseirs}
\lim_{n\to + \infty}{\rm Mmmse}_n(\Delta^{-1}) = 
4 \frac{d}{d\Delta^{-1}}\min_{E\in[0,v]}i_{\rm RS}(E; \Delta).
\end{align}
It remains to compute the right hand side. 
Let $E_0(\Delta)$ denote the (global) minimum of $i_{\rm RS}(E; \Delta)$. For $\Delta \neq \Delta_{\rm RS}$ this is a differentiable function of $\Delta$ with locally bounded derivative. 
Hence using a similar calculation to the one done in \eqref{first_1}, we obtain
\begin{align}\label{first}
\frac{d}{d\Delta^{-1}}\Big(\min_{E\in[0,v]}i_{\rm RS}(E; \Delta)\Big)
& =
\frac{di_{\rm RS}}{d\Delta^{-1}}(E_{0}; \Delta) 
\nonumber \\ &
= 
 \frac{\partial i_{\rm RS}}{\partial E}(E_{0}; \Delta) \frac{d E_{0}}{d\Delta^{-1}}
 + \frac{\partial i_{\rm RS}}{\partial\Delta^{-1}}(E_{0}; \Delta)
 \nonumber \\ &
=  \frac{\partial i_{\rm RS}}{\partial\Delta^{-1}}(E_{0} ; \Delta).
\end{align}
To compute the partial derivative with respect to 
$\Delta$ we first note that $\frac{\partial i_{\rm RS}}{\partial E}\Big\vert_{E_0} = 0$  implies
\begin{align}\label{second}
 0 = -\frac{v-E_0}{2\Delta} - \frac{\partial \mathbb{E}_{S, Z}[\cdots]}{\partial \Sigma^{-2}}\bigg\vert_{E_0} \frac{\partial\Sigma^{-2}}{\partial E}\bigg\vert_{E_0}, 
\end{align}
where $\mathbb{E}_{S, Z}[\cdots]$ is the expectation that appears 
in the RS potential \eqref{eq:potentialfunction}.
This immediately gives 
\begin{align}
 \frac{\partial \mathbb{E}_{S, Z}[\cdots]}{\partial \Sigma^{-2}}\Big\vert_{E_0} = \frac{1}{2}(v-E_0)
\end{align}
Thus
\begin{align}\label{third}
\frac{\partial i_{\rm RS}}{\partial\Delta^{-1}}(E_{0} ; \Delta) & = \frac{(v-E_0)^2 +v^2}{4}
- \frac{\partial \mathbb{E}_{S, Z}[\cdots]}{\partial \Sigma^{-2}}\bigg\vert_{E_0}\frac{\partial\Sigma^{-2}}{\partial\Delta^{-1}}\bigg\vert_{E_0} 
\nonumber \\ &
= \frac{v^2 - (v-E_0)^2}{4}. 
\end{align}
From \eqref{mmseirs}, \eqref{first}, \eqref{third} we obtain the desired result, formula \eqref{coro1}.

We now turn to the proof of \eqref{coro2} for the 
expression of the asymptotic vector-MMSE. 
From Lemma \ref{lemma:Mmmse} and the suboptimality of the AMP algorithm (here $E_0(\Delta)$ is the global minimum of $i_{\rm RS}(E;\Delta)$ and 
$E^{(\infty)}$ the fixed point of state evolution)
\begin{align}
v^2-(v - E_0)^2 & = \lim_{n\rightarrow\infty} {\rm Mmmse}_n
\nonumber \\ & 
\le \liminf_{n\rightarrow\infty} \big( v^2-(v - {\rm Vmmse}_n)^2 \big) 
\nonumber \\ &
\le 
\limsup_{n\rightarrow\infty} \big( v^2-(v - {\rm Vmmse}_n)^2 \big)
\nonumber \\ &
\le v^2-(v - E^{(\infty)})^2 ,
\end{align}
For $\Delta \notin [\Delta_{\rm AMP},\Delta_{\rm RS}]$, we have that $E_0 = E^{(\infty)}$ which ends the proof.   
\subsection{Proof of Corollary \ref{perfamp}} \label{sec:proof_coro_1}
In view of \eqref{eq:mseampmatriciel} we have
\begin{align}
\lim_{t\to +\infty}\lim_{n\to +\infty} {\rm Vmse}_{n, \rm AMP}^{(t)}(\Delta^{-1}) &= E^{(\infty)},
\label{eq:proofCorollaryAMP_1}
\\
\lim_{t\to +\infty}\lim_{n\to +\infty} {\rm Mmse}_{n, \rm AMP}^{(t)}(\Delta^{-1})
&= v^2 - (v - E^{(\infty)})^2.
\label{eq:proofCorollaryAMP_2}
\end{align}
For $\Delta \notin [\Delta_{\rm AMP},\Delta_{\rm RS}]$ we have $E^{(\infty)} = {\rm{argmin}}_{E\in [0, v]} i_{\rm RS}(E; \Delta)$ and also the two formulas of Corrolary \ref{cor:MMSE}
hold. This directly implies the formulas \eqref{ampcoro1} and \eqref{ampcoro2}.  

On the other hand for $\Delta \in [\Delta_{\rm AMP},\Delta_{\rm RS}]$ we have $E^{(\infty)} > {\rm{argmin}}_{E\in [0, v]} i_{\rm RS}(E; \Delta)$, 
so
using the monotonicity of $E^{(t)}$ leads to strict inequalities in \eqref{eq:proofCorollaryAMP_1} and \eqref{eq:proofCorollaryAMP_2} and thus to \eqref{ampsub1} and \eqref{ampsub2}.

\appendix
\section{Upper bound on the mutual information} \label{app:upperBound}
For the completeness of this work we revisit the proof of the upper bound (\ref{eq:guerrabound}) on
mutual information. This result was already obtained by [\cite{krzakala2016mutual}] using a Toninelli-Guerra type 
interpolation and is used in this paper so we only sketch the main steps. 

We consider the following interpolating inference problem 
\begin{equation}\label{eq:denoising}
\begin{cases}
w_{ij} = \frac{s_is_j}{\sqrt n} + \sqrt{\frac{\Delta}{t}} z_{ij},\\
y_i = s_i + \sqrt{\frac{\Delta}{m(1-t)}} z_i^{\prime},
\end{cases}
\end{equation}
with $m \defeq v-E \in [0,v]$, and $Z_i^{\prime} \sim \mathcal{N}(0,1)$. For $t=1$ we find back the original problem (\ref{eq:mainProblem}) since the $y_i$ observations become useless and for $t=0$ we have 
a set of decoupled observations from a Gaussian channel. 
The interpolating posterior distribution associated to this set of observations is 
\begin{align}\label{t-post}
P_t(\bx \vert \bs, \bz, \bz^\prime) \defeq \frac{e^{-\mathcal{H}(t)}\prod_{i=1}^n P_0(x_i)}{\int \big\{\prod_{i=1}^n dx_iP_0(x_i)\big\} e^{-\mathcal{H}(t)}} \defeq \frac{1}{\mathcal{Z}(t)}e^{-\mathcal{H}(t)}\prod_{i=1}^n P_0(x_i)
\end{align}
where
\begin{align*}
\mathcal{H}(t) & = \sum_{i\le j=1}^n \Big( \frac {t}{2\Delta n} x_i^2 x_j^2 - \frac{t}{\Delta n} x_i x_j s_i s_j  - \sqrt{ \frac{t}{n \Delta}}  x_i x_j   z_{ij}\Big)
\nonumber \\ &
\quad 
+ 
\sum_{i=1}^n \Big(\frac{ m(1-t)}{2\Delta} x_i^2 - 
\frac{m(1-t)}{\Delta}x_is_i + \sqrt{\frac{m(1-t)}{\Delta}}x_iz_i^\prime\Big) .
\end{align*}
can be interpreted as a ``Hamiltonian'' and the normalizing factor $\mathcal{Z}(t)$ is interpreted as a ``partition function''. We adopt the Gibbs ``bracket'' notation $\langle -\rangle_t$ for the expectation with respect to the posterior
\eqref{t-post}. The mutual information associated to interpolating inference problem is 
\begin{align}
 i(t) = - \frac{1}{n} \mathbb{E}_{\bf S, \bf Z, \bf Z^{\prime}}[\ln \mathcal{Z}(t)] + \frac{v^2}{4\Delta} + \frac{1}{4 \Delta n} (2\mathbb{E}[S^4] - v^2)
\end{align}
Note that on one hand $i(1) = \frac{1}{n}I(\tbf{S}; \tbf{W})$ the mutual information of the original matrix factorization problem and on the other hand 
\begin{align}\label{eq:freeEnergy_Denoising}
i_0 &= - \frac{1}{n}\mathbb{E}_{\bf S, \bf Z^{\prime}} \Big[\ln\Big(\int \big\{\prod_{i=1}^n dx_iP_0(x_i)\big\}
e^{-\frac{m\| \mathbf{x} \|_2^2}{2\Delta} 
+ \mathbf{x}^{\intercal}\bigl(\frac{m\mathbf{S}}{\Delta} + \sqrt{\frac{m}{\Delta}}\mathbf{Z}^{\prime}\bigr)}\Big)\Big]  + \frac{v^2}{4\Delta} + \frac{1}{4 \Delta n} (2\mathbb{E}[S^4] - v^2) \nonumber \\ 
&= - \mathbb{E}_{S, Z^{\prime}} \Big[\ln\Big(\int dx\,P_0(x)
e^{-\frac{mx^2}{2\Delta} 
+ x\bigl(\frac{mS}{\Delta} + \sqrt{\frac{m}{\Delta}}Z^{\prime}\bigr)}\Big)\Big] + \frac{v^2}{4\Delta} + \frac{1}{4 \Delta n} (2\mathbb{E}[S^4] - v^2)
\nonumber \\
& = i_{\rm RS}(E;\Delta) - \frac{m^2}{4\Delta} + \frac{1}{4 \Delta n} (2\mathbb{E}[S^4] - v^2),
\end{align}
From the fundamental theorem of calculus
\begin{equation} \label{eq:freeEnergy_originalVSdenoising_link}
i(1) - i(0) = - \frac{1}{n} \int_0^1 dt \frac{d}{dt} \mathbb{E}_{\bf S, \bf Z, \bf Z^{\prime}}[\ln \mathcal{Z}(t)]
\end{equation}
so we get
\begin{align}\label{weget}
\frac{1}{n}I(\tbf{S} ; \tbf{W}) = 
i_{\rm RS}(E;\Delta) -  \frac {m^2}{4\Delta} + \frac{1}{4 \Delta n} (2\mathbb{E}[S^4] - v^2) - \frac{1}{n} \int_0^1 dt \frac{d}{dt} \mathbb{E}_{\bf S, \bf Z, \bf Z^{\prime}}[\ln \mathcal{Z}(t)].
\end{align}
We proceed to the computation of the derivative under the integral over $t$.  Denoting by $\langle - \rangle_t$ the expectation with respect to the posterior \eqref{t-post}, we have 
\begin{align}
\frac{d}{dt}\mathbb{E}_{\bf S, \bf Z, \bf Z^{\prime}}[\ln \big(\mathcal{Z}(t)\big)] = 
\mathbb{E}_{\bf S, \bf Z, \bf Z^{\prime}}\big[-\big\langle \frac{d \mathcal{H}(t)}{d t}\big\rangle_{t}\big].
\end{align}
Hence, a simple differentiation of the Hamiltonian w.r.t. $t$ yields
\begin{align}\label{eq:derivative_freeEnergy_1}
\frac{d}{dt}\mathbb{E}_{\bf S, \bf Z, \bf Z^{\prime}}[\ln \big(\mathcal{Z}(t)\big)] = \mathbb{E}_{\bf S, \bf Z, \bf Z^{\prime}}\Big[&
 \sum_{i\le j=1}^n\Big(-\frac {\langle X_i^2 X_j^2 \rangle_{t} }{2\Delta n} + \frac{\langle X_i X_j \rangle_{t}S_i S_j }{\Delta n} + \frac {Z_{ij}\langle X_i X_j  \rangle_{t}}{2\sqrt{n \Delta t}} \Big) \nonumber \\
&+ \sum_{i=1}^n \Big( m \frac{\langle X_i^2 \rangle_{t}}{2\Delta} - m
\frac{\langle X_i\rangle_{t}S_i}{\Delta} - \frac{Z^{\prime}_i \langle X_i  \rangle_{t} }{2}  \sqrt{\frac{m}{\Delta (1-t)}}\Big) \Big] .
\end{align}
We now simplify this expression using integration by parts with respect to the Gaussian noises and the Nishimori identity \eqref{eq:nishCond} in Appendix \ref{app:nishi}. 
Integration by parts with respect to $Z_{ij}$ and $Z_i^\prime$ yields
\begin{align}\label{eq:DCT1}
\mathbb{E}_{\bf S, \bf Z, \bf Z^{\prime}}[Z_{ij} \langle X_i X_j  \rangle_{t}] 
= \mathbb{E}_{\bf S, \bf Z, \bf Z^{\prime}}[ \partial_{Z_{ij}} \langle X_i X_j  \rangle_{t}] 
= \sqrt{\frac{t}{n\Delta}}  \mathbb{E}_{\bf S, \bf Z, \bf Z^{\prime}}\Big[ \langle X_i^2 X_j^2  \rangle_{t} - \langle X_i X_j  \rangle_{t}^2 \Big] 
\end{align}
and
\begin{align}\label{eq:DCT2}
\mathbb{E}_{\bf S, \bf Z, \bf Z^{\prime}}[Z_{i}^{\prime} \langle X_i \rangle_{t}] = \mathbb{E}_{\bf S, \bf Z, \bf Z^{\prime}}[ \partial_{Z_{i}^{\prime}} \langle X_i \rangle_{t}] = \sqrt{\frac{m(1-t)}{\Delta}} \mathbb{E}_{\bf S, \bf Z, \bf Z^{\prime}}\Big[ \langle X_i^2 \rangle_{t} - \langle X_i  \rangle_{t}^2 \Big].
\end{align}
An application of the Nishimori identity yields 
\begin{align}\label{eq:nishimori_upperBound}
\mathbb{E}_{\bf S, \bf Z, \bf Z^{\prime}}[ \langle X_i X_j  \rangle_{t} S_i S_j ] = \mathbb{E}_{\bf S, \bf Z, \bf Z^{\prime}}[ \langle X_i X_j  \rangle_{t}^2 ]
\end{align}
and 
\begin{align}\label{eq:nshiother}
\mathbb{E}_{\bf S, \bf Z, \bf Z^{\prime}}[ \langle X_i  \rangle_{t} S_i] = \mathbb{E}_{\bf S, \bf Z, \bf Z^{\prime}}[ \langle X_i  \rangle_{t}^2 ]
\end{align}
Combining (\ref{eq:derivative_freeEnergy_1}) - (\ref{eq:nshiother}) we get
\begin{align*}
 \frac{1}{n}\frac{d}{dt} \mathbb{E}_{\bf S, \bf Z, \bf Z^{\prime}} [\ln \big(\mathcal{Z}(t)\big)] 
 & =\frac 1{2\Delta n^2} \sum_{i\le j=1}^n \mathbb{E}_{\bf S, \bf Z, \bf Z^{\prime}} [ \langle X_i X_j S_i S_j \rangle_{t}]  
 -\frac{m}{2\Delta n} \sum_{i=1}^n  \mathbb{E}_{\bf S, \bf Z, \bf Z^{\prime}} [\langle X_i S_i \rangle_{t} ]
 \nonumber \\ &
 = 
 \frac 1{4\Delta}\mathbb{E}_{\bf S, \bf Z, \bf Z^{\prime}} \big[\langle q(\tbf{S}, \bX)^2\rangle_t - 2\langle q(\tbf{S}, \bX)\rangle_t m\big] 
 + \frac{1}{4\Delta n^2}\sum_{i=1}^n \mathbb{E}_{\bf S, \bf Z, \bf Z^{\prime}}[\langle X_i^2\rangle_t S_i^2]
 \end{align*}
where we have introduced the  ``overlap'' $q(\tbf{S}, \bX) \defeq n^{-1}\sum_{i=1}^n S_i X_i$. Replacing this result in \eqref{weget} we obtain the remarkable sum rule
(recall $m \defeq v-E$)
\begin{align}
\frac{1}{n}I(\tbf{S} ; \tbf{W}) & = 
i_{\rm RS}(E;\Delta) -  \frac{1}{4\Delta}
\int_0^1 dt \mathbb{E}_{\bf S, \bf Z, \bf Z^{\prime}} \big[\langle (q(\tbf{S}, \bX) -  m)^2\rangle_t\big] 
\nonumber \\ &
\quad
- \frac{1}{4\Delta n^2}\sum_{i=1}^n \mathbb{E}_{\bf S, \bf Z, \bf Z^{\prime}}[\langle X_i^2\rangle_t S_i^2] + \frac{1}{4 \Delta n} (2\mathbb{E}[S^4] - v^2).
\end{align}
Thus for any $E\in [0,v]$ we have 
\begin{align}
\limsup_{n\to +\infty} \frac{1}{n}I(\tbf{S} ; \tbf{W}) \leq 
i_{\rm RS}(E;\Delta)
\end{align}
and \eqref{eq:guerrabound} follows by optimizing the right hand side over $E$. 
\section{Relating the mutual information to the free energy} \label{app:freeEnergy}
The mutual information between $\bS$ and $\bW$ is defined as $I(\bS; \bW)  = H (\bS)  - H (\bS \vert \bW)$ with
\begin{align}
H (\bS \vert \bW) = - \mathbb{E}_{\bS, \bW}\big[\ln P(\bS \vert \bW)\big] = - \mathbb{E}_{\bS, \bZ}\big[ \langle \ln P(\bX \vert \bS,\bZ) \rangle \big].
\end{align}
By substituting the posterior distribution in \eqref{gibbsi}, one obtains
\begin{align}
H (\bS \vert \bW) =  \mathbb{E}_{\bS, \bZ}[ \ln \mathcal{Z}] + \mathbb{E}_{\bS, \bZ}\big[ \langle \mathcal{H}(\bX\vert \bS, \bZ) \rangle \big] + H(\bS).
\end{align}
Furthermore, using the Gaussian integration by part as in \eqref{eq:DCT1} and the Nishimori identity \eqref{eq:nishimori_upperBound} yield
\begin{align}
\mathbb{E}_{\bS, \bZ}\big[ \langle \mathcal{H}(\bX\vert \bS, \bZ) \rangle \big] = -\frac{1}{2 \Delta n} \mathbb{E}_{\bS} \big[\!\sum_{i\leq j=1}^n (S_i^2 S_j^2) \,\big] = - \frac{1}{4\Delta} \big( v^2 (n-1) + 2 \mathbb{E}[S^4] \big). 
\end{align}
Hence, the normalized mutual information is given by \eqref{eq:energy_MI}, which we repeat here for better referencing
\begin{align}\label{eq:energy_MI_app}
\frac{1}{n} I(\tbf{S}; \bW) = -\frac{1}{n} \mathbb{E}_{\bS, \bZ}[ \ln \mathcal{Z}] + \frac{v^2}{4\Delta} + \frac{1}{4 \Delta n} (2\mathbb{E}[S^4] - v^2).
\end{align}

Alternatively, one can define the mutual information as  $I(\bS; \bW)  = H (\bW)  - H (\bW \vert \bS)$. For the AWGN, it is easy to show that
\begin{align}\label{eq:condEntropy_app}
H (\bW \vert \bS) = \frac{n(n+1)}{4} \ln (2 \pi \Delta e).
\end{align}
Furthermore, $H(\bW) = - \mathbb{E}_{\bW}\big[\ln P(\bW)\big]$ with
\begin{align}
P(\bw) = \int \big\{\prod_{i=1}^n dx_i P_0(x_i)\big\} P(\bw \vert \bx) &=
\frac{1}{(2 \pi \Delta)^{\frac{n(n+1)}{4}}}\int \big\{\prod_{i=1}^n dx_i P_0(x_i)\big\} e^{-\frac{1}{2\Delta}\sum_{i\leq j}\big(\frac{x_i x_j}{\sqrt n} - w_{ij}\big)^2} \nonumber \\
& = \frac{\tilde{\mathcal{Z}}}{(2 \pi \Delta)^{\frac{n(n+1)}{4}}},
\end{align}
where $\tilde{\mathcal{Z}}$ is the partition function with complete square \eqref{eq:posterior_dist}. Hence, $H(\bW)$ reads
\begin{align}\label{eq:Entropy_app}
H(\bW) &= \frac{n(n+1)}{4} \ln (2 \pi \Delta)- \mathbb{E}_{\bW}\big[\ln \tilde{\mathcal{Z}} \big] \nonumber \\
&= \frac{n(n+1)}{4} \ln (2 \pi \Delta)- \mathbb{E}_{\bS,\bZ}\big[\ln \mathcal{Z} \big] + \frac{1}{2\Delta} \mathbb{E}_{\bS,\bZ}\big[\sum_{i\leq j} (\frac{S_i S_j}{\sqrt{n}} + \sqrt{\Delta} Z_{ij})^2 \big],
\end{align}
with $\mathcal{Z}$ the simplified partition function obtained after expanding the square \eqref{gibbsi}. A straightforward calculation yields
\begin{align}\label{eq:Entropy_1_app}
\frac{1}{2\Delta} \mathbb{E}_{\bS,\bZ}\big[\sum_{i\leq j} (\frac{S_i S_j}{\sqrt{n}} + \sqrt{\Delta} Z_{ij})^2 \big] &= \frac{n(n+1)}{4} + \frac{1}{2 \Delta n} \mathbb{E}_{\bS} \big[\sum_{i\leq j} (S_i^2 S_j^2)\big] \nonumber \\
&=  \frac{n(n+1)}{4} \ln (e) + \frac{1}{4\Delta} \big( v^2 (n-1) + 2 \mathbb{E}[S^4] \big).
\end{align}
Finally, combining \eqref{eq:condEntropy_app}, \eqref{eq:Entropy_app} and \eqref{eq:Entropy_1_app} yields the same identity \eqref{eq:energy_MI_app}.
\section{Proof of the I-MMSE relation}\label{app:immse}
For completeness, we give a detailed proof for the I-MMSE relation of Lemma \ref{lemma:immse} following the lines of [\cite{GuoShamaiVerdu}]. In the calculations below differentiations, expectations  and integrations commute (see Lemma~8 in [\cite{GuoShamaiVerdu}]). All the matrices are symmetric and $Z_{ij}\sim \mathcal{N}(0,1)$ for $i\leq j$. 

Instead of  \eqref{eq:mainProblem} it is convenient to work with the equivalent model $w_{ij} = \frac{s_i s_j}{\sqrt{n\Delta}} + z_{ij}$ and set $s_is_j =u_{ij}$. In fact, all subsequent calculations do not depend on the rank of the matrix ${\bf u}$ and are valid for any finite rank matrix estimation problem as long as the noise is Gaussian. 
The mutual information is $I({\bf S};{\bf W}) = H({\bf W}) - H({\bf W} | {\bf S})$ and 
$H({\bf W}|{\bf S}) = \frac{n(n+1)}{2}\ln(\sqrt{2\pi e})$. Thus 
\begin{align}\label{derivativeIM}
\frac{1}{n}\frac{d I({\bf S};{\bf W})}{d\Delta^{-1}} = \frac{1}{n}\frac{d H({\bf W})}{d\Delta^{-1}}.
\end{align}
We have $H({\bf W}) = -\mathbb{E}_{{\bf W}} [\ln P({\bf W})]$ where
\begin{align}
P({\bf w}) &= \mathbb{E}_{{\bf U}} [P({\bf w} | {\bf U})] =  \mathbb{E}_{{\bf U}} \Big[(2\pi)^{-\frac{n(n+1)}{4}} e^{-\frac{1}{2}\sum_{i\le j} \big(\frac{{U}_{ij}}{\sqrt{n\Delta}} - w_{ij}\big)^2}\Big].
\label{eq:defPwgivenU}
\end{align}
Differentiating w.r.t $\Delta^{-1}$
\begin{align}
\frac{dH({\bf W})}{d\Delta^{-1}}= -\mathbb{E}_{\bf U}\Big[\int d{\bf w} (1+\ln P({\bf w})) \frac{dP({\bf w}|{\bf U})}{d\Delta^{-1}}\Big]\label{eq:derH_1}
\end{align}
%
%
and 
\begin{align}
\frac{dP({\bf w}|{\bf U})}{d\Delta^{-1}}&=\sqrt{\frac{\Delta}{4n}} \sum_{k\le l} U_{kl}\Big(w_{kl} - \frac{U_{kl}}{\sqrt{\Delta n}}\Big) e^{-\frac{1}{2} \sum_{i\le j} \big(\frac{U_{ij}}{\sqrt{n\Delta}} - w_{ij}\big)^2} (2\pi)^{-\frac{n(n+1)}{2}}\\
&= -\sqrt{\frac{\Delta}{4n}} \sum_{k\le l}U_{kl}\frac{d}{dw_{kl}}e^{-\frac{1}{2} \sum_{i\le j} \big(\frac{U_{ij}}{\sqrt{n\Delta}} - w_{ij}\big)^2}(2\pi)^{-\frac{n(n+1)}{2}}.
\end{align}
Replacing this last expression in \eqref{eq:derH_1}, using an integration by part w.r.t $w_{kl}$ (the boundary terms can be shown to vanish), then Bayes formula, and finally \eqref{eq:defPwgivenU}, one obtains
\begin{align}
\sqrt{\frac{4n}{\Delta}}\frac{dH({\bf W})}{d\Delta^{-1}}&= \sum_{k\le l} \mathbb{E}_{\bf U}\Big[U_{kl} \int d{\bf w}(1+\ln P({\bf w}))  \frac{d}{dw_{kl}}e^{-\frac{1}{2} \sum_{i\le j}\big(\frac{U_{ij}}{\sqrt{n\Delta}} - w_{ij}\big)^2}(2\pi)^{-\frac{n(n+1)}{2}}\Big] \nonumber \\
&=-\sum_{k\le l} \int d{\bf w} \, \mathbb{E}_{\bf U}\Big[U_{kl} \frac{P({\bf w}|{\bf U})}{P({\bf w})} \Big]  \frac{dP({\bf w})}{dw_{kl}} \nonumber \\
&= -\sum_{k\le l} \int d{\bf w} \, \mathbb{E}_{{\bf U}|{\bf w}}\big[U_{kl}\big] \mathbb{E}_{\bf U} \Big[\frac{dP({\bf w}|{\bf U})}{dw_{kl}} \Big] \nonumber \\
&=\sum_{k\le l} \int d{\bf w} \, \mathbb{E}_{{\bf U}|{\bf w}}\big[U_{kl}\big] \mathbb{E}_{\bf U} \Big[\Big(w_{kl}-\frac{U_{kl}}{\sqrt{n\Delta}}\Big) P({\bf w}|{\bf U})\Big] \nonumber\\
&=\mathbb{E}_{\bf W}\Big[\sum_{k\le l} \mathbb{E}_{{\bf U}|{\bf W}}\big[U_{kl}\big]\Big({W}_{kl}- \frac{1}{\sqrt{n\Delta}}\mathbb{E}_{{\bf U}|{\bf W}} \big[U_{kl}\big]\Big)\Big].
\end{align}
Now we replace ${\bf w} = \frac{{\bf u}^0}{\sqrt{n\Delta}} + {\bf z}$, where ${\bf u}^0$ is an independent copy of ${\bf u}$. 
We denote $\mathbb{E}_{\bf W}[\cdot] = \mathbb{E}_{{\bf U}^0,{\bf Z}}[\cdot]$ the joint expectation. The last result then reads
\begin{align}\label{thelastresult}
\frac{dH({\bf W})}{d\Delta^{-1}}&=\frac{1}{2n}\mathbb{E}_{\bf W}\Big[\sum_{k\le l} \mathbb{E}_{{\bf U}|{\bf W}}\big[U_{kl}\big] \Big({U}^0_{kl}- \mathbb{E}_{{\bf U}|{\bf W}} \big[U_{kl}\big] + {Z}_{kl}\sqrt{n\Delta}\Big)\Big].
\end{align}
Now note the two Nishimori identities (see Appendix~\ref{app:nishi})
\begin{align}
\mathbb{E}_{\bf W}\Big[\mathbb{E}_{{\bf U}|{\bf W}}\big[U_{kl}\big]{U}^0_{kl}\Big] &= \mathbb{E}_{\bf W}\Big[\mathbb{E}_{{\bf U}|{\bf W}}\big[U_{kl}\big]^2\Big],\\
\mathbb{E}_{\bf W}\Big[({U}^0_{kl})^2\Big] &= \mathbb{E}_{\bf W}\Big[\mathbb{E}_{{\bf U}|{\bf W}}\big[U_{kl}^2\big] \Big],
\end{align}
and the following one obtained by a Gaussian integration by parts
\begin{align}\label{following}
\sqrt{n\Delta} \, \mathbb{E}_{\bf W}\Big[\mathbb{E}_{{\bf U}|{\bf W}}\big[U_{kl}\big]{Z}_{kl}\Big] &= \mathbb{E}_{\bf W}\Big[\mathbb{E}_{{\bf U}|{\bf W}}\big[U_{kl}^2\big]- \mathbb{E}_{{\bf U}|{\bf W}}\big[U_{kl}\big]^2\Big].
\end{align}
Using the last three identities, equation \eqref{thelastresult} becomes 
\begin{align}
\frac{1}{n}\frac{dH({\bf W})}{d\Delta^{-1}} & =  \frac{1}{2n^2} \mathbb{E}_{\bf W}\Big[\sum_{k\le l}\mathbb{E}_{{\bf U}|{\bf W}}\big[U_{kl}^2\big]- \mathbb{E}_{{\bf U}|{\bf W}}\big[U_{kl}\big]^2\Big] 
\nonumber \\ &
=\frac{1}{2n^2}\mathbb{E}_{\bf W}\Big[\sum_{k\le l} \Big(U_{kl}^0 - \mathbb{E}_{{\bf U}|{\bf W}}\big[U_{kl}\big]\Big)^2 \Big],
\end{align}
which, in view of \eqref{derivativeIM}, ends the proof.
\section{Nishimori identity} \label{app:nishi}
Take a random vector ${\tbf S}$ distributed according to some known prior $P_0^{\otimes{n}}$ and an observation $\bW$ is drawn from some known 
conditional distribution $P_{\bW\vert {\tbf S}}(\bw|\bs)$. 
Take $\bX$ drawn from a posterior distribution (for example this may be \eqref{eq:posterior_dist})
$$
P(\bx|\bw) = \frac{P_0^{\otimes{n}}(\bx)P_{\bW\vert {\tbf S}}(\bw|\bx)}{P(\bw)}.
$$
Then for any (integrable) function $g(\bs, \bx)$ the Bayes formula implies 
\begin{align}
\mathbb{E}_{\bS}\mathbb{E}_{\bW|\bS}\mathbb{E}_{\bX|\bW} [g(\bS, \bX)]
= \mathbb{E}_{\bW}  \mathbb{E}_{\bX'|\bW}\mathbb{E}_{\bX|\bW} [g(\bX', \bX)]
\end{align}
where $\bX, \bX'$ are independent random vectors distributed according to the posterior distribution. Therefore
\begin{align}
\mathbb{E}_{\bS,\bW}\mathbb{E}_{\bX|\bW} [g(\bS, \bX)] = \mathbb{E}_{\bW}  \mathbb{E}_{\bX'|\bW}\mathbb{E}_{\bX|\bW} [g(\bX', \bX)].\label{eq:nishCond}
\end{align}

In the statistical mechanics literature this identity is sometimes called the Nishimori identity and we adopt this language here. For model \eqref{eq:mainProblem} for example
we can express $\bW$ in the posterior in terms of $\bS$ and $\bZ$ which are independent and  $\mathbb{E}_{\bX|\bW}[ - ] = \langle -\rangle$. Then the Nishimori identity  reads 
\begin{align}\label{statform}
 \mathbb{E}_{\bS,\bZ}[\langle g(\bS, \bX)\rangle] = \mathbb{E}_{\bS, \bZ}[  \langle g(\bX', \bX)\rangle ].
\end{align}
An important case for $g$ depending only on the first argument is 
$\mathbb{E}_{\bS}[g(\bS)] = \mathbb{E}_{\bS, \bZ}[  \langle g(\bX)\rangle ]$.

Special cases that are often used in this paper are
\begin{align}
\begin{cases}
\mathbb{E}_{\bS,\bZ}[S_i \langle X_i\rangle] = \mathbb{E}_{\bS,\bZ}[\langle X_i\rangle^2]\\
\mathbb{E}_{\bS,\bZ}[S_iS_j \langle X_iX_j\rangle] = \mathbb{E}_{\bS,\bZ}[\langle X_iX_j\rangle^2]\\
\mathbb{E}[S^2] = \mathbb{E}_{\bS, \bZ}[\langle X_i^2\rangle]. 
\end{cases}
\end{align}
A mild generalization of \eqref{statform} which is also used is
\begin{align}
 \mathbb{E}_{\bS,\bZ}[S_i S_j\langle X_i\rangle\langle X_j\rangle] = \mathbb{E}_{\bS,\bZ}[\langle X_iX_j\rangle \langle X_i\rangle \langle X_j\rangle].
\end{align}
We remark that these identities are used with brackets $\langle - \rangle$ corresponding to various ``interpolating'' posteriors.

\section{Proof of Lemmas~\ref{lemma:fixedpointSE_extPot_underlying} and \ref{lemma:fixedpointSE_extPot}} \label{app:Lemma14}
We show the details for Lemma \ref{lemma:fixedpointSE_extPot_underlying}. The proof of Lemma \ref{lemma:fixedpointSE_extPot}
follows the same lines. 
A straightforward differentiation of $f_{\rm RS}^{\rm u}$ w.r.t. $E$ gives  
\begin{align}
\frac{df^{\rm u}_{\rm {RS}}(E;\Delta)}{dE} = \frac{E-v}{2\Delta} + \frac{1}{2\Delta} \mathbb{E}_{{Z}, {S}}\Big[\Big\langle-{X}^2 + 2 {X}{S} + ZX \sqrt{\frac{\Delta}{v-E}} \Big\rangle \Big]. 
\label{eq:app_fixedPointF}
\end{align}
Recall that here the posterior expectation $\langle\cdot\rangle$ is defined by \eqref{eq:defBracket_posteriorMean}. A direct application of the Nishimori condition gives
\begin{align}
v \defeq \mathbb{E}_{S}[S^2] &= \mathbb{E}_{S, Z} [\langle X^2 \rangle],\label{eq:app_nish1} \\
\mathbb{E}_{S, Z} [S \langle X\rangle] &= \mathbb{E}_{S, Z} [\langle X \rangle^2],\label{eq:app_nish2}
\end{align}
which implies 
\begin{align}
\mathbb{E}_{S, Z} [(S - \langle X \rangle_E)^2] &= \mathbb{E}_{S, Z} [\langle X^2 \rangle] - \mathbb{E}_{S, Z} [\langle X \rangle^2]. \label{eq:app_EisV}
\end{align}
Thus from \eqref{eq:app_fixedPointF} we see that stationary points of $f_{\rm RS}^{\rm u}$ satisfy
\begin{align}
E = 2v - 2\mathbb{E}_{S, Z} [\langle X \rangle^2] - \sqrt{\frac{\Delta}{v-E}} \mathbb{E}_{S, Z} [Z\langle X \rangle]. \label{eq:app_EisV_1}
\end{align}
Now using an integration by part w.r.t $Z$, one gets 
\begin{align}
\sqrt{\frac{\Delta}{v-E}} \mathbb{E}_{S, Z} [Z\langle X \rangle ] = v - \mathbb{E}_{S, Z} [\langle X\rangle^2],
\end{align}
which allows to rewrite \eqref{eq:app_EisV_1} as
\begin{align}
E  = v- \mathbb{E}_{S, Z} [\langle X\rangle^2] 
= \mathbb{E}_{S, Z} [(S - \langle X \rangle)^2]
\end{align}
where the second equality follows from \eqref{eq:app_nish1} and \eqref{eq:app_EisV}. Recalling the expression \eqref{def:SEop_underlying} of the state evolution operator we recognize the equation $E= T_{\rm u}(E)$.
\section{Analysis of $i_{\rm RS}(E; \Delta)$ for $\Delta \to 0$} \label{appendix_smallnoise}
In this appendix, we prove that $\lim_{\Delta\to 0} \min_{E}i_{\rm RS}(E; \Delta) = H(S)$. First, a simple calculation leads to the following relation between $i_{\rm RS}$ and the mutual information of the scalar denoising problem for $E\!\in\! [0, v]$
\begin{align}\label{eq:app_RS_denoising}
i_{\rm RS}(E; \Delta) = I\big(S;S+\Sigma(E)Z\big) + \frac{E^2}{4\Delta}, 
\end{align}
where $Z \sim \mathcal{N}(0,1)$ and $\Sigma(E)^2 \defeq \Delta/ (v  -  E)$. Note that as $\Delta \rightarrow 0$, $\Sigma(E)\rightarrow 0$ (for $E \neq v$). Therefore, $\lim_{\Delta\to 0} I\big(S;S+\Sigma(E)Z\big) = H(S)$. Now let $E_0$ be the global minimum of $i_{\rm RS}(E; \Delta)$. By evaluating both sides of \eqref{eq:app_RS_denoising} at $E_0$ and taking the limit $\Delta \rightarrow 0$, it remains to show that $E_0^2/(4\Delta) \rightarrow 0$ as $\Delta \rightarrow 0$ (i.e. $E_0^2 \rightarrow 0$ faster than $\Delta$). Since $E_0$ is the global minimum of the RS potential, then $E_0 = T_{{\rm u}}(E_0) = {\rm mmse}(\Sigma(E_0)^{-2})$ by Lemma \ref{lemma:fixedpointSE_extPot_underlying}. Moreover, one can show, under our assumptions on $P_0$, that the scalar MMSE function scales as
\begin{align}
{\rm mmse}(\Sigma^{-2}) = \mathcal{O}(e^{-c\Sigma^{-2}}),
\end{align}
with $c$ a non-negative constant that depends on $P_0$ [\cite{BMDK_trans2017}]. Hence, $E_0^2/(4\Delta) \rightarrow 0$ as $\Delta \rightarrow 0$, which ends the proof.


%
\section{Proof of Proposition \ref{lemma:openVSclosed}}\label{appendix-pinfree}
%

Call $\mathcal{H}^{\rm per}$ and $\langle-\rangle_{\rm per}$ the Hamiltonian and posterior average associated to the periodic SC system with mutual information $i_{n,w,L}^{\rm per}$. 
Similarly call $\mathcal{H}^{\tilde{\rm cou}}$ and $\langle-\rangle_{\rm cou}$ the Hamiltonian and posterior average associated to the pinned SC system with mutual information $i_{n,w,L}^{\rm cou}$.
The Hamiltonians 
satisfy the identity $\mathcal{H}^{\rm cou} - \mathcal{H}^{\rm per} = \delta \mathcal{H}$ with
\begin{align}
&\delta \mathcal{H}=  \sum_{\mu \in \mathcal{B}} \Lambda_{\mu,\mu} \sum_{i_\mu\le j_\mu}\bigg[ \frac{x_{i_\mu}^2x_{j_\mu}^2 + s_{i_\mu}^2s_{j_\mu}^2}{2n\Delta} -  \frac{s_{i_\mu}s_{j_\mu}x_{i_\mu}x_{j_\mu}}{n\Delta} - \frac{(x_{i_\mu}x_{j_\mu}-s_{i_\mu}s_{j_\mu})z_{i_\mu j_\mu}}{\sqrt{n \Delta \Lambda_{\mu,\mu}}}\bigg]\nonumber\\
+ &\sum_{\mu \in \mathcal{B}}\sum_{\nu\in \{\mu+1:\mu+w\}\cap \mathcal{B}}\Lambda_{\mu,\nu}\sum_{i_\mu\le j_\nu}\bigg[ \frac{x_{i_\mu}^2x_{j_\nu}^2 + s_{i_\mu}^2s_{j_\nu}^2}{2n\Delta} -  \frac{s_{i_\mu}s_{j_\nu}x_{i_\mu}x_{j_\nu}}{n \Delta} - \frac{(x_{i_\mu}x_{j_\nu}-s_{i_\mu}s_{j_\nu})z_{i_\mu j_\nu} }{\sqrt{n \Delta \Lambda_{\mu,\nu}}}\bigg]\nonumber\\
+ &\sum_{\mu \in \mathcal{B}}\sum_{\nu\in \{\mu-w:\mu-1\}\cap \mathcal{B}}\Lambda_{\mu,\nu}\sum_{i_\mu>j_\nu }\bigg[ \frac{x_{i_\mu}^2x_{j_\nu}^2 + s_{i_\mu}^2s_{j_\nu}^2}{2n\Delta} -  \frac{s_{i_\mu}s_{j_\nu}x_{i_\mu}x_{j_\nu}}{n\Delta} - \frac{(x_{i_\mu}x_{j_\nu}-s_{i_\mu}s_{j_\mu})z_{i_\mu j_\nu} }{\sqrt{n\Delta\Lambda_{\mu,\nu}}}\bigg] \nonumber. 
\end{align}
%
%
It is easy to see that 
\begin{align}\label{eq:energy_openVSclosed}
 i_{n, w, L}^{\rm per} - i_{n, w, L}^{\rm cou} &= \frac{1}{n(L+1)}\mathbb{E}_{{\bf S}, {\bf Z}}[\ln  \langle e^{ -\delta \mathcal{H}} \rangle_{\rm cou} ],\\
  i_{n, w, L}^{\rm cou} - i_{n, w,L}^{\rm per}& = \frac{1}{n(L+1)}\mathbb{E}_{{\bf S}, {\bf Z}}[\ln  \langle e^{ \delta \mathcal{H}} \rangle_{\rm per} ].
\end{align}
and using the convexity of the exponential, we get
\begin{equation}\label{eq:openVSclosed_sandwich}
i_{n, w, L}^{\rm cou } + \frac{\mathbb{E}_{{\bf S}, {\bf Z}}\big[ \langle \delta \mathcal{H} \rangle_{\rm per} \big]}{n(L+1)} 
\le i_{n, w, L}^{\rm per} 
\le i_{n, w, L}^{\rm cou} + \frac{\mathbb{E}_{{\bf S}, {\bf Z}}\big[ \langle \delta \mathcal{H} \rangle_{\rm cou} \big]}{n(L+1)}.
\end{equation}
Due to the pinning condition we have
$\mathbb{E}_{{\bf S}, {\bf Z}}[ \langle \delta \mathcal{H}(\bX) \rangle_{\rm cou} ] = 0$, and thus we get the upper bound $i_{n, w, L}^{\rm per}
\le i_{n, w, L}^{\rm cou}$. Let us now look at the lower bound. We note that by the Nishimori identity in  Appendix~\ref{app:nishi}, as long as $P_0$ has finite first four moments,
we can find 
constants $K_1, K_2$ independent of $n, w,L$ such that
$\mathbb{E}_{{\bf S}, {\bf Z}} [\langle X_{i_\mu}^2X_{j_\nu}^2 \rangle_{\rm per}]\leq K_1$ 
and $\mathbb{E}_{{\bf S}, {\bf Z}}[\langle X_{i_\mu}^4\rangle_{\rm per}]\leq K_2$. First we use Gaussian integration by parts to eliminate $z_{i_\mu j_\nu}$ from the brackets,
the Cauchy-Schwartz inequality, and the Nishimori identity of Appendix~\ref{app:nishi}, to get an upper bound where only fourth order moments of signal are involved. Thus as long as $P_0$ has finite first four moments we find
\begin{align}\label{eq:openVSclosed_order}
\frac{|\mathbb{E}_{{\bf S}, {\bf Z}}[ \langle \delta \mathcal{H}(\bX) \rangle_{\tilde{\rm c}} ]|}{ n(L+1)} &\le C \frac{\Lambda^*(2w+1)^2}{L+1} = \mathcal{O}(\frac{w}{L}).
\end{align}
for some constant $C > 0$ independent of $n, w, L$ and we recall $\Lambda_* \defeq \sup\Lambda_{\mu\nu} = \mathcal{O}(w^{-1})$.
Thus we get the lower bound $i_{n, w, L}^{\rm cou } - \mathcal{O}(\frac{w}{L}) \le i_{n, w, L}^{\rm per}$. This completes the proof of Proposition \ref{lemma:openVSclosed}.

\acks{The work of Jean Barbier and Mohamad Dia was supported by the Swiss National Foundation for Science grant number 200021-156672. We thank
Thibault Lesieur for help with the phase diagrams.}
\vskip 0.2in
\bibliography{sample}
\end{document}